\documentclass[11pt]{article}

\usepackage{authblk,amssymb,fullpage,latexsym,amsfonts,amsmath,amsthm}
\usepackage{color}
\usepackage{enumerate}
\usepackage{algorithm}
\usepackage{algpseudocode}
\usepackage{bm}
\usepackage{subfigure} 
\usepackage{mathtools}
\usepackage[titletoc,title]{appendix}
\usepackage[T1]{fontenc}
\usepackage{url}
\usepackage{float}
\usepackage{placeins}
\usepackage[labelfont=bf]{caption}
\usepackage[numbers,sort&compress]{natbib}
\usepackage{dsfont}
\usepackage{booktabs}
\usepackage{enumitem}
\usepackage{endnotes}
\usepackage{bbm}
\usepackage{comment}
\usepackage{xtab}
\usepackage{array}
\usepackage{arydshln}
\usepackage[textsize=tiny]{todonotes}
\usepackage{graphicx}
\usepackage{xcolor}
\usepackage{flafter}
\usepackage{multirow}
\usepackage{tikz}
\usepackage{bibunits}
\usepackage{hyperref}
\usepackage{cleveref}
\usepackage{soul}
\hypersetup{
colorlinks = true,
linkcolor = blue, %Colour of internal links
filecolor = magenta,
urlcolor = blue, %Colour for external hyperlinks
citecolor = blue %Colour of citations
}

\allowdisplaybreaks
\theoremstyle{plain}
\newtheorem{proposition}{Proposition}
\newtheorem{theorem}{Theorem}
\newtheorem{lemma}{Lemma}
\newtheorem{corollary}{Corollary}

\theoremstyle{definition}
\newtheorem{definition}{Definition}
\newtheorem{example}{Example}

\newcommand{\mP}{\mathbb{P}}
\newcommand{\mE}{\mathbb{E}}
\newcommand{\Var}{\mathrm{Var}}

\newcommand{\iid}{ \overset{\mathrm{i.i.d.}}{\sim}}

\newcommand{\uSB}{\hat{u}_\alpha^{\mathrm{SB}}}
\newcommand{\uSBmin}{\hat{u}_{\alpha,\mathrm{min}}^{\mathrm{SB}}}
\newcommand{\uTB}{\hat{u}_\alpha^{\mathrm{TB}}}
\newcommand{\ple}{p_{\scriptscriptstyle \le}}

\newcommand{\parheading}[1]{\par\smallskip\noindent\textbf{#1}\enspace\ignorespaces}

\defaultbibliography{reference}

\begin{document}

\begin{center}
\vspace*{0.2in}
{\LARGE Aggregation of Statistical Evidence under Exchangeability}
\vspace{0.3in}

{\large
Antonin Schrab$^{1}$ \quad
Rajen Shah$^{2}$ \quad
Arthur Gretton$^{3}$ \quad
Ilmun Kim$^{4,*}$
}

\vspace{0.15in}

\begin{tabular}{c}
$^1$Department of Computer Science and Technology, University of Cambridge, UK\\
$^2$Statistical Laboratory, University of Cambridge, UK\\
$^3$Gatsby Computational Neuroscience Unit, University College London, UK\\
$^4$Department of Mathematical Sciences, KAIST, South Korea
\end{tabular}

\def\thefootnote{*}
\footnotetext{Corresponding author: \texttt{ilmunk@kaist.ac.kr}.}
\def\thefootnote{\arabic{footnote}}

\vspace{0.2in}

\today
\end{center}

\begin{abstract} 
We study aggregation of statistical evidence under unknown and potentially complex dependence using group-invariance. Building on permutation-based constructions that treat transformed datasets as exchangeable units, we aggregate evidence across statistics for each transformed dataset and calibrate the resulting aggregates across transformations. We develop a finite-sample power and adaptivity theory for this framework, together with extensions to sequential and data-dependent aggregation that preserve validity. For single-batch aggregation, which uses one collection of transformed datasets for both standardization and calibration, we show that the critical values uniformly improve on deterministic calibrations valid under arbitrary dependence, including Bonferroni correction, while adapting to the unknown dependence structure. We also introduce a sequential alpha-spending version that permits early rejection when evidence is strong, and a two-batch extension that separates standardization from calibration to accommodate learned aggregation rules and reduce computation. Applications to adaptive nonparametric testing and conformal prediction illustrate how these results sharpen existing aggregation methods.
\medskip

\end{abstract}

\section{Introduction}

Modern nonparametric testing and distribution-free inference routinely generate collections of test statistics or p-values whose dependence structure is unknown and often highly non-trivial. Such multiplicity arises naturally when aggregating over tuning-parameter grids, combining multiple statistics to capture complementary aspects of the alternative, splitting data repeatedly for stability, or merging multiple prediction sets in conformal inference. The central challenge in these problems is to aggregate evidence across multiple statistics or p-values while retaining rigorous finite-sample validity and high power.

A large body of existing work addresses this challenge by first converting each statistic into a p-value and then combining these p-values using procedures that are valid under arbitrary dependence \citep[e.g.,][]{ruger1978,rueschendorf1982,meng1994posterior,vovk2020combining,vovk2022admissible}. While such p-value merging methods enjoy universal validity, their calibration is necessarily driven by worst-case dependence scenarios, which restricts the class of admissible merging functions and often leads to overly conservative procedures. An alternative line of work considers max-type aggregation procedures \citep{albert2022adaptive,schrab2023mmd} with data-dependent calibration via Monte Carlo approximation. Although these methods can yield power improvements in practice, their guarantees are approximation-based, and uniform finite-sample type~I error control remains unresolved. More recently, \citep{guo2024rank} propose a rank-transformed subsampling approach for aggregating exchangeable statistics or p-values that is asymptotically tight; however, their validity guarantees are asymptotic and require knowledge of the limiting null distribution of a single test statistic.

In this work, we revisit aggregation through the exchangeability structure induced by group-invariance schemes under the null hypothesis. Our main theoretical object is what we refer to as single-batch (SB) aggregation, a row-wise construction that applies a single collection of randomized datasets both to standardize individual statistics and to calibrate the aggregated evidence. This construction is closely related to nonparametric combination~\citep{pesarinsalmaso2010} and Westfall--Young-type permutation methods~\citep{westfall1993resampling}. More broadly, the exchangeability structure underlying such aggregation methods has been used to establish finite-sample validity~\citep{shah2018goodness,paik2025integral}. Building on these foundations, we develop a finite-sample theory of power and dependence adaptivity for exchangeability-based aggregation. Our analysis shows that calibrating directly on the realized exchangeable array can yield less conservative thresholds than deterministic worst-case calibrations while preserving finite-sample validity. We also develop sequential and two-batch extensions that broaden the scope of the framework.

\subsection{Contributions}
Our main contributions are detailed as follows.
\begin{itemize}

	\item We formalize single-batch aggregation and analyze its power. We show that its critical value strictly dominates any deterministic calibration valid under arbitrary p-value dependence (\Cref{thm:single-batch power}), and that it adapts to the dependence structure induced by the group-invariance scheme (\Cref{prop: single-batch perfect rank alignment} and \Cref{prop:asymptotic-adaptation}). This can yield significant power improvements over classical p-value merging methods calibrated for worst-case dependence.

	\item We develop a sequential alpha-spending extension of SB minimum aggregation (\Cref{sec:sequential-alpha-spending-aggregation}), maintaining non-asymptotic level control while allowing rejection as soon as early ordered statistics provide strong evidence.

	\item We introduce a two-batch extension of the SB framework (\Cref{sec:two-batch-aggregation}) as a principled mechanism for separating standardization from calibration. This construction enables aggregation rules to be learned from a reference batch while preserving finite-sample validity on the testing batch, and offers substantial computational advantages in settings such as conformal prediction (\Cref{sec:conformal-prediction}).

	\item We demonstrate the scope of the framework through applications to adaptive hypothesis testing and conformal prediction, where our methods recover and extend recent results with simpler arguments, sharper power guarantees, and improved empirical performance.
\end{itemize}

\subsection{Related work}
Several lines of work have addressed the problem of aggregating multiple tests under unknown dependence.

\parheading{P-value merging under dependence.}
Classical p-value combination methods, such as Fisher's method~\citep{fisher1925statistical} and Stouffer's method~\citep{stouffer1949american}, are standard tools for combining evidence across multiple tests. Their usual calibrations, however, typically rely on independence assumptions among the input p-values and may fail to control type~I error under arbitrary dependence. In parallel, a class of universally valid p-value combination methods that remain valid under arbitrary dependence has been established~\citep[e.g.,][]{rueschendorf1982,meng1994posterior,ruger1978} and more recently revisited and systematically characterized by \citep{vovk2020combining,vovk2022admissible}, among others. While these methods enjoy broad validity, their calibration is necessarily driven by worst-case dependence scenarios and can therefore be conservative. Recent work by \citep{gasparin2025combining} shows that p-value aggregation can be strictly improved under exchangeability of the input p-values, by relaxing symmetry constraints and processing p-values sequentially.  In contrast to these approaches, which operate at the level of p-values and are calibrated for broad dependence classes, we leverage the data-level exchangeable structure induced by group-invariance under the null. This allows our procedures to adapt to the specific dependence structure of the transformed statistics and can yield strict power improvements over methods calibrated at the p-value level.

\parheading{Max-type aggregation with Monte Carlo calibration.}
A different line of work focuses on what we refer to as max-type aggregation procedures with data-dependent calibration via Monte Carlo approximation \citep{baraud2003adaptive,fromont2006,fromont2013,albert2015tests,albert2022adaptive,schrab2022ksd,schrab2025unified,schrab2023mmd,schrab2022efficient}. These methods construct a global test by rejecting whenever at least one marginal statistic exceeds an oracle threshold designed to tightly control the type~I error rate. Since this oracle threshold either depends on the unknown joint distribution of the statistics or is difficult to compute in closed form, it is approximated in practice via Monte Carlo calibration. While this approach can yield non-conservative procedures, the resulting finite-sample behavior depends critically on how the Monte Carlo calibration is implemented. In particular, the calibration step uses a separate batch of transformed statistics to approximate the rejection probability but does not enforce joint exchangeability between the calibration batch and the original data. This can result in a type~I error rate that exceeds the nominal level~$\alpha$ for finite $B$, as illustrated in \Cref{sec: maxT aggregation}. As further discussed in \Cref{sec: deferred maxT comparison}, the procedure can equivalently be viewed as aggregating p-values by taking their minimum and estimating a critical value via Monte Carlo approximation. The SB and TB constructions studied here avoid this finite-$B$ approximation issue by calibrating directly across exchangeable testing rows (resulting in valid type I error control), while accommodating general merging functions beyond the minimum p-value.

\parheading{Nonparametric combination and permutation-based constructions.}
Related ideas also appear in earlier work on nonparametric combination (NPC) methods and permutation-based tests. The NPC framework, systematized by \citep[][Chapter~4.2]{pesarinsalmaso2010}, combines dependent permutation tests by applying a combining function to permutation distributions of individual statistics, and has been applied to replicated designs \citep{solmi2014combining} as well as to testing global null hypotheses defined as intersections of multiple component hypotheses \citep{caughey2017nonparametric}. More recently, \citep{shah2018goodness} consider goodness-of-fit tests that aggregate multiple standardized statistics computed from the same data, with validity ensured by exchangeability of replicas generated under the null. Closely related ideas appear in \citep{paik2025integral}, who aggregate statistics across multiple function classes and calibrate the aggregate using a single collection of permutations. Thus, the exchangeability-based validity argument for SB aggregation follows the same principle as earlier permutation-combination work. Our contribution is instead to analyze this construction from the perspective of finite-sample power and dependence adaptivity, and to extend it through sequential spending and a two-batch construction that enables data-dependent aggregation while retaining non-asymptotic validity.

\parheading{Rank-transformed subsampling.}
\citep{guo2024rank} study rank-transformed subsampling methods for aggregating exchangeable statistics or p-values, and establish asymptotic tightness under suitable regularity conditions. Their approach is complementary to ours in that it targets subsampling-based aggregation and asymptotic regimes, whereas we focus on aggregation under group-invariance and derive non-asymptotic validity and power guarantees.

\subsection{Organization}

The remainder of the paper is organized as follows. \Cref{sec: background and prior work} introduces the setup and reviews existing approaches to aggregating multiple tests, including p-value merging methods and max-type aggregation with Monte Carlo calibration. \Cref{sec:single-batch-aggregation} presents the single-batch aggregation procedure, establishes its finite-sample validity, and develops a detailed power analysis, including adaptivity to dependence among p-values and several refined special cases. \Cref{sec:sequential-alpha-spending-aggregation} presents the sequential extension of SB minimum aggregation. \Cref{sec:two-batch-aggregation} introduces the two-batch aggregation procedure, establishes its finite-sample validity, and highlights settings in which separating 
standardization 
from calibration is essential. \Cref{sec: applications} presents examples on adaptive hypothesis testing and conformal prediction, which highlight the scope and practical implications of the proposed methods. \Cref{sec: simulations} presents numerical experiments that corroborate our theoretical findings and demonstrate the empirical performance of the proposed methods before we conclude in \Cref{sec: discussion}. Additional technical results and omitted proofs are collected in the supplementary material.

\subsection{Notation} 
For an integer $n \ge 1$, let $[n] \coloneqq \{1,2,\ldots,n\}$ and $[n]_0 \coloneqq \{0,1,2,\ldots,n\}$. For a real number $x$, we denote by $\lfloor x \rfloor$ and $\lceil x \rceil$ its floor and ceiling, respectively. A random variable $X$ is said to be \emph{super-uniform} if $\mP(X \le t) \le t$ for all $t \in [0,1]$.
For $q\in\mathbb R$, we denote by $\mathrm{Quantile}_{q}\{X_0,X_1,\ldots,X_B\}$ the empirical $q$-quantile of $X_0,X_1,\ldots,X_B$, defined as 
\begin{align*}
	\mathrm{Quantile}_{q}\{X_0,X_1,\ldots,X_B\}
	\coloneqq
	\inf\biggl\{
	t\in\mathbb R:
	\frac1{B+1}\sum_{i=0}^B
	\mathds{1}(X_i\le t)\ge q
	\biggr\},
\end{align*}
with the conventions $\mathrm{Quantile}_{q}=-\infty$ for $q\le0$ and $\mathrm{Quantile}_{q}=+\infty$ for $q>1$. For a sequence of random variables $(Y_n)$ and a random variable $Y$, we write $Y_n \overset{d}{\longrightarrow} Y$ to denote convergence in distribution and $Y_n \overset{p}{\longrightarrow} Y$ to denote convergence in probability. For a given set $A$, we use $|A|$ to denote its cardinality. For two real positive sequences $(a_n)$ and $(b_n)$, we write $a_n \lesssim b_n$ if there exists a constant $C>0$ such that $a_n \leq C b_n$. If both $a_n \lesssim b_n$ and $b_n \lesssim a_n$ hold, we write $a_n \asymp b_n$.

\section{Setup and existing aggregation methods} \label{sec: background and prior work}

This section introduces the group-invariance setup and transformed arrays used throughout the paper. We then review the classical permutation test, which we use as an umbrella term for invariance-based testing procedures induced by group transformations. Although the term “randomization test” is also common in the literature \citep[e.g.,][page 832]{lehmann2022}, we use the terminology “permutation test” throughout for brevity. Finally, we survey existing approaches for aggregating multiple tests, including methods that combine p-values via averaging or order statistics \citep[e.g.,][]{vovk2020combining,vovk2022admissible,gasparin2025combining}, as well as max-type aggregation procedures based on Monte Carlo calibration \citep[e.g.,][]{fromont2006,fromont2013,albert2022adaptive,schrab2023mmd,schrab2025optimal}.

\subsection{Group-invariance and transformed arrays}
\label{subsec:problem setup}

We work within a general group-invariance framework. Let $\mathcal{G}$ be a finite group of transformations acting on the sample space, and let $\mathbf{X}$ denote the observed data. The classical group-invariance hypothesis, also referred to as the randomization hypothesis \citep[][Definition~17.2.1]{lehmann2022}, is stated as follows.

\begin{definition}[Group-invariance hypothesis]
\label{def:group-invariance hypothesis} Under the null hypothesis, the distribution of $\mathbf{X}$ is invariant under the action of $\mathcal{G}$; that is, for every $g\in\mathcal{G}$, the random variables $g(\mathbf{X})$ and $\mathbf{X}$ have the same distribution.
\end{definition}

This framework encompasses a broad class of inference problems and underlies many distribution-free methods in statistics. In permutation testing, for example, $\mathcal{G}$ may consist of all permutations of the pooled sample in two-sample or $k$-sample problems, of sign-flip transformations in paired or symmetric designs, or of permutations that shuffle one variable while holding the other fixed in tests of independence~\citep[e.g.,][]{fisher1935design,pesarinsalmaso2010, good2005permutation,lehmann2022}. More generally, conditional randomization tests arise when $\mathcal{G}$ is induced by resampling a subset of features conditional on the others~\citep[e.g.,][]{candes2018panning,berrett2020conditional}. Similar exchangeability arguments also underpin conformal inference, where they justify rank-based calibration of conformity scores~\citep[e.g.,][]{vovk2005algorithmic,lei2018distribution,angelopoulos2024theoretical}.

Let $T^1,\ldots,T^K$ be real-valued test statistics computed from the data $\mathbf{X}$, and consider the problem of aggregating the evidence provided by these statistics into a single test of the group-invariance hypothesis. Let $g_0 \coloneqq \mathrm{id}$ and let $g_1,\ldots,g_B\in\mathcal{G}$ denote a collection of transformations. Unless stated otherwise, we assume that the transformations are chosen so that the transformed data $g_0(\mathbf{X}),g_1(\mathbf{X}),\ldots,g_B(\mathbf{X})$ are \emph{exchangeable} under the group-invariance hypothesis. This assumption includes, as special cases, the canonical Monte Carlo setting in which $g_1,\ldots,g_B$ are i.i.d.\ draws from the uniform distribution on $\mathcal{G}$ conditional on $\mathbf{X}$, as well as sampling without replacement from $\mathcal{G}\setminus\{\mathrm{id}\}$.

Define the statistics computed from the transformed data by
\begin{equation} \label{eq: transformed statistics}
	T_b^k \coloneqq T^k\bigl(g_b(\mathbf{X})\bigr),
	\quad b=0,1,\ldots,B,\ k=1,\ldots,K.	
\end{equation}
Under the group-invariance hypothesis, the collection of row vectors $(T_b^1,\ldots,T_b^K)$ is exchangeable in the index $b=0,1,\ldots,B$. This joint exchangeability is the basic structural fact behind row-wise permutation aggregation. Our analysis below asks how much power is gained by calibrating the aggregate on this realized exchangeable array, and how the construction can be extended to data-dependent aggregation rules.

\subsection{Classical permutation tests}
\label{sec: permutation test}

We begin by reviewing the classical permutation test for assessing the group-invariance hypothesis using a single test statistic $T$. The permutation test dates back to the seminal work of \citep{fisher1935design,pitman1937significance} and has since been extensively studied; see, for example, \cite{pesarinsalmaso2010} and \cite[][Chapter 17]{lehmann2022}. 

Given a test statistic $T$, designed such that larger values indicate stronger evidence against the null hypothesis, the permutation p-value corresponding to the observed statistic $T_0$ is defined as
\begin{align*}
	p(T_{0}) = \frac{1}{B+1} \sum_{b=0}^B \mathds{1}(T_{b} \geq T_{0}),
\end{align*}
where $T_0$ denotes the statistic computed from the original data, and $T_1, \ldots, T_B$ denote the statistics computed from the transformed datasets as in \eqref{eq: transformed statistics}. 
The permutation test rejects the null hypothesis when $p(T_{0}) \leq \alpha$. The permutation test can equivalently be expressed in terms of the test statistic itself. In particular, rejecting the null hypothesis when $p(T_{0}) \leq \alpha$ is equivalent to rejecting when
\begin{align*}
	T_{0} > \mathrm{Quantile}_{1-\alpha} \bigl\{T_{0}, T_{1}, \ldots, T_{B}\bigr\} = T_{(\lceil (1-\alpha)(B+1) \rceil)},
\end{align*}
where $T_{(k)}$ denotes the $k$-th smallest order statistic of $T_0, T_1, \ldots, T_B$. See e.g., \Cref{Lemma: permutation p-value} in \Cref{sec: useful facts on quantiles}. This formulation proves convenient in what follows, especially when we compare and calibrate aggregated statistics. 

Under exchangeability of $T_0,T_1,\ldots,T_B$, the permutation test has rejection probability at most $\alpha$ in finite samples; see, for example, \citep{romano2005exact,hemerik2018exact,ramdas2023permutation}, and \Cref{Lemma: permutation p-value}. As emphasized earlier, our primary interest lies in aggregating multiple test statistics $T_0^1, \ldots, T_0^K$ into a single test for evaluating the group-invariance hypothesis, while maintaining rigorous type~I error control and favorable power properties. We next review existing approaches to this problem. 

\subsection{P-value merging under arbitrary dependence}
\label{sec:merging}
\label{sec:p-value-merging}
In much of the existing literature on multiple testing and aggregation, individual test statistics are first converted into p-values, which are then combined to form a single global test. We review such p-value-based aggregation methods under arbitrary dependence.
Let $p_1,\ldots,p_K$ denote $K$ individual p-values, each assumed to be super-uniform.  A \emph{p-merging function} is a measurable mapping $f \colon [0,1]^K \to [0,1]$ such that, for any collection of input p-values, the aggregated random variable $f(p_1,\ldots,p_K)$ is itself super-uniform.  Equivalently, a p-merging function combines multiple p-values into a single valid p-value while guaranteeing type~I error control under \emph{arbitrary} dependence structures among the inputs.

Classical p-value merging procedures combine marginal p-values through deterministic rules that remain valid under arbitrary dependence, such as order-statistic merging and generalized-mean merging. Two prominent examples are the O-family~\citep{ruger1978} and the M-family~\citep{vovk2020combining}. These procedures are necessarily calibrated for worst-case dependence, which can lead to conservativeness and reduced power in realistic settings. Details on the O- and M-families, including their precise formulas, sharp constants, and admissibility properties, are recalled in \Cref{sec:p-merging-families}.

\subsection{Max-type aggregation with Monte Carlo calibration}
\label{sec: maxT aggregation}

We review a class of max-type (MaxT) aggregation procedures, following \citep{albert2022adaptive} and \citep{schrab2023mmd}, with earlier roots in \cite{baraud2003adaptive,fromont2006,fromont2013,albert2015tests}. For a single statistic $T^k$, the permutation test rejects the null hypothesis whenever
\begin{align*}
T_{0}^k > \mathrm{Quantile}_{1-\alpha}\bigl\{T_{0}^k, T_{1}^k, \ldots, T_{B}^k\bigr\}.
\end{align*}
To aggregate multiple statistics $T_0^1,\ldots,T_0^K$, the procedures of
\cite{albert2022adaptive,schrab2023mmd} construct a MaxT test that rejects the null hypothesis if at least one statistic exceeds its marginal permutation threshold after inflating the nominal significance level by a data-dependent factor $\tilde{u}_\alpha$, namely,
\begin{align} \label{Eq. max-type test}
	\max_{k \in [K]}
	\Bigl\{
	T_{0}^k - \mathrm{Quantile}_{1-\tilde{u}_\alpha}
	\bigl\{T_{0}^k,\ldots,T_{B}^k\bigr\}\Bigr\} > 0.
\end{align}
The correction factor $\tilde{u}_\alpha$ is intended to approximate the largest inflation of the nominal level that preserves type~I error control. In practice, \citep{albert2022adaptive} and \citep{schrab2023mmd} estimate this quantity via Monte Carlo calibration
\begin{align}
\label{eq: estimated u alpha}
\tilde{u}_\alpha
\coloneqq
\sup\biggl\{
u \in (0,1)
:
\frac{1}{B}
\sum_{b=1}^B
\mathds{1}
\biggl(
\max_{k \in [K]}
\Bigl\{
T_{B+b}^k
-
\mathrm{Quantile}_{1-u}
\bigl\{T_{0}^k,\ldots,T_{B}^k\bigr\}
\Bigr\}
> 0
\biggr)
\leq \alpha
\biggr\},
\end{align}
where $\tilde{u}_\alpha$ is set to zero if the supremum is taken over an empty set.\footnote{Prior work allows non-uniform non-negative weights $w_k$ in $\mathrm{Quantile}_{1-u w_k}\{T_0^k,\ldots,T_B^k\}$. We take $w_k=1$ here for notational simplicity.} Operationally, this procedure uses an additional batch of transformed statistics $T_{B+1}^k,\ldots,T_{2B}^k$ to Monte Carlo-approximate the population-level rejection probability, which may be viewed as an additional calibration step for type~I error control. While the supremum in \eqref{eq: estimated u alpha} is typically computed via bisection, we derive in \Cref{prop: u_alpha expression} a closed-form expression for $\tilde{u}_\alpha$, eliminating the need for iterative search.

Despite its practical appeal, the Monte Carlo calibration in \eqref{eq: estimated u alpha} does not in general guarantee finite-sample level control. This can already be seen in the single-statistic case $K=1$. Suppose that $T_0,T_1,\ldots,T_{2B}$ are exchangeable and almost surely distinct, and let $\tilde u_\alpha$ be the Monte Carlo critical value obtained from \eqref{eq: estimated u alpha} with $K=1$. Then the resulting test satisfies
\begin{align*}
\mP\!\left(
T_0 >
\mathrm{Quantile}_{1-\tilde u_\alpha}
\{T_0,\ldots,T_B\}
\right)
=
\frac{\lfloor B\alpha \rfloor+1}{B+1};
\end{align*}
see \Cref{prop:failure of size control without bisection}. This quantity can exceed $\alpha$ for finite $B$, although the discrepancy vanishes as $B\to\infty$; the empirical level study in \Cref{fig: figure_1} illustrates this finite-$B$ size inflation. Thus, the Monte Carlo step is only an approximation, not a finite-sample rank calibration. The SB and TB procedures below avoid this issue by calibrating directly over exchangeable testing indices.

\section{Single-batch aggregation}
\label{sec:single-batch-aggregation}

We begin with single-batch aggregation, which uses the same batch of transformations for both standardization and calibration. This section is organized as follows. We formulate the SB procedure in \Cref{sec:SB-procedure}, establish its finite-sample validity in \Cref{sec:SB-finite-sample-validity}, and study its power properties in \Cref{sec:SB-power-analysis}. In \Cref{sec: SB adaptivity}, we investigate its adaptivity to dependence among the test statistics. \Cref{sec: SB minimum merging function} provides a detailed treatment of the minimum merging function as a representative special case. Additional material on data-driven selection of merging functions within the SB framework is deferred to \Cref{sec: data-driven aggregation of merging functions}.

\subsection{Procedure}
\label{sec:SB-procedure}

\begin{algorithm}[t]
\caption{Single-Batch Aggregation (SB)}
\label{alg:SB}
\begin{algorithmic}[1]
\Require data $\mathbf{X}$; statistics $T^1,\ldots,T^K$; transformations $g_0,\ldots,g_B$ ($g_0=\mathrm{id}$); merging function $f$; level $\alpha$.
\State \textbf{Transformed statistics:} For each $b\in[B]_0$ and $k\in[K]$, compute $T_b^k \gets T^k(g_b(\mathbf{X}))$.
\State \textbf{Standardization via p-values:} For $b\in[B]_0$ and $k\in[K]$, compute the permutation p-value
\begin{align*}
p\big(T_b^k\big) \gets \frac{1}{B+1}\sum_{i=0}^B \mathds{1}\big(T_i^k \ge T_b^k\big).
\end{align*}
\State \textbf{Aggregation row-wise:} For each $b\in[B]_0$, set
\begin{align*}
f_b \gets f\bigl(p\big(T_b^1\big),\ldots,p\big(T_b^K\big)\bigr).
\end{align*}
\State \textbf{Calibration:} Compute the p-value
\begin{align*}
p_{\mathrm{SB}} \gets \frac{1}{B+1}\sum_{b=0}^B \mathds{1}\big(f_b \le f_0\big).
\end{align*}
\State \textbf{Decision rule:} Reject if $p_{\mathrm{SB}}\le\alpha$.
\end{algorithmic}
\end{algorithm}

We formulate the SB aggregation procedure, which has its roots in \cite{pesarinsalmaso2010,westfall1993resampling,solmi2014combining,caughey2017nonparametric,shah2018goodness,paik2025integral}. The central construction is to aggregate evidence across multiple test statistics while preserving exchangeability by applying the same set of transformations to all statistics simultaneously. While prior work has primarily focused on validity, a systematic analysis of power and adaptivity properties has remained largely unexplored. Moreover, some earlier formulations \citep[e.g.,][]{schrab2023mmd} rely on Monte Carlo permutation p-values that do not fully preserve joint exchangeability across transformed copies, which can lead to slight finite-sample liberality as pointed out in \Cref{prop:failure of size control without bisection}. Other approaches \citep[e.g.,][]{baraud2003adaptive,fromont2006,fromont2013,albert2015tests,albert2022adaptive} assume access to an oracle critical value, but such procedures are not practically implementable, requiring knowledge of the unknown joint null distribution. By contrast, the SB formulation enforces rank-based calibration across all transformations and is exactly implementable, requiring neither approximation nor oracle knowledge.

Consider $B+1$ transformations $g_0,\ldots,g_B$, with $g_0$ being the identity. For each statistic $T^k$ and each transformation index $b\in[B]_0$, compute
\begin{align*}
T_b^k \coloneqq T^k\bigl(g_b(\mathbf{X})\bigr).
\end{align*}
Using these values, define the permutation p-value
\begin{align} \label{Eq: individual p-value}
p(T_b^k)
\coloneqq
\frac{1}{B+1}\sum_{i=0}^B \mathds{1}\!\bigl(T_i^k \ge T_b^k\bigr),
\qquad b\in[B]_0,\ k\in[K],	
\end{align}
which ranks each $T_b^k$ among $\{T_0^k,\ldots,T_B^k\}$.

Let $f\colon [0,1]^K \to\mathbb{R}$ be a function that aggregates $K$ permutation p-values into a single number. No structural assumptions on $f$ are required for finite-sample validity of the SB procedure. For convenience of interpretation, however, we assume throughout that smaller values of $f$ indicate stronger evidence against the null hypothesis, consistent with the interpretation of p-values. Typical examples include the minimum, the average, and the median of p-values. For each transformation index $b$, aggregate the p-values computed from the corresponding transformed dataset via 
\begin{align*}
f_b \coloneqq f\bigl(p(T_b^1),\ldots,p(T_b^K)\bigr),
\qquad b\in[B]_0.
\end{align*}
Crucially, the statistics $T_b^1,\ldots,T_b^K$ are all computed under the same transformation $g_b$. Thus, the aggregation is performed \emph{row-wise} across statistics sharing a common transformation, inducing a coupling that preserves exchangeability. The SB aggregation test then compares the aggregated value $f_0$ (corresponding to the original data) to its transformed counterparts:
\begin{align*}
p_{\mathrm{SB}}
\coloneqq
\frac{1}{B+1}\sum_{b=0}^B \mathds{1}\bigl(f_b \le f_0\bigr),
\end{align*}
and rejects the null hypothesis whenever $p_{\mathrm{SB}}\le\alpha$. The SB aggregation procedure is summarized in \Cref{alg:SB}; a schematic illustration is deferred to \Cref{fig:diagram-SB} in \Cref{sec: diagrams}.

\subsection{Finite-sample validity}
\label{sec:SB-finite-sample-validity}

The finite-sample validity of the SB aggregation test follows directly from exchangeability. Under the group-invariance hypothesis, the row vectors
\begin{align*}
(T_0^1,\ldots,T_0^K),\ldots,(T_B^1,\ldots,T_B^K)
\end{align*}
are exchangeable, since each row is obtained by applying the same transformation to the data. The column-wise p-value transformation preserves this exchangeability in the sense that permuting the rows of the statistics matrix induces the same permutation of the rows of the p-value matrix. Applying the merging function row-wise therefore yields aggregated values $(f_0,\ldots,f_B)$ that are also exchangeable. Consequently, the rank-based p-value $p_{\mathrm{SB}}$ is super-uniform under the null, and the SB aggregation test controls the type~I error rate at level~$\alpha$ in finite samples. We formalize this statement in the following theorem.
\begin{theorem}
\label{prop: SB level}
Under the group-invariance hypothesis, the SB aggregation test satisfies
\begin{align*}
\mP\bigl(p_{\mathrm{SB}} \le \alpha\bigr) \le \alpha
\quad \text{for all } \alpha \in (0,1) \text{ and } B \ge 1.
\end{align*}
Moreover, if $f_0,f_1,\ldots,f_B$ are distinct with probability one, then
\begin{align*}
\mP\bigl(p_{\mathrm{SB}} \le \alpha\bigr)
=
\frac{\lfloor (B+1)\alpha \rfloor}{B+1}.
\end{align*}
\end{theorem}

We stress again that the above validity result holds for any choice of the merging function $f$. In particular, the p-value transformation in \Cref{alg:SB} is not essential for finite-sample validity, and $f$ may be applied directly to the collection of statistics $\{T_b^k\}_{k\in[K]}$ computed from each transformed dataset \citep{biggs2023mmdfuse,ribero2026regularized,zhou2026dual}. Nevertheless, working with p-values provides a natural standardization across statistics, which reduces sensitivity to differences in scale. As alternative forms of standardization, one may also studentize each statistic prior to aggregation \citep[e.g.,][]{shah2018goodness}, or transform each statistic using its null distribution function when the (asymptotic) null distribution is known \citep[e.g.,][]{guo2024rank}.

More generally, the validity result extends to procedures that break ties in the marginal statistics $T_b^k$ and/or in the merged statistics $f_b$ using auxiliary variables, provided that augmenting each row with these variables preserves exchangeability. This includes standard randomized constructions, such as lexicographical tie-breaking with i.i.d.~uniform auxiliary variables. The same principle also accommodates deterministic, data-dependent choices of the auxiliary variables, enabling tie-breaking rules that exploit the structure of the aggregated statistics while preserving finite-sample validity. Concrete constructions and their power implications are deferred to \Cref{sec:tie-breaking}.

For later use, we record the following equivalent threshold representation of the SB decision rule, which rejects when
\begin{align}
\label{eq:SB_statistic_view}
	f_0
	<
	-\mathrm{Quantile}_{1-\alpha}\bigl\{-f_0,-f_1,\ldots,-f_B\bigr\}
	\coloneqq
	\uSB.
\end{align}
To facilitate several subsequent arguments, we provide an equivalent characterization of $\uSB$ in terms of a \emph{supremum-based} empirical quantile.
\begin{proposition}
\label{prop: single-batch equivalence}
The SB threshold in \eqref{eq:SB_statistic_view} admits the equivalent representation
\begin{align}
\label{eq: u_alpha general expression any function}
\uSB = 
\sup \biggl\{u \in \mathbb{R} : \frac{1}{B+1} \sum_{b=0}^B \mathds{1}\bigl(
f_b \leq u \bigr) \leq \alpha  \biggr\}.
\end{align}
\end{proposition}
This alternative characterization will be repeatedly invoked throughout the paper. In particular, it plays a central role in the power and adaptivity analysis.

\subsection{Power dominance over deterministic calibration}
\label{sec:SB-power-analysis}

We study the power properties of the SB aggregation test. Our main result (\Cref{thm:single-batch power}) shows that, for any merging function $f$, the SB procedure is at least as powerful as any test based on a deterministic threshold calibrated under arbitrary dependence among p-values. Moreover, the SB threshold adapts to the dependence structure of the p-values and can yield strict power improvements over worst-case dependence calibrations. 

We begin with a super-uniformity lemma that underpins the power analysis. A more general weighted version is given in \Cref{Lemma: generalized superuniform lemma}, which extends
\citep[][Lemma~A1]{harrison2012conservative} to arbitrary measures.

\begin{lemma}
\label{lemma:super-uniform-property}
Let $T_0,T_1,\ldots,T_B \in \mathbb{R}$ and define
\begin{align*}
p_b \coloneqq \frac{1}{B+1} \sum_{i=0}^B \mathds{1}(T_i \geq T_b),
\qquad b\in[B]_0.
\end{align*}
Then, for any $\alpha \in [0,1]$ and any integer $B \geq 1$,
\begin{align*}
\frac{1}{B+1} \sum_{b=0}^B \mathds{1}(p_b \leq \alpha)
\leq
\frac{\lfloor (B+1)\alpha \rfloor}{B+1}
\leq
\alpha.
\end{align*}
Moreover, if $T_0,\ldots,T_B$ are distinct, the first inequality holds with equality.
\end{lemma}

By \Cref{lemma:super-uniform-property}, conditional on the observed data $\mathbf{X}$, the permutation p-value $p(T_b^k)$ with $b \sim \mathrm{Unif}([B]_0)$ is super-uniform. In particular, for each $k \in [K]$,
\begin{align*}
\frac{1}{B+1} \sum_{b=0}^B \mathds{1}\bigl(p(T_b^k) \leq \alpha\bigr)
\leq
\frac{\lfloor (B+1)\alpha \rfloor}{B+1}
\leq
\alpha.
\end{align*}
This conditional super-uniformity underlies the finite-sample power comparison below.

\begin{theorem}\label{thm:single-batch power}
Fix $\alpha\in(0,1)$. Let $U=(U_1,\ldots,U_K)$ be a random vector with an arbitrary joint distribution such that each coordinate $U_k$ is super-uniform. Let $\mathcal{U}_{\mathrm{sup}}^K$ denote the class of all such joint laws. Suppose that $c_{\alpha,K}$ is a deterministic constant satisfying
\begin{align*}
\sup_{\nu\in\mathcal{U}_{\mathrm{sup}}^K}
\mP_{\nu}\bigl(f(U_1,\ldots,U_K)\le c_{\alpha,K}\bigr)
\le \alpha.
\end{align*}
Then it holds almost surely that
\(
c_{\alpha,K} < \uSB.
\)
Consequently, the SB aggregation test defined in
\Cref{alg:SB} satisfies
\begin{align*}
\mP\!\left( f_0 \ge \uSB \right)
\le
\mP\!\left( f_0 > c_{\alpha,K} \right).
\end{align*}
\end{theorem}

\Cref{thm:single-batch power} shows that the SB threshold strictly dominates any deterministic threshold ensuring type~I error control under arbitrary dependence. Equivalently, the type~II error of the SB test is no larger than that of any level-$\alpha$ test based on such a deterministic threshold.
Despite the strength of the result, the argument is elementary. It uses only the fact that, conditional on the observed data, the permutation p-values $p(T_b^1),\ldots,p(T_b^K)$ are super-uniform when $b$ is drawn uniformly from $[B]_0$. A deterministic threshold must be valid for every joint law with super-uniform margins, whereas the SB threshold is calibrated directly on the realized conditional empirical distribution. This is why SB obtains a larger critical value and hence higher power. It is worth emphasizing that \Cref{thm:single-batch power} is entirely deterministic and imposes no assumptions on the transformations $g_1,\ldots,g_B$. Consequently, the result continues to hold for general weighted permutation schemes, including those considered in \citep[][Theorem~2]{ramdas2023permutation}.

As an immediate consequence, we obtain the following comparison with existing p-merging procedures calibrated under arbitrary dependence. 

\begin{corollary}
\label{corollary: single-batch power O and M}
Let $f$ be a classical p-merging rule calibrated under arbitrary dependence, such as an O-family or M-family rule; see \Cref{sec:p-merging-families} for the definitions. Since such rules output valid p-values under arbitrary dependence, the corresponding deterministic calibration is $c_{\alpha,K}=\alpha$.
Consequently, the SB aggregation test is always at least as powerful as the corresponding worst-case calibrated test.
\end{corollary}

\subsection{Adaptivity to dependence} \label{sec: SB adaptivity}

We now illustrate how the SB procedure adapts to the unknown dependence structure among p-values. A key limitation of p-merging methods calibrated under arbitrary dependence is that their critical values are typically driven by worst-case dependence and can therefore be overly conservative outside special cases. In contrast, the SB procedure is data-dependent and can adapt to the dependence structure induced by the underlying data. 
We illustrate this adaptivity in two settings: the extreme case of perfect rank alignment and a more general asymptotic regime.

\parheading{Extreme case: perfect rank alignment.}
We begin by considering an extreme case in which no multiplicity adjustment is needed. The test statistics need not be identical across coordinates; it suffices that they induce the same ordering across transformations.
Formally, suppose that for all $k,k' \in [K]$ and all $b,b' \in [B]_0$,
\begin{equation}
\label{eq:rank_alignment}
T_b^k \le T_{b'}^k
\quad \Longleftrightarrow \quad T_b^{k'} \le T_{b'}^{k'} .
\end{equation}
That is, the vectors $(T_0^k,\ldots,T_B^k)$ are perfectly rank-aligned across coordinates, although their numerical values may differ (e.g., different scalings of the same statistic). Under
\eqref{eq:rank_alignment}, the permutation p-values satisfy
\begin{align*}
p(T_b^1)=\cdots=p(T_b^K)
\qquad \text{for all } b\in[B]_0,
\end{align*}
since each column induces the same ranking of the transformed statistics.

Consider a monotone merging function $f$ satisfying
\begin{equation}
\label{eq:diagonal_monotone}
f(x,\ldots,x) \le f(y,\ldots,y)
\quad \Longleftrightarrow \quad x \le y,
\qquad x,y \in \mathbb{R}.
\end{equation}
This condition is satisfied by many commonly used aggregation rules, including the minimum, the median, and the arithmetic mean. The next proposition shows that, under perfect rank alignment, the resulting SB aggregation test coincides with the usual permutation test based on a single statistic.

\begin{proposition}
\label{prop: single-batch perfect rank alignment} 
Consider the SB aggregation test in \Cref{alg:SB}. Assume that the merging function $f$ satisfies the diagonal monotonicity condition \eqref{eq:diagonal_monotone}. If the statistics $\{T_b^k\}$ satisfy the rank alignment condition \eqref{eq:rank_alignment}, then
\begin{align*}
p_{\mathrm{SB}}
=
\frac{1}{B+1}\sum_{b=0}^B \mathds{1}(T_b \ge T_0),
\end{align*}
where $T_b$ denotes any representative statistic (e.g., $T_b^1$).
\end{proposition}

Thus, whenever multiple test statistics convey the same ordering information across permutation samples, the SB procedure automatically detects this redundancy and reduces to a single-statistic permutation test. In this case, the effective number of tests is one, and no multiplicity correction is applied.

\parheading{Asymptotic adaptivity.}
We next develop more general adaptivity results using asymptotic arguments. Suppose that the limiting null distribution of the merging function $f$ applied to permutation p-values were known. This distribution implicitly encodes the dependence structure among the p-values and thus defines an oracle benchmark. In that case, one could reject the null hypothesis whenever $f_0$ is less than its oracle $\alpha$-quantile. We show below that, under suitable conditions, the SB threshold converges to this oracle $\alpha$-quantile.

\begin{proposition}
\label{prop:asymptotic-adaptation}
For each $n$, let $\mathbf X^{(n)}$ denote the observed data set, and let
$f_{b,n}$, $b\in[B]_0$, be the SB aggregation scores computed from the
corresponding permutation p-values. Suppose that, conditional on
$\mathbf X^{(n)}$, the transformations $g_1,\ldots,g_B$ are i.i.d.\ uniform on
$\mathcal G$. Fix $\alpha\in(0,1)$, and consider any joint asymptotic regime in
which $B,n\to\infty$. Assume that
\begin{align*}
(f_{1,n},f_{2,n})
\;\overset{d}{\longrightarrow}\;
(V_\infty,V_\infty'),
\end{align*}
where $V_\infty$ and $V_\infty'$ are i.i.d.\ with distribution function $F$.
Let
\(
Q^\star_\alpha
\coloneqq
\inf\{u\in\mathbb R:F(u)\ge\alpha\}
\)
denote the oracle lower $\alpha$-quantile, and assume that it is well separated:
for every $\varepsilon>0$,
\begin{align*}
F(Q^\star_{\alpha}-\varepsilon)<\alpha<F(Q^\star_{\alpha}+\varepsilon).
\end{align*}
Then the SB critical value satisfies \(\uSB \;\overset{p}{\longrightarrow}\; Q^\star_\alpha\).
\end{proposition}

As a consequence, when both $B$ and $n$ are large, the SB threshold behaves as if the limiting null distribution of the merging function were known. This provides an asymptotic explanation for the adaptivity of the SB aggregation test to the unknown dependence structure among the test statistics or p-values. The joint convergence of $(f_{1,n},f_{2,n})$ to an i.i.d.\ limit is standard in the asymptotic theory of permutation distributions and is indeed a necessary condition for the permutation distribution to converge; see \citep[][Theorem~5.1]{chung2013}. A uniform version over classes of null distributions is deferred to \Cref{prop:uniform-asymptotic-adaptation}.

\subsection{Minimum merging and Westfall--Young calibration}
\label{sec: SB minimum merging function}

We now specialize the SB aggregation procedure to the minimum merging function, which recovers the Westfall--Young single-step method \citep{westfall1993resampling}; see also \citep{meinshausen2011asymptotic} for its asymptotic properties. Specifically, we consider
\begin{align}
\label{eq: minimum merging function}
f\bigl(p(T_b^1),\ldots,p(T_b^K)\bigr)
=
\min_{k\in[K]} p(T_b^k)
\eqqcolon f_b^{\mathrm{min}},
\qquad b\in[B]_0.
\end{align}
Using the statistic formulation in \eqref{eq:SB_statistic_view}, the SB minimum test rejects the null hypothesis whenever $f_0^{\mathrm{min}} < \hat u_{\alpha,\min}^{\mathrm{SB}}$, where
\begin{align*}
\hat u_{\alpha,\min}^{\mathrm{SB}}
\coloneqq
-
\mathrm{Quantile}_{1-\alpha}
\bigl\{-f_0^{\mathrm{min}},-f_1^{\mathrm{min}},\ldots,-f_B^{\mathrm{min}}\bigr\}.
\end{align*}
The next proposition shows that the SB minimum test is bounded between the Bonferroni procedure and the unadjusted minimum test.

\begin{proposition}
\label{prop: minimum single-batch calibration statistic view}
The SB minimum test satisfies, almost surely,
\begin{align*}
\mathds{1}\!\left(
\min_{k \in [K]} p(T_{0}^k)
\leq \frac{\alpha}{K}\right)
\leq
\mathds{1}\!\left(
\min_{k \in [K]} p(T_{0}^k)
< \uSBmin
\right)
\leq
\mathds{1}\!\left(
\min_{k \in [K]} p(T_{0}^k)
\leq \alpha\right).
\end{align*}
Moreover, $\uSBmin \geq (\lfloor(B+1)\alpha/K\rfloor + 1)/(B+1)$.
\end{proposition}

The additional term $1/(B+1)$ in the second statement reflects the strict inequality in the SB decision rule: the SB minimum test rejects when $\min_{k \in [K]} p(T_{0}^k) < u$ rather than $\le u$ as in the other procedures, which shifts the rejection threshold by one discretization step. The inequalities in Proposition~\ref{prop: minimum single-batch calibration statistic view} are tight with each bound attainable. In particular, when the coordinate-wise rejection events are disjoint, the Bonferroni correction can become tight (and is tight when the p-values are exactly uniform); in such cases the SB minimum test may coincide with the Bonferroni test (see \Cref{ex:SBmin_equals_Bonf} for an explicit example). At the opposite extreme, when the rejection events coincide across coordinates, the SB minimum test reduces to the unadjusted single test (see \Cref{sec: SB adaptivity}). Together, these cases show that the SB minimum test interpolates between the Bonferroni and unadjusted procedures according to the dependence structure among the permutation p-values, while preserving level~$\alpha$.

Additional material on data-driven aggregation over merging functions is deferred to \Cref{sec: data-driven aggregation of merging functions}.

\section{Sequential minimum aggregation}
\label{sec:sequential-alpha-spending-aggregation}

We next consider an alpha-spending version of SB minimum aggregation for ordered coordinates. In many applications, the coordinates are not merely an unordered collection of statistics, but have a natural order, such as increasing resolution, decreasing bandwidth, growing model complexity, or sequentially available data sources. In such settings, the alternative may already be visible in an early prefix of the sequence. Calibrating only the final minimum over all $K$ coordinates can then be inefficient because later coordinates with little signal may still enter the calibration and make rejection harder.

The general sequential aggregation principle is recorded in \Cref{sec: general sequential aggregation}. Here we specialize it to cumulative minima. The goal is to spend the global level $\alpha$ across a pre-specified sequence of prefix tests: at stage $j$, we inspect the cumulative minimum over the first $j$ coordinates and spend only a portion $\alpha_j$ of the total level. For $b\in[B]_0$ and $j\in[K]$, set $m_{b,j}\coloneqq \min_{1\le k\le j} p(T_b^k)$. Let $\alpha_1,\ldots,\alpha_K\ge 0$ satisfy $\sum_{j=1}^K \alpha_j\le \alpha$, and write $q_j\coloneqq \lfloor (B+1)\alpha_j\rfloor$.

The construction can be understood as a row-removal procedure. The set $S_{j-1}$ contains the permutation rows that have not yet been used for rejection at earlier stages. At stage $j$, among these surviving rows, we look at the prefix statistic $m_{b,j}$ and remove the rows whose values are among the most extreme according to the stage-$j$ spending budget. If the observed row $b=0$ is removed at any stage, the global null is rejected. Removing rows, rather than testing each prefix separately, keeps the eliminated sets disjoint and makes the finite-sample level calculation explicit.

\begin{algorithm}[t]
\caption{Sequential SB Minimum Aggregation (SeqSB)}
\label{alg:SeqSB-min}
\begin{algorithmic}[1]
\Require data $\mathbf X$; ordered statistics $T^1,\ldots,T^K$; transformations $g_0,\ldots,g_B$ with $g_0=\mathrm{id}$; spending sequence $\alpha_1,\ldots,\alpha_K$.
\State Set $S_0\gets [B]_0$ and $m_{b,0}\gets 1$ for all $b\in[B]_0$.
\For{$j=1,\ldots,K$}
\State Compute $T_b^j \gets T^j(g_b(\mathbf X))$ for all $b\in [B]_0$.
\State Compute permutation p-values $p(T_b^j)$ as in \eqref{Eq: individual p-value} for all $b\in S_{j-1}$.
    \State Update $m_{b,j}\gets \min\{m_{b,j-1},p(T_b^j)\}$ for all $b\in S_{j-1}$.
    \State Set $q_j\gets \lfloor(B+1)\alpha_j\rfloor$.
    \State Let $c_j$ be the $(q_j+1)$-st order statistic of $\{m_{b,j}:b\in S_{j-1}\}$.
    \State Set $A_j\gets \{b\in S_{j-1}:m_{b,j}<c_j\}$.
    \If{$0\in A_j$}
        \State Reject $H_0$ and stop.
    \EndIf
    \State Set $S_j\gets S_{j-1}\setminus A_j$.
\EndFor
\State Do not reject $H_0$.
\end{algorithmic}
\end{algorithm}

To define the procedure, initialize the survivor set as $S_0=[B]_0$. At stage $j$, given the current survivor set $S_{j-1}$, let
\begin{align*}
    m_{(1),j}^{S_{j-1}}
    \le
    \cdots
    \le
    m_{(|S_{j-1}|),j}^{S_{j-1}}
\end{align*}
denote the order statistics of $\{m_{b,j}:b\in S_{j-1}\}$. Define the stage-$j$ threshold $c_j \coloneqq m_{(q_j+1),j}^{S_{j-1}}$, eliminate
$A_j\coloneqq\{b\in S_{j-1}:m_{b,j}<c_j\}$, and update $S_j=S_{j-1}\setminus A_j$. The sequential minimum aggregation test rejects $H_0$ if $0\in \bigcup_{j=1}^K A_j$. The procedure is summarized in \Cref{alg:SeqSB-min}, and the next proposition establishes its finite-sample validity.

\begin{proposition}
\label{prop: sequential minimum aggregation}
Fix $\alpha\in(0,1)$, $B\ge 1$, and
$\alpha_1,\ldots,\alpha_K\ge 0$ such that
$\sum_{j=1}^K\alpha_j\le \alpha$. Under the group-invariance null hypothesis,
the sequential minimum aggregation test satisfies
\begin{align*}
    \mP\biggl(0\in \bigcup_{j=1}^K A_j\biggr)
    \le
    \frac{1}{B+1}\sum_{j=1}^K
    \bigl\lfloor (B+1)\alpha_j\bigr\rfloor
    \le
    \alpha .
\end{align*}
Moreover, if for every $j\in[K]$ the survivor values
$\{m_{b,j}:b\in S_{j-1}\}$ are almost surely distinct, then the first
inequality is an equality.
\end{proposition}

The proof is based on a simple exchangeability argument. The row-removal construction is equivariant with respect to permutations of the indices $[B]_0$: if the rows of the transformed array are permuted, then the survivor sets and eliminated sets are permuted in the same way. Under the null hypothesis, the row array is exchangeable, and therefore for each stage $j$, 
\begin{align*}
    \mP(0\in A_j)
    =
    \frac{\mE |A_j|}{B+1}.
\end{align*}
The eliminated sets $A_1,\ldots,A_K$ are disjoint by construction, and
$|A_j|\le q_j$. Hence
\begin{align*}
    \mP\biggl(0\in\bigcup_{j=1}^K A_j\biggr)
    =
    \sum_{j=1}^K \mP(0\in A_j)
    \le
    \frac{1}{B+1}\sum_{j=1}^K q_j .
\end{align*}
If the survivor values are distinct at every stage, then exactly $q_j$ rows are
removed at stage $j$, giving equality in the first bound.

The SB minimum rule is recovered by spending all level at the final coordinate,
$\alpha_K=\alpha$ and $\alpha_j=0$ for $j<K$. Thus, the sequential procedure can
be viewed as a generalization of the SB minimum rule to ordered prefixes. It is
not uniformly more powerful than the final SB minimum rule, since spending level
early necessarily leaves less level for later coordinates. However, it can be
advantageous when the signal is expected to appear in an early prefix. In that
case, the stage-$j$ threshold is calibrated only against the first $j$
coordinates, rather than against the full minimum over all $K$ coordinates, and
hence may be less conservative. The procedure also has an operational
advantage, as once the observed row is eliminated, the test can stop without
computing later coordinates.

\section{Two-batch aggregation and data-dependent rules}
\label{sec:two-batch-aggregation}

The SB procedure in \Cref{sec:single-batch-aggregation} uses a single collection of transformed datasets both to (i) standardize each coordinate statistic via a permutation p-value map and (ii) calibrate the merged evidence through a permutation test. This section studies a TB aggregation procedure that leverages an extra batch for calibration while preserving finite-sample type~I error control. Conceptually, TB aggregation proceeds as follows: a \emph{reference batch} of transformations is used to construct a standardized p-value map for each coordinate statistic, and an independent \emph{testing batch} is then used to run a permutation test on the merged evidence computed from these standardized p-values.

\subsection{Holdout standardization}
\label{sec:TB-procedure-and-finite-sample-validity}

We begin with a formal description of the TB aggregation procedure.
Let $g_0=\mathrm{id}$ and let $g_1,\ldots,g_{2B}$ be additional transformations. As in the SB aggregation procedure, assume that $g_1,\ldots,g_{2B}$ are chosen such that $g_0(\mathbf{X}), g_1(\mathbf{X}),\ldots, g_{2B}(\mathbf{X})$ are exchangeable under the group-invariance hypothesis. We view $\{g_0,g_1,\ldots,g_B\}$ as a \emph{testing batch}, and $\{g_{B+1},\ldots,g_{2B}\}$ as a \emph{reference batch}.\footnote{For notational convenience, we take the testing batch to have size $B+1$ and the reference batch to have size $B$; however, all results extend directly to arbitrary batch sizes $B_1$ and $B_2$.} Define the statistics
\begin{align*}
T_b^k \coloneqq T^k\bigl(g_b(\mathbf{X})\bigr),
\qquad b\in[B]_0,\ k\in[K].
\end{align*}
For notational convenience, we re-index the \emph{reference batch} by setting
\begin{align*}
\tilde g_i \coloneqq g_{B+i},\qquad i\in[B],
\qquad
\tilde T_i^k \coloneqq T^k\bigl(\tilde g_i(\mathbf{X})\bigr),
\qquad i\in[B],\ k\in[K].
\end{align*}
For each coordinate $k\in[K]$ and each testing index $b\in[B]_0$, define the holdout permutation p-value by ranking the testing statistic $T_b^k$ against the \emph{reference} batch $\{\tilde T_1^k,\ldots,\tilde T_B^k\}$:
\begin{align*}
p_{\mathrm{HO}}\bigl(T_b^k\bigr)
\coloneqq
\frac{1}{B+1}\biggl\{
1+\sum_{i=1}^{B} \mathds{1}\!\left(\tilde T_i^k \ge T_b^k\right)
\biggr\},
\qquad b\in[B]_0,\ k\in[K].
\end{align*}

\begin{algorithm}[t]
\caption{Two-Batch Aggregation (TB)}
\label{alg:TB}
\begin{algorithmic}[1]
\Require data $\mathbf{X}$; statistics $T^1,\ldots,T^K$; transformations $g_0,\ldots,g_{2B}$ ($g_0=\mathrm{id}$); merging function $f$; level $\alpha$.
\State \textbf{Two batches:} Define testing batch $\{g_0,\ldots,g_B\}$ and reference batch $\{g_{B+1},\ldots,g_{2B}\}$.
\State \textbf{Transformed statistics:} Compute $T_b^k \gets T^k(g_b(\mathbf{X}))$ for $b\in[B]_0$, $k\in[K]$.
\State \textbf{Transformed holdout statistics:} Compute $\tilde T_i^k \gets T^k(g_{B+i}(\mathbf{X}))$ for $i\in[B]$, $k\in[K]$.
\State \textbf{Standardization via holdout p-values:} For each $b\in[B]_0$, $k\in[K]$, set
\begin{align*}
p_{\mathrm{HO}}(T_b^k)\gets \frac{1}{B+1}\biggl\{1+\sum_{i=1}^{B} \mathds{1}\big(\tilde T_i^k\ge T_b^k\big)\biggr\}.
\end{align*}
\State \textbf{Aggregation row-wise:} For each $b\in[B]_0$, set $$\tilde f_b \gets f\bigl(p_{\mathrm{HO}}\big(T_b^1\big),\ldots,p_{\mathrm{HO}}\big(T_b^K\big)\bigr).$$
\State \textbf{Calibration:} Compute the p-value $$p_{\mathrm{TB}} \gets \frac{1}{B+1}\sum_{b=0}^B \mathds{1}\big(\tilde f_b \le \tilde f_0\big).$$
\State \textbf{Decision rule:} Reject if $p_{\mathrm{TB}}\le\alpha$.
\end{algorithmic}
\end{algorithm}

Let $f\colon[0,1]^K\to\mathbb{R}$ be a merging function that maps the $K$ holdout p-values to a single real-valued score. As before, smaller values are interpreted as stronger evidence against the null hypothesis.
For each $b\in[B]_0$, define
\begin{align*}
\tilde f_b
\coloneqq f\bigl(
p_{\mathrm{HO}}(T_b^1),\ldots,p_{\mathrm{HO}}(T_b^K)
\bigr),
\qquad b\in[B]_0.
\end{align*}
Finally, define the TB permutation p-value by comparing $\tilde f_0$ to its testing-batch counterparts:
\begin{align}
\label{eq: TB p-value}
p_{\mathrm{TB}}
\coloneqq
\frac{1}{B+1}\sum_{b=0}^B \mathds{1} \bigl(\tilde f_b \le \tilde f_0\bigr),
\qquad
\text{and reject if } p_{\mathrm{TB}}\le\alpha.
\end{align}
In contrast to SB aggregation, where the same transformed statistics are used for both p-value construction and calibration, the TB procedure splits these roles across two independent batches conditional on $\mathbf X$: a reference batch for standardization and a testing batch for calibration. 
The full procedure is summarized in \Cref{alg:TB}, and a schematic illustration is given in \Cref{fig:diagram-TB} in \Cref{sec: diagrams}.

As in \Cref{sec:SB-procedure}, the decision rule
\eqref{eq: TB p-value} admits an equivalent quantile-based form:
\begin{align*}
\tilde f_0
\;<\;
-\mathrm{Quantile}_{1-\alpha}\bigl\{-\tilde f_0,-\tilde f_1,\ldots,-\tilde f_B\bigr\}
\;\eqqcolon\;
\uTB.
\end{align*}
Equivalently, $\uTB$ can be written in the supremum-quantile form
\begin{align}
\label{eq: uTB_sup_form}
\uTB
=
\sup\biggl\{
u\in\mathbb{R}:\ 
\frac{1}{B+1}\sum_{b=0}^B \mathds{1}\!\left(\tilde f_b \le u\right)\le \alpha
\biggr\}.
\end{align}

\subsection{Finite-sample validity}
\label{sec:TB-finite-sample-validity}

Under the group-invariance hypothesis, the testing-batch row vectors
\begin{align*}
(T_0^1,\ldots,T_0^K),\ldots,(T_B^1,\ldots,T_B^K)
\end{align*}
are conditionally exchangeable given
the reference batch $\{(\tilde T_i^1,\ldots,\tilde T_i^K)\}_{i=1}^B$. Since the holdout p-value map and the merging function are applied row-wise, the merged values $(\tilde f_0,\ldots,\tilde f_B)$ remain conditionally exchangeable. Therefore, the rank-based p-value $p_{\mathrm{TB}}$ in
\eqref{eq: TB p-value} is super-uniform under the null and controls the type~I error rate at level~$\alpha$ in finite samples as stated below.
\begin{theorem}
\label{prop: TB level}
Under the group-invariance hypothesis, the TB aggregation test controls the type~I error rate at level~$\alpha$, that is,
\begin{align*}
\mP\bigl(p_{\mathrm{TB}}\le \alpha\bigr)\le \alpha
\quad\text{for all }\alpha\in(0,1)\text{ and }B\ge 1.
\end{align*}
Moreover, if $\tilde f_0,\tilde f_1,\ldots,\tilde f_B$ are distinct with probability one, then
\begin{align*}
\mP\bigl(p_{\mathrm{TB}}\le \alpha\bigr)
=
\frac{\lfloor (B+1)\alpha\rfloor}{B+1}.
\end{align*}
\end{theorem}
More generally, finite-sample validity still holds if the merging function is chosen based on the reference batch, so it may be data-dependent as long as it depends only on the reference information. However, this batch separation can lead to lower finite-sample power when $B$ is small; a more detailed comparison with SB aggregation is deferred to \Cref{sec: power properties of TB}. The next subsection highlights settings where the separation is nevertheless useful.

\subsection{Learning aggregation rules from the reference batch}
\label{sec: TB_favorable}

A key advantage of TB aggregation is that the reference batch can be used to
construct or choose the aggregation rule before the final calibration step. This
enables inference tasks and aggregation strategies that are incompatible with
the SB construction. We highlight two representative settings where the
separation between reference and testing batches is essential.

\parheading{Merging conformal prediction sets.}
In conformal prediction, the score of the test point at its unknown true label is unavailable at calibration time. Because TB holdout p-values are computed from the reference batch alone, the calibration step can be precomputed and reused across candidate labels. This avoids the set-inversion bottleneck that would generally arise from an SB-style calibration incorporating each candidate test-point score into its own ranking set. We return to this connection and its algorithmic implications in \Cref{sec:conformal-prediction} and \Cref{sec:alg-conformal}.

\parheading{Learning the aggregation map.}
Because the aggregation rule may be chosen using only the reference batch, TB aggregation allows data-dependent choices such as learned weights, projections, or dimension-reduced summaries of the coordinate-wise p-value vector. Conditional on the reference batch, the testing-batch rows remain exchangeable, so the final permutation calibration remains finite-sample valid. Such data-adaptive choices are not available under SB aggregation without additional correction, since learning the rule and calibrating the test on the same transformed rows would generally destroy the exchangeability underpinning the SB procedure. A PCA-based example is given in \Cref{sec:pca-reference-batch}.

\section{Applications} \label{sec: applications}
This section presents notable theoretical and methodological applications of the preceding results. \Cref{sec: adaptive_minimax} and \Cref{sec: adaptive MMD and HSIC tests} consider general and kernel-based adaptive nonparametric testing with a focus on SB aggregation, while \Cref{sec:conformal-prediction} examines both SB and TB aggregation in the context of conformal prediction, with particular emphasis on the computational advantages of TB aggregation. Additional material on SB aggregation of data-asymmetric tests is deferred to \Cref{sec: aggregation of data-asymmetric tests}.

\subsection{Adaptive nonparametric testing}
\label{sec: adaptive_minimax}

Many nonparametric testing problems require adaptation over unknown tuning parameters, such as smoothness or intrinsic dimension. We show that SB aggregation can replace Bonferroni aggregation in such adaptive constructions, preserving existing separation-rate guarantees while improving finite-sample power.

\parheading{Setup and minimax separation radius.}
Let $\mathcal P$ be a model class and consider testing
\begin{align*}
H_0: P\in\mathcal P_0
\qquad \text{versus} \qquad H_1: P\in\mathcal P_1(\rho),
\end{align*}
where $\mathcal P_0\subset\mathcal P$ is a class of null distributions under which the group-invariance hypothesis holds for a given transformation scheme. The alternative class is indexed by a separation parameter $\rho>0$ and takes the generic form
\begin{align*}
\mathcal P_1(\rho)
\coloneqq
\bigl\{P\in\mathcal P:\ d(P^{(0)},P^{(1)})\ge \rho\bigr\},
\end{align*}
where $(P^{(0)},P^{(1)})$ denotes the pair of distributions whose equality characterizes the null, and $d$ is a problem-specific distance or semi-metric.

For a test $\Delta$ based on a sample of size $n$ with error levels $\alpha \in (0,1)$, $\beta \in (0,1-\alpha)$, define the uniform separation radius as
\begin{align*}
\rho(\Delta;\alpha,\beta,n)
\coloneqq
\inf\Bigl\{\rho>0:\ 
\sup_{P\in\mathcal P_0}\mP_P(\Delta=1)\le\alpha,\
\sup_{P\in\mathcal P_1(\rho)}\mP_P(\Delta=0)\le\beta
\Bigr\}.
\end{align*}
Let $\Omega_\alpha$ denote the class of all level-$\alpha$ tests, 
and define the minimax benchmark
\begin{align*}
\rho^\star(\alpha,\beta,n)
\coloneqq
\inf_{\Delta\in\Omega_\alpha}\rho(\Delta;\alpha,\beta,n).
\end{align*}
We say that a test $\Delta$ is minimax optimal if
\(\rho(\Delta;\alpha,\beta,n) \lesssim \rho^\star(\alpha,\beta,n)\), where the inequality holds up to constants independent of $(\alpha,\beta,n)$ but may depend on other problem parameters. This formulation is standard in nonparametric testing; see, e.g., \cite{ingster2012nonparametric,schrab2025unified}.

\parheading{Adaptive aggregation and transfer to SB calibration.}
A standard route to adaptivity aggregates a collection of tests $\{\Delta_k\}_{k\in[K]}$ via Bonferroni correction:
\begin{align*}
\Delta_{\mathrm{Bonf}} \coloneqq \mathds{1}\left\{\min_{k\in[K]} p^{(k)} \le \alpha/K \right\},
\end{align*}
where $p^{(k)}$ is the $p$-value associated with $\Delta_k$. This procedure is valid under arbitrary dependence and satisfies
\begin{align*}
\rho(\Delta_{\mathrm{Bonf}};\alpha,\beta,n)
\;\le\;
\min_{k\in[K]} \rho(\Delta_k;\alpha/K,\beta,n).
\end{align*}
In many minimax problems, the separation radius depends on $\alpha$ only through a slowly varying term, so Bonferroni aggregation over a tuning-parameter grid incurs only a logarithmic adaptivity cost. This strategy underlies adaptive permutation tests such as those in \cite{kim2022minimax,berrett2021optimal,hagrass2024spectral,mun2024minimax,kent2025locally}. Under the same group-invariance scheme, Bonferroni aggregation can be replaced by SB calibration. By \Cref{prop: SB level} and \Cref{prop: minimum single-batch calibration statistic view}, the SB minimum test controls the type~I error non-asymptotically and uniformly dominates the Bonferroni test. Consequently, existing Bonferroni-based adaptive separation bounds transfer directly to the SB-calibrated counterpart. More precisely, denoting by $\Delta_{\mathrm{SB}}$ the SB minimum test, we have
\begin{align*}
\rho(\Delta_{\mathrm{SB}};\alpha,\beta,n)
\;\le\;
\rho(\Delta_{\mathrm{Bonf}};\alpha,\beta,n).
\end{align*}
This inequality is immediate from the finite-sample dominance of SB over Bonferroni calibration and requires no additional analysis beyond that of the existing adaptive constructions.

In \Cref{sec: adaptive MMD and HSIC tests}, we instantiate this principle for adaptive two-sample testing and independence testing using kernel-based statistics, and show that the SB minimum test achieves the same adaptive separation rates as MaxT-based procedures~\citep{albert2022adaptive,schrab2023mmd} while retaining finite-sample validity.

\subsection{Conformal prediction}
\label{sec:conformal-prediction}
Conformal prediction~\citep{vovk2005algorithmic} often yields multiple valid prediction sets arising from different nonconformity scores, feature representations, or training procedures. A natural question is how to merge such sets into a single prediction set that preserves finite-sample coverage while improving efficiency. In this section, we show how the SB and TB aggregation frameworks developed earlier can be used to combine $K$ conformal predictors. The key distinction between the two is computational: SB aggregation yields a fully self-calibrated conformal p-value but requires recalibration for each candidate label, whereas TB aggregation decouples calibration from label evaluation and enables efficient set inversion. Related approaches include validity-preserving selection of a single conformal predictor \citep{yang2024selection} and aggregation or selection rules valid under arbitrary dependence \citep{gasparin2024merging,hegazy2025valid}. Our approach instead uses permutation-based aggregation to exploit the underlying dependence among conformal p-values.

\parheading{Setup.}
To make the discussion concrete, we begin by recalling the standard split conformal setup~\citep{papadopoulos2008inductive,lei2018distribution} with multiple scores. Let $\{(X_i,Y_i)\}_{i=1}^{n}$ be a calibration sample and let $(X_{n+1},Y_{n+1})$ be a new exchangeable observation. For each $k\in[K]$, let $s_k:\mathcal X\times\mathcal Y\to\mathbb R$ be a nonconformity score, and define
\begin{align*}
S_i^k \coloneqq s_k(X_i,Y_i),\qquad i\in[n],
\qquad\text{and}\qquad
S_{n+1}^k(y)\coloneqq s_k(X_{n+1},y),\quad y\in\mathcal Y.
\end{align*}
The associated split conformal p-value is
\begin{align}
\label{eq:conf_pval_pky}
p_k(y) \coloneqq  \frac{1 + \sum_{i=1}^n \mathds{1}\{S_i^k \ge S_{n+1}^k(y)\}}{n+1},
\end{align}
which induces the prediction set
\begin{align*}
C_k(X_{n+1}) \coloneqq \{y\in\mathcal Y:\ p_k(y)>\alpha\}.
\end{align*}
By exchangeability, $p_k(Y_{n+1})$ is super-uniform, and hence $\mP\bigl(Y_{n+1}\in C_k(X_{n+1})\bigr)\ge 1-\alpha$.

Our objective is to merge the collection $\{C_k\}_{k=1}^K$ into a single prediction set that retains finite-sample coverage while exploiting dependence across scores to improve efficiency. To this end, we consider a merging function $f:[0,1]^K\to\mathbb R$, such as the minimum, median, or mean of the individual p-values.

\parheading{SB aggregation.}
For a candidate label $y$, SB aggregation constructs an aggregated conformal p-value as follows. For each score $k$, form an augmented sample by appending the test-point score to the calibration scores and define
\begin{align*}
S_i^k(y)\coloneqq
\begin{cases}
S_i^k, & i\in[n],\\
S_{n+1}^k(y), & i=n+1.
\end{cases}
\end{align*}
Based on this augmented sample, define the column-wise permutation p-values
\begin{equation}\label{eq:cp_SB_column_pvalues}
P_i^k(y)
\coloneqq
\frac{1}{n+1}\sum_{j=1}^{n+1}
\mathds{1}\!\left\{S_j^k(y)\ge S_i^k(y)\right\},
\qquad i\in[n+1],\ k\in[K],
\end{equation}
so that $P_{n+1}^k(y)=p_k(y)$. The standardized p-values are then aggregated across scores in a row-wise manner using the merging function $f$,
\begin{align*}
M_i(y) \coloneqq f\bigl(P_i^1(y),\ldots,P_i^K(y)\bigr),
\qquad i\in[n+1].
\end{align*}
The SB aggregated conformal p-value is obtained by calibrating the merged test-point value against its calibration counterparts,
\begin{align*}
p_{\mathrm{SB}}(y)
\coloneqq
\frac{1}{n+1}\sum_{i=1}^{n+1}
\mathds{1}\!\left\{M_i(y)\le M_{n+1}(y)\right\},
\end{align*}
and the resulting merged conformal prediction set is
\begin{align*}
C_{\mathrm{SB}}(X_{n+1})
\coloneqq
\{y\in\mathcal Y:\ p_{\mathrm{SB}}(y)>\alpha\}.
\end{align*}
When $y=Y_{n+1}$, the augmented score matrix is row-exchangeable. Since the permutation-equivariant 
column-wise p-value map in \eqref{eq:cp_SB_column_pvalues} and the merging
function $f$ are applied row-wise, the merged values are exchangeable, so
$p_{\mathrm{SB}}(Y_{n+1})$ is super-uniform and
\begin{align*}
\mP\bigl(Y_{n+1}\in C_{\mathrm{SB}}(X_{n+1})\bigr)\ge 1-\alpha.
\end{align*}
The SB power comparison in \Cref{thm:single-batch power} further implies that the SB merged conformal set is uniformly no larger than the conformal set obtained by calibrating $f\bigl(p_1(y),\ldots,p_K(y)\bigr)$ using any deterministic worst-case correction.

The computational drawback is that $p_{\mathrm{SB}}(y)$ must typically be
evaluated over many candidate labels $y$. The test-point score
$S_{n+1}^k(y)$ enters the ranking set in \eqref{eq:cp_SB_column_pvalues}, so the
column-wise p-values, the merged values, and the final calibration all change
with $y$. Thus, SB aggregation requires full recalibration for each candidate
label.

\parheading{TB aggregation.}
We now introduce a TB construction whose key feature is that the calibration threshold can be computed once, independently of the candidate label $y$. Partition the calibration indices as $[n]=I_{\mathrm{ref}}\cup I_{\mathrm{agg}}$ with $|I_{\mathrm{ref}}|=n_1$ and $|I_{\mathrm{agg}}|=n_2$ (e.g., via a random split), and retain $X_{n+1}$ as the test feature. Using the reference batch, standardize each score via reference-based p-values. Specifically, for each $k\in[K]$ and $i\in I_{\mathrm{agg}}$, define
\begin{equation}\label{eq:cp_TB_column_pvalues_cal}
P_i^{k,\mathrm{ref}}
\coloneqq
\frac{1 + \sum_{j\in I_{\mathrm{ref}}}\mathds{1}\{S_j^k \ge S_i^k\}}{n_1+1}.
\end{equation}
For a candidate label $y$, define analogously the test-point p-values
\begin{equation}
\label{eq:cp_TB_column_pvalues_test}
P_{n+1}^{k,\mathrm{ref}}(y)
\coloneqq
\frac{1 + \sum_{j\in I_{\mathrm{ref}}}\mathds{1}\{S_j^k \ge S_{n+1}^k(y)\}}{n_1+1}.
\end{equation}
Aggregate across $k$ using the merging function $f$ to obtain
\begin{align*}
M_i \coloneqq f\bigl(P_i^{1,\mathrm{ref}},\ldots,P_i^{K,\mathrm{ref}}\bigr),
\quad i\in I_{\mathrm{agg}},
\qquad
M_{n+1}(y) \coloneqq f\bigl(P_{n+1}^{1,\mathrm{ref}}(y),\ldots,P_{n+1}^{K,\mathrm{ref}}(y)\bigr).
\end{align*}
Calibration compares the merged test-point value to the empirical distribution of the merged calibration values. Define
\begin{equation}\label{eq:cp_TB_threshold}
\begin{aligned}
u_{\alpha}^{\mathrm{TB}}
&\coloneqq
-\,\mathrm{Quantile}_{(1-\alpha)(1+1/n_2)}
\bigl\{-M_i:\ i\in I_{\mathrm{agg}}\bigr\},\\
C_{\mathrm{TB}}(X_{n+1})
&\coloneqq
\{y\in\mathcal Y:\ M_{n+1}(y)\ge u_{\alpha}^{\mathrm{TB}}\}.
\end{aligned}
\end{equation}
The factor $(1+1/n_2)$ accounts for calibrating with the $n_2$ aggregation-batch
values alone, rather than with the augmented set that also includes the test
point.

A key feature of the TB construction is that the resulting threshold $u_{\alpha}^{\mathrm{TB}}$ depends only on $\{M_i\}_{i\in I_{\mathrm{agg}}}$ and is thus independent of $y$. As a result, when many labels are evaluated, the calibration threshold is computed once and reused, and only the test-point quantity $M_{n+1}(y)$ needs to be updated for each candidate label $y$. This decoupling of calibration from label evaluation is the central computational advantage of TB aggregation over the SB construction, which requires full recalibration for each candidate label. 

The finite-sample coverage guarantee follows from the same conditional
exchangeability argument. Given the split, the reference-based p-value maps are
fixed, so the aggregation-batch values and the test-point value at the true
label are treated symmetrically. Thus,
\begin{align*}
\mP\bigl(Y_{n+1}\in C_{\mathrm{TB}}(X_{n+1})\bigr)\ge 1-\alpha.
\end{align*}
\parheading{Example: TB intersection shortcut for residual scores.}
A particularly simple and useful specialization arises in regression with absolute residual scores. Suppose that
\begin{align*}
s_k(x,y)=|y-\widehat{\mu}_k(x)|,
\qquad\text{so that}\qquad S_i^k = |Y_i-\widehat{\mu}_k(X_i)|,
\end{align*}
where $\widehat{\mu}_k:\mathcal X\to\mathbb R$ is an arbitrary prediction function (e.g., a regression estimator) associated with score $k$ and trained on separate data. Consider the minimum merging rule
\begin{align*}
f(p_1,\ldots,p_K)=\min_{k\in[K]} p_k .
\end{align*}
Let $\{R_j^k\}_{j\in I_{\mathrm{ref}}}$ denote the reference residuals $R_j^k\coloneqq |Y_j-\widehat{\mu}_k(X_j)|$, and write $R_{(1)}^k\le \cdots \le R_{(n_1)}^k$ for their order statistics, with the convention $R_{(n_1+1)}^k\coloneqq \infty$. Then the TB set in \eqref{eq:cp_TB_threshold} admits the following closed-form intersection representation.

\begin{proposition}
\label{prop:cp_intersection}
Suppose that $f(p_1,\ldots,p_K)=\min_{k\in[K]}p_k$.
If $u_{\alpha}^{\mathrm{TB}}=-\infty$ then $C_{\mathrm{TB}}(X_{n+1})=\mathbb{R}$, otherwise, letting
$\ell\coloneqq n_1 + 2 - \left\lceil u_{\alpha}^{\mathrm{TB}}(n_1+1)\right\rceil$, it holds that
\begin{align*}
C_{\mathrm{TB}}(X_{n+1})
=
\bigcap_{k=1}^K
\Bigl[
\widehat{\mu}_k(X_{n+1})-R_{(\ell)}^k,\,
\widehat{\mu}_k(X_{n+1})+R_{(\ell)}^k
\Bigr].
\end{align*}
\end{proposition}

Thus, under minimum merging, TB aggregation reduces conformal set inversion to
intersecting $K$ residual-based intervals. Consequently, unlike the SB
construction, \Cref{prop:cp_intersection} yields an explicit intersection form,
so that conformal set inversion reduces to computing fixed residual quantiles
rather than repeatedly recalibrating over candidate labels $y$.

\section{Numerical experiments}
\label{sec: simulations}

In the main text, we evaluate the proposed methods in three representative settings:
two-sample mean shift testing (\Cref{subsec:two-sample-mean-shift}), sequential two-sample nonparametric testing
(\Cref{subsec:sequential-two-sample-nonparametric}), and conformal prediction
(\Cref{subsec:conformal-prediction-experiments}). Additional simulation results,
including empirical level assessment, one-sample zero mean testing, and
independence testing, are deferred to \Cref{sec: additional simulation results}. All
experiments were conducted using a single NVIDIA A100 GPU (40GB VRAM). The code
for reproducing the experiments is available at
\url{https://github.com/antoninschrab/sb-paper}.

\subsection{Two-sample mean shift testing}
\label{subsec:two-sample-mean-shift}

\begin{figure}[!t]
	\centering
	\includegraphics[width=\textwidth]{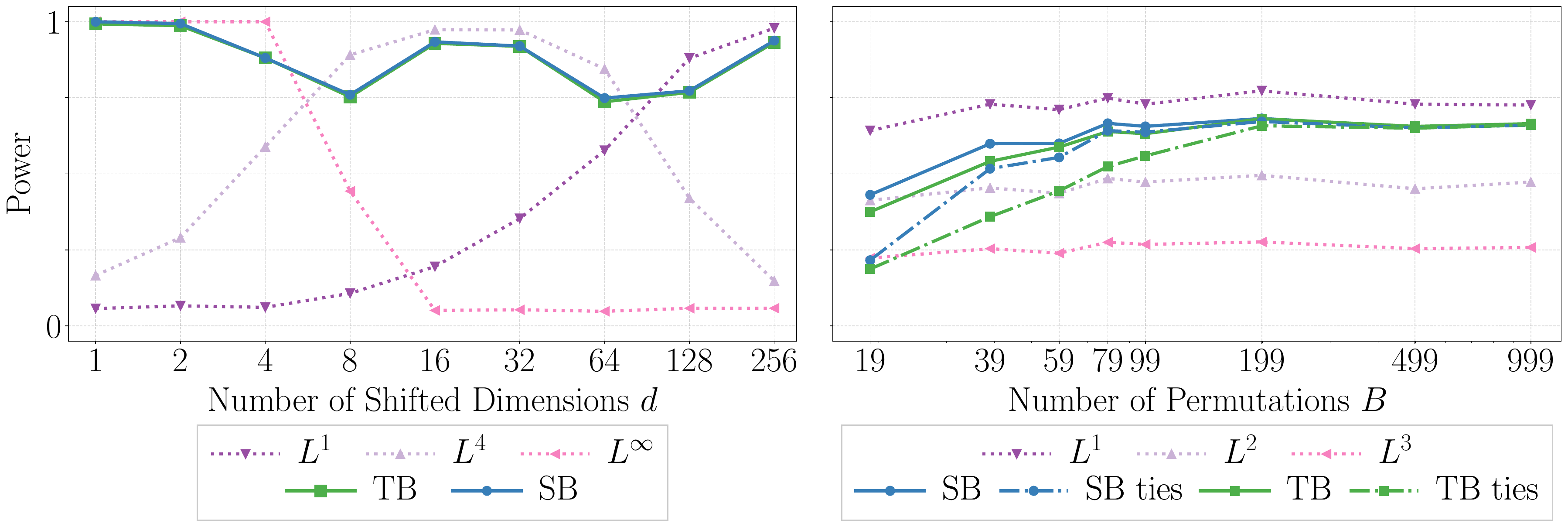}
	\caption{Power estimation for two-sample $d$-dimensional mean shift detection.} 
	\label{fig: figure_2}
\end{figure}

We evaluate the empirical power of the proposed tests using data drawn from a high-dimensional Laplace distribution with a targeted mean shift. 
Specifically, we generate two independent samples of sizes $m$ and $n$ in $\mathbb{R}^D$, where the features are initially drawn from a standard Laplace distribution. 
A signal is then introduced by replacing the values in the first $d$ dimensions of the first sample by $\Delta/2$ and of the second sample by $-\Delta/2$, yielding a true mean difference of $\Delta$ across the $d$ shifted dimensions. 
The empirical power is averaged over 1000 independent repetitions at a nominal significance level of $\alpha = 0.05$. 
For the first experiment varying signal sparsity in \Cref{fig: figure_2}, we fix the sample sizes to $m = n = 100$, the ambient dimension to $D = 20000$, and the number of permutations to $B = 199$. 
We vary the number of shifted dimensions $d \in \{1, 2, 4, 8, 16, 32, 64, 128, 256\}$, decaying the shift magnitude $\Delta$ according to $\Delta = 1.1 / (1 + 0.0318(d - 1)^{0.826})$ to maintain comparable non-trivial power across regimes. 
For the second experiment, we fix $m = n = 1000$, $D = 20000$, $d = 128$, and $\Delta = 0.1$, while varying the number of permutations $B$ from 19 to 999.
For the group transformations, we use permutations of the pooled data to construct two permuted samples.
We consider the classical permutation tests using the $L^p$ norm of the difference in sample means as test statistics, as well as the SB and TB tests aggregating over the same three $L^p$ norm statistics being considered. 
We implement the standard versions of SB and TB using tie-breaking completely at random, and their variants preserving ties (i.e., SB/TB ties).

The left panel of \Cref{fig: figure_2} demonstrates that the individual $L^1$, $L^4$, and $L^\infty$ tests achieve high power only against specific, distinct alternatives (i.e., $L^p$ detects only dense signal for low $p$ and only sparse signals for large $p$). 
In contrast, the combined SB and TB tests are highly adaptive, achieving high empirical power against all considered alternatives. 
This adaptivity only comes at a small cost in test power compared to the single best-performing test for any given regime. 
Furthermore, with $B = 199$ permutations, the SB and TB tests achieve exactly the same power. 
The right panel illustrates the effect of the permutation count $B$ and tie-breaking strategies. 
Employing tie-breaking greatly increases the statistical power, especially for small values of $B$. 
For these small $B$ values, the SB test outperforms the TB test (confirming the result of \Cref{prop:SBvsTB:finiteB_gap_main}); without tie-breaking, this difference in power is quite pronounced, whereas with tie-breaking, SB still outperforms TB, but the gap in power is much less significant. 
For large values of $B$, all adaptive variants (SB and TB, with and without tie-breaking) converge to the same power (confirming the result of \Cref{prop: SB-TB equivalence rate}), again achieving an adaptive performance just below the best individual test. 
As such, we choose $B = 199$ for all our subsequent experiments since we observe that the power remains the same for larger values of $B$ while the computational cost increases linearly with $B$.
We refer to \cite{domingo2025cheap} for a method to significantly reduce the number of permutations required.

\begin{figure}[!t]
	\centering
	\includegraphics[width=\textwidth]{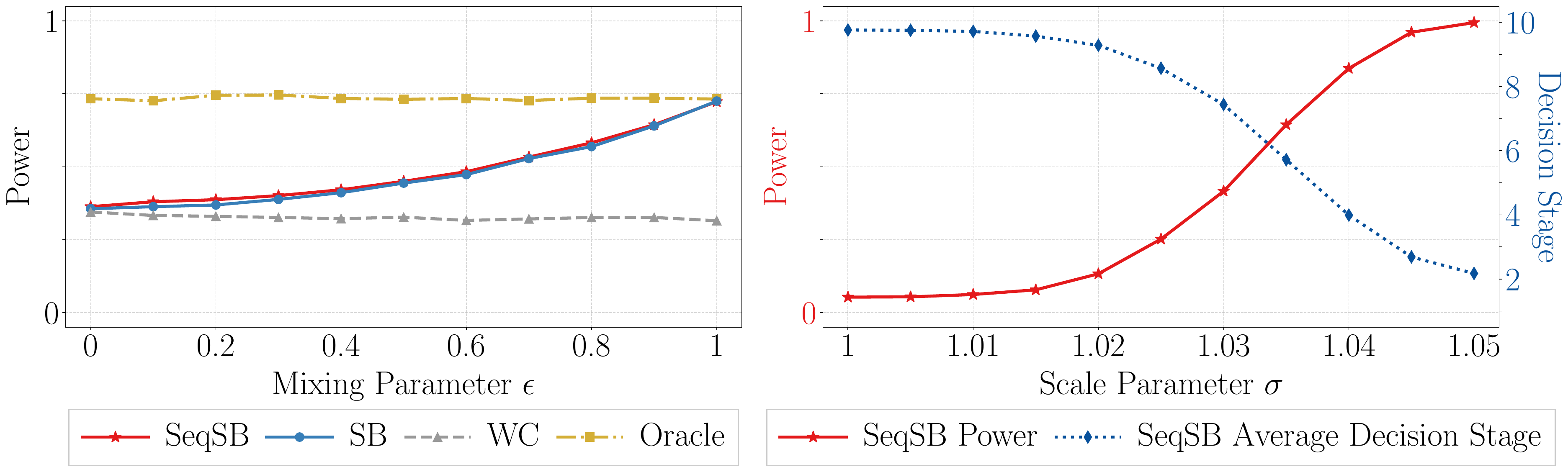}
	\caption{Power estimation for kernel-based MMD two-sample nonparametric testing in a sequential setting.} 
	\label{fig: figure_5}
\end{figure}

\subsection{Sequential two-sample nonparametric testing}
\label{subsec:sequential-two-sample-nonparametric}

We evaluate the empirical performance of the Sequential SB minimum test (SeqSB, see \Cref{alg:SeqSB-min}) using simulated data evaluated across $K=10$ sequential stages with equal spending $\alpha_j=\alpha/K$ for $j\in[K]$, designed to introduce a controlled correlation structure along with a localized departure from the null.
This experimental setting is closely related to the problem of change point detection.
The sequential hypothesis testing framework proceeds as follows: we first receive data $\mathbf{X}$, then at stage $k=1,\dots,K$, we receive data $\mathbf{Y}_k$, and we test the global null that $\mathbf{X}$ and $\mathbf{Y}_k$ are identically distributed for all $k=1,\dots,K$.
More specifically, we set $m=n$ and first draw $n=1000$ samples from $\mathcal{N}(0, I_{50})$ for $\mathbf{X}$, and for each stage $k=1,\dots,K$, we independently draw $n=1000$ samples from $\mathcal{N}(0, I_{50})$ for $\mathbf{Z}_k$.
For some scale parameter $\sigma>0$ ($\sigma=1$ corresponds to the null), we construct $\mathbf{Y}_1 = \mathbf{Z}_1$, then $\mathbf{Y}_2 = \sigma f(\mathbf{Y}_1,\mathbf{Z}_2)$, then $\mathbf{Y}_3 = f(\mathbf{Y}_2/\sigma,\mathbf{Z}_3)$, and $\mathbf{Y}_k = f(\mathbf{Y}_{k-1},\mathbf{Z}_k)$ for stages $k=4,\dots,K$. 
The transition function is defined as $f(\mathbf{Y}, \mathbf{Z}) = \sqrt{1 - (1-\epsilon)^2} \mathbf{Y} + (1-\epsilon) \mathbf{Z}$, ensuring that for any mixing parameter $\epsilon\in[0,1]$, the distribution of $\mathbf{Y}_k$ remains $\mathcal{N}(0, I_{50})$ for all $k\neq 2$ and $\mathbf{Y}_2$ is i.i.d. $\mathcal{N}(0, \sigma^2I_{50})$.
If $\epsilon = 0$, the samples $\mathbf{Y}_1,\dots,\mathbf{Y}_K$ are mutually independent, and as $\epsilon$ increases, the correlation between the samples increases, with $\epsilon=1$ corresponding to perfect correlation across all stages.
We set $\epsilon=0.9$ to introduce a strong correlation structure across the stages in the right panel of \Cref{fig: figure_5}.
At each stage $k=1,\dots,K$, we compute the Maximum Mean Discrepancy U-statistic (MMD, \cite{gretton2012kernel}) between $\mathbf{X}$ and $\mathbf{Y}_k$ using a Laplace kernel with fixed bandwidth.
Using the minimum merging function, we compare the proposed SeqSB test against the non-sequential SB test and to the Worst-Case (WC) baseline which corresponds to Bonferroni correction across the $K$ stages.
As a reference point, we also include the power of the Oracle MMD permutation test between $\mathbf{X}$ and $\mathbf{Y}_2$, which leverages prior knowledge that the departure from the null occurs exactly (and only) at stage $k=2$.
The empirical power and the average decision stage (defined as the mean stage at which the sequential procedure terminates) are averaged over 1000 independent repetitions. 
All tests use $B=199$ wild bootstrap MMD U-statistics, using the same Rademacher variables across all $K$ stages to maintain the dependence structure across the test statistics.
Note that we assume that the total number of stages $K$ is known and fixed (see \citep{shekhar2023nonparametric} for an anytime valid MMD-based test).

The empirical results are summarized in \Cref{fig: figure_5}. The left panel illustrates the effect of the mixing parameter $\epsilon$ under a fixed departure scale of $\sigma = 1.035$. 
Notably, the SeqSB and SB tests achieve identical empirical power, though the SeqSB test provides the critical operational advantage of early stopping. 
When $\epsilon=0$, the SB procedure coincides with the conservative WC Bonferroni baseline. 
However, as the dependence $\epsilon$ increases, the power of the SB tests rises, eventually achieving the exact same power as the Oracle single test as $\epsilon$ tends to $1$. 
This behavior demonstrates that the SB methodology adaptively calibrates to the dependence structure across statistics, maintaining nominal level control without over-correcting for redundant information (\Cref{prop: minimum single-batch calibration statistic view}).
The right panel demonstrates the effect of varying the scale parameter $\sigma$ from $1$ to $1.05$ under strong stage correlation ($\epsilon=0.9$). 
As the signal strength increases, the empirical power of the SeqSB test monotonically approaches 1. 
Concurrently, the average decision stage drops dramatically; for sufficiently strong departures from the null hypothesis, the SeqSB test is capable of rejecting the null as soon as signal is injected (e.g., at stage $k=2$ out of 10 for $\sigma=1.05$), underscoring its efficiency and responsiveness in sequential monitoring environments.

\subsection{Conformal prediction}
\label{subsec:conformal-prediction-experiments}

\begin{figure}[!t]
	\centering
	\includegraphics[width=\textwidth]{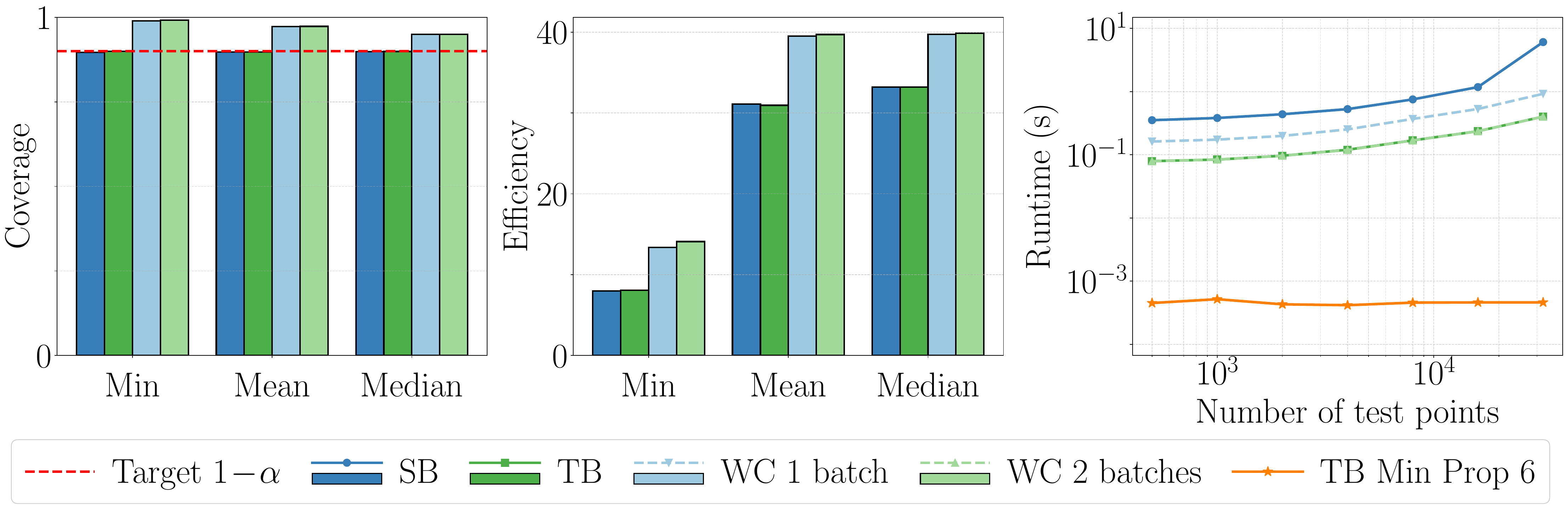}
	\caption{Marginal coverage, efficiency (average prediction set size) and runtimes (in seconds) for merging multiple conformal prediction sets.} 
	\label{fig: figure_6}
\end{figure}

In \Cref{fig: figure_6}, we evaluate the proposed methods within a conformal prediction framework, aiming to construct valid prediction sets for a regression target. 
The objective is to ensure precise marginal coverage (defined as the probability that the true test response falls within the constructed set) while simultaneously minimizing the prediction set size (average interval width). 
Data is generated from a one-dimensional heteroscedastic regression model with covariates $X \sim \text{Uniform}(-2, 2)$ and responses $Y = \sin(3X) + \varepsilon$, with noise $\varepsilon \sim \mathcal{N}(0, (1+|X|))$. 
We evaluate $K=20$ distinct predictors, each approximating the true function but injected with independent additive noise, as well as with noise shared across all $K$ predictors. 
The first predictor has noise of much smaller scale.
Formally, we define $Y^k = \sin(3X) + s_k (0.99\varepsilon_\mathrm{shared} + 0.01\varepsilon_k)$ for independent noise variables $\varepsilon_\mathrm{shared}, \varepsilon_1, \dots, \varepsilon_K \sim \mathcal{N}(0, 1)$ and scaling parameter $s_1 = 0.01$ and $s_k = 10$ for $k=2,\dots,K$.
We use the absolute residuals as non-conformity scores and a target marginal coverage of $1-\alpha=0.90$.
For SB we use $n=999$, while for TB we split these samples into $n_1=500$ for the reference batch and $n_2=499$ for the aggregation batch, ensuring that $\lfloor(n+1)\alpha\rfloor / (n+1) = \lfloor(n_2+1)\alpha\rfloor / (n_2+1) = \alpha$.
We emphasize that our implementation using absolute residual scores is exact in the sense that it does not require a discretization of the response space, avoiding both the computational cost and the approximation errors inherent to grid-based evaluations.
Validity metrics are averaged over 1000 independent repetitions, each considering 100 new test points (overall averaged over 100,000 test points).
We systematically benchmark the Single-Batch (SB) and Two-Batch (TB) methods against Worst-Case (WC 1 batch and WC 2 batches) baselines using minimum, mean, and median merging functions, as detailed below. 
Finally, to assess computational scalability, end-to-end execution times are recorded strictly after Just-In-Time (JIT) compilation, ensuring a fair evaluation of runtime complexity as the number of test points varies (e.g., 500, 1000, 2000, 4000, 8000, 16000, 32000) with fixed $n=500$ and $K=20$.

In this setting, we have $K$ conformal prediction sets $C^k_\alpha(X_{n+1})\coloneqq\{y\in\mathcal Y:\ p_k(y)>\alpha\}$, each constructed such that $\mathbb P(Y_{n+1}\in C^k_\alpha(X_{n+1}))\geq 1-\alpha$, for $k\in[K]$.
The p-values $p_k$ are constructed either using the full data (i.e., single batch, \eqref{eq:conf_pval_pky}) or using a subsample of the data (i.e., two batches, \eqref{eq:cp_TB_column_pvalues_test}).
The aim is then to combine $C^1_\cdot(X_{n+1}),\dots, C^K_\cdot(X_{n+1})$ into a single set $C_\alpha(X_{n+1})$ which still satisfies the marginal coverage guarantee $\mathbb P(Y_{n+1}\in C_\alpha(X_{n+1}))\geq 1-\alpha$.
The SB and TB methods presented in \Cref{sec:conformal-prediction} provide valid constructions for this problem.
Alternatively, one can also consider the Worst-Case (WC) baselines which construct $C^\mathrm{WC}_\alpha(X_{n+1})\coloneqq\{y\in\mathcal Y:\ f(p_1(y), \dots,p_K(y))>\alpha\}$ for specifically designed p-merging functions (\Cref{sec:p-value-merging} and \Cref{sec:p-merging-families}) such as the minimum p-merging function $f(p_1, \dots, p_K) = K\min(p_1, \dots, p_K)$, the mean p-merging function $f(p_1, \dots, p_K) = 2\,\mathrm{mean}(p_1, \dots, p_K)$ and the median p-merging function $f(p_1, \dots, p_K) = 2\,\mathrm{median}(p_1, \dots, p_K)$, with scaling specifically chosen to ensure $1-\alpha$ marginal coverage under arbitrary dependence \citep{vovk2020combining}.

The empirical results of \Cref{fig: figure_6} confirm the theoretical guarantees and highlight the computational superiority of the TB precomputation approach. 
As observed in the first analysis, both the SB and TB procedures tightly control the marginal coverage near the target level $1-\alpha$ for all merging functions. 
In contrast, the WC baselines are markedly conservative. 
For this experimental setting in which $Y^1$ is injected with less noise than the other predictors, the minimum merging methods yield significantly smaller prediction sets than the mean and median methods which are dominated by the noisy predictors.
Furthermore, SB and TB achieve much smaller average prediction set sizes compared to WC, as the aggregation methods leverage the dependence structure across the $K$ predictors.
Crucially, while SB and TB exhibit identical statistical validity and efficiency, their execution times diverge significantly. 
The runtime plot demonstrates that the TB procedure generally executes considerably faster than SB, and that with the expression of \Cref{prop:cp_intersection} designed specially for the minimum merging function, TB is orders of magnitude faster than SB.
This performance gap is directly explained by their time complexities: $\mathcal{O}(K n \log n + M K n \log(Kn))$ for the general SB and TB variants\footnote{This assumes that $n_1$ and $n_2$ are of the same order as $n$, and that the merging function admits an incremental update rule (e.g., the minimum, average, and median functions).}, and $\mathcal{O}(K n \log n + M K)$ for the TB formulation in \Cref{prop:cp_intersection}, where $M$ represents the number of test points.
This highlights the advantages of the TB approach for conformal prediction applications, while the SB procedure is superior for hypothesis testing applications.
We refer the reader to \Cref{sec:alg-conformal} for a detailed implementation of the TB and SB conformal prediction procedures, bypassing the need for discretization of the response space, and of their computational complexities.

\section{Discussion} \label{sec: discussion}
This paper studies row-wise permutation aggregation for statistical evidence under exchangeability. Building on permutation-combination ideas, we characterize the finite-sample power and dependence adaptivity of SB aggregation and extend the framework through sequential spending and two-batch data-dependent aggregation. By operating at the level of transformed data rather than relying solely on worst-case super-uniformity, the proposed methods can exploit the underlying dependence across statistics while retaining finite-sample validity. The applications to adaptive testing and conformal prediction illustrate how this power theory and the TB extension can be used in practice.

Several directions remain open. First, permutation p-values provide a canonical standardization, but for small $B$ they coarsen evidence onto a sparse grid. Continuous standardizations, such as null-CDF or studentized standardizations, may improve power, and their gains remain to be characterized. Second, a general optimality theory for aggregation is missing. Aggregation adapts across heterogeneous alternatives but incurs a calibration cost; identifying when this cost is unavoidable, and proving minimax, oracle-adaptive, or lower-bound guarantees, are central questions. Third, the sequential theory could be extended beyond fixed alpha-spending. SeqSB gives finite-sample valid early stopping for ordered statistics, but leaves open how to choose or learn the spending sequence, handle streaming or adaptively ordered statistics, and characterize the tradeoff between early stopping, power, and calibration cost. Finally, it would be useful to extend the theory beyond exact exchangeability. The current guarantees assume that the transformed rows are exactly exchangeable, but in large or constrained transformation spaces one may only have approximate randomization. Characterizing how aggregation validity and power degrade under approximate exchangeability, and how to correct for this degradation, would substantially broaden the scope of the framework.

\bibliographystyle{apalike}
\bibliography{reference.bib}

\begin{thebibliography}{}

\bibitem[Albert, 2015]{albert2015tests}
Albert, M. (2015).
\newblock {\em Tests of independence by bootstrap and permutation: an
  asymptotic and non-asymptotic study. Application to neurosciences.}
\newblock PhD thesis, Universit{\'e} Nice Sophia Antipolis.

\bibitem[Albert et~al., 2022]{albert2022adaptive}
Albert, M., Laurent, B., Marrel, A., and Meynaoui, A. (2022).
\newblock {Adaptive test of independence based on HSIC measures}.
\newblock {\em The Annals of Statistics}, 50(2):858--879.

\bibitem[Angelopoulos et~al., 2024]{angelopoulos2024theoretical}
Angelopoulos, A.~N., Barber, R.~F., and Bates, S. (2024).
\newblock {Theoretical Foundations of Conformal Prediction}.
\newblock {\em arXiv preprint arXiv:2411.11824}.

\bibitem[Baraud et~al., 2003]{baraud2003adaptive}
Baraud, Y., Huet, S., and Laurent, B. (2003).
\newblock Adaptive tests of linear hypotheses by model selection.
\newblock {\em The Annals of Statistics}, 31(1):225--251.

\bibitem[Berrett et~al., 2021]{berrett2021optimal}
Berrett, T.~B., Kontoyiannis, I., and Samworth, R.~J. (2021).
\newblock {Optimal rates for independence testing via U-statistic permutation
  tests}.
\newblock {\em The Annals of Statistics}, 49(5):2457--2490.

\bibitem[Berrett et~al., 2020]{berrett2020conditional}
Berrett, T.~B., Wang, Y., Barber, R.~F., and Samworth, R.~J. (2020).
\newblock {The conditional permutation test for independence while controlling
  for confounders}.
\newblock {\em Journal of the Royal Statistical Society Series B: Statistical
  Methodology}, 82(1):175--197.

\bibitem[Biggs et~al., 2023]{biggs2023mmdfuse}
Biggs, F., Schrab, A., and Gretton, A. (2023).
\newblock {MMD-FUSE}: {L}earning and combining kernels for two-sample testing
  without data splitting.
\newblock {\em Advances in Neural Information Processing Systems}, 36.

\bibitem[Candes et~al., 2018]{candes2018panning}
Candes, E., Fan, Y., Janson, L., and Lv, J. (2018).
\newblock Panning for gold: {‘Model-X’} knockoffs for high dimensional
  controlled variable selection.
\newblock {\em Journal of the Royal Statistical Society Series B: Statistical
  Methodology}, 80(3):551--577.

\bibitem[Caughey et~al., 2017]{caughey2017nonparametric}
Caughey, D., Dafoe, A., and Seawright, J. (2017).
\newblock {Nonparametric combination (NPC): A framework for testing elaborate
  theories}.
\newblock {\em The Journal of Politics}, 79(2):688--701.

\bibitem[Cha et~al., 2026]{cha2026more}
Cha, S., Lee, S., Schrab, A., and Kim, I. (2026).
\newblock {More Permutations Do Not Always Increase Power: Non-monotonicity in
  Monte Carlo Permutation Tests}.
\newblock {\em arXiv preprint arXiv:2605.03886}.

\bibitem[Chau et~al., 2025]{chau2024credal}
Chau, S.~L., Schrab, A., Gretton, A., Sejdinovic, D., and Muandet, K. (2025).
\newblock Credal two-sample tests of epistemic uncertainty.
\newblock In {\em Proceedings of The 28th International Conference on
  Artificial Intelligence and Statistics}, volume 258 of {\em Proceedings of
  Machine Learning Research}, pages 127--135. PMLR.

\bibitem[Choi and Kim, 2023]{choi2023averaging}
Choi, W. and Kim, I. (2023).
\newblock Averaging p-values under exchangeability.
\newblock {\em Statistics \& Probability Letters}, 194:109748.

\bibitem[Chung and Romano, 2013]{chung2013}
Chung, E. and Romano, J.~P. (2013).
\newblock Exact and asymptotically robust permutation tests.
\newblock {\em The Annals of Statistics}, 41(2):484--507.

\bibitem[Cox, 1975]{cox1975note}
Cox, D.~R. (1975).
\newblock A note on data-splitting for the evaluation of significance levels.
\newblock {\em Biometrika}, 62(2):441--444.

\bibitem[Domingo-Enrich et~al., 2025]{domingo2025cheap}
Domingo-Enrich, C., Dwivedi, R., and Mackey, L. (2025).
\newblock Cheap permutation testing.
\newblock {\em arXiv preprint arXiv:2502.07672}.

\bibitem[Fisher, 1925]{fisher1925statistical}
Fisher, R.~A. (1925).
\newblock {\em Statistical Methods for Research Workers}.
\newblock Oliver and Boyd, Edinburgh.

\bibitem[Fisher, 1935]{fisher1935design}
Fisher, R.~A. (1935).
\newblock {\em {The Design of Experiments}}.
\newblock Oliver and Boyd, Edinburgh.

\bibitem[Fromont and Laurent, 2006]{fromont2006}
Fromont, M. and Laurent, B. (2006).
\newblock Adaptive goodness-of-fit tests in a density model.
\newblock {\em The Annals of Statistics}, 34(2):680--720.

\bibitem[Fromont et~al., 2013]{fromont2013}
Fromont, M., Laurent, B., and Reynaud-Bouret, P. (2013).
\newblock The two-sample problem for poisson processes: Adaptive tests with a
  nonasymptotic wild bootstrap approach.
\newblock {\em The Annals of Statistics}, 41(3):1431--1461.

\bibitem[Gasparin and Ramdas, 2024]{gasparin2024merging}
Gasparin, M. and Ramdas, A. (2024).
\newblock Merging uncertainty sets via majority vote.
\newblock {\em arXiv preprint arXiv:2401.09379}.

\bibitem[Gasparin et~al., 2025]{gasparin2025combining}
Gasparin, M., Wang, R., and Ramdas, A. (2025).
\newblock Combining exchangeable p-values.
\newblock {\em Proceedings of the National Academy of Sciences},
  122(11):e2410849122.

\bibitem[Good, 2005]{good2005permutation}
Good, P. (2005).
\newblock {\em Permutation, parametric and bootstrap tests of hypotheses}.
\newblock Springer.

\bibitem[Gretton, 2015]{gretton2015simpler}
Gretton, A. (2015).
\newblock {A simpler condition for consistency of a kernel independence test}.
\newblock {\em arXiv preprint arXiv:1501.06103}.

\bibitem[Gretton et~al., 2012]{gretton2012kernel}
Gretton, A., Borgwardt, K.~M., Rasch, M.~J., Sch{{\"o}}lkopf, B., and Smola, A.
  (2012).
\newblock A kernel two-sample test.
\newblock {\em Journal of Machine Learning Research}, 13(25):723--773.

\bibitem[Gretton et~al., 2005]{gretton2005kernel}
Gretton, A., Herbrich, R., Smola, A., Bousquet, O., and Sch{\"o}lkopf, B.
  (2005).
\newblock {Kernel methods for measuring independence}.
\newblock {\em Journal of Machine Learning Research}, 6:2075--2129.

\bibitem[Guo and Shah, 2025]{guo2024rank}
Guo, F.~R. and Shah, R.~D. (2025).
\newblock Rank-transformed subsampling: inference for multiple data splitting
  and exchangeable p-values.
\newblock {\em Journal of the Royal Statistical Society Series B: Statistical
  Methodology}, 87(1):256--286.

\bibitem[Hagrass et~al., 2024]{hagrass2024spectral}
Hagrass, O., Sriperumbudur, B., and Li, B. (2024).
\newblock Spectral regularized kernel two-sample tests.
\newblock {\em The Annals of Statistics}, 52(3):1076--1101.

\bibitem[Harrison, 2012]{harrison2012conservative}
Harrison, M.~T. (2012).
\newblock Conservative hypothesis tests and confidence intervals using
  importance sampling.
\newblock {\em Biometrika}, 99(1):57--69.

\bibitem[Hegazy et~al., 2025]{hegazy2025valid}
Hegazy, M., Aolaritei, L., Jordan, M.~I., and Dieuleveut, A. (2025).
\newblock Valid selection among conformal sets.
\newblock In {\em Advances in Neural Information Processing Systems},
  volume~38.

\bibitem[Hemerik and Goeman, 2018]{hemerik2018exact}
Hemerik, J. and Goeman, J. (2018).
\newblock Exact testing with random permutations.
\newblock {\em Test}, 27(4):811--825.

\bibitem[Ingster and Suslina, 2012]{ingster2012nonparametric}
Ingster, Y. and Suslina, I.~A. (2012).
\newblock {\em Nonparametric goodness-of-fit testing under Gaussian models},
  volume 169.
\newblock Springer Science \& Business Media.

\bibitem[Jankov{\'a} et~al., 2020]{jankova2020goodness}
Jankov{\'a}, J., Shah, R.~D., B{\"u}hlmann, P., and Samworth, R.~J. (2020).
\newblock Goodness-of-fit testing in high dimensional generalized linear
  models.
\newblock {\em Journal of the Royal Statistical Society Series B: Statistical
  Methodology}, 82(3):773--795.

\bibitem[Kent et~al., 2026]{kent2025locally}
Kent, A., Berrett, T.~B., and Yu, Y. (2026).
\newblock {Locally Differentially Private Two-Sample Testing}.
\newblock {\em Biometrika}, page asag034.

\bibitem[Kim et~al., 2022]{kim2022minimax}
Kim, I., Balakrishnan, S., and Wasserman, L. (2022).
\newblock Minimax optimality of permutation tests.
\newblock {\em The Annals of Statistics}, 50(1):225--251.

\bibitem[Kim et~al., 2024]{kim2024conditional}
Kim, I., Neykov, M., Balakrishnan, S., and Wasserman, L. (2024).
\newblock {Conditional independence testing for discrete distributions: Beyond
  $\chi^2$- and G-tests}.
\newblock {\em Electronic Journal of Statistics}, 18(2):4767--4794.

\bibitem[Kim and Ramdas, 2024]{kim2024dimension}
Kim, I. and Ramdas, A. (2024).
\newblock {Dimension-agnostic inference using cross U-statistics}.
\newblock {\em Bernoulli}, 30(1):683--711.

\bibitem[Kim et~al., 2021]{kim2021classification}
Kim, I., Ramdas, A., Singh, A., and Wasserman, L. (2021).
\newblock Classification accuracy as a proxy for two-sample testing.
\newblock {\em The Annals of Statistics}, 49(1):411--434.

\bibitem[Kim and Schrab, 2026]{kim2025differentially}
Kim, I. and Schrab, A. (2026).
\newblock {Differentially Private Permutation Tests}.
\newblock {\em Journal of the American Statistical Association}, pages 1--13.

\bibitem[Lehmann and Romano, 2022]{lehmann2022}
Lehmann, E. and Romano, J.~P. (2022).
\newblock {\em {Testing Statistical Hypotheses}}.
\newblock Springer Texts in Statistics. Springer, 4th edition.

\bibitem[Lei et~al., 2018]{lei2018distribution}
Lei, J., G’Sell, M., Rinaldo, A., Tibshirani, R.~J., and Wasserman, L.
  (2018).
\newblock Distribution-free predictive inference for regression.
\newblock {\em Journal of the American Statistical Association},
  113(523):1094--1111.

\bibitem[Liu et~al., 2020]{liu2020learning}
Liu, F., Xu, W., Lu, J., Zhang, G., Gretton, A., and Sutherland, D.~J. (2020).
\newblock Learning deep kernels for non-parametric two-sample tests.
\newblock In {\em International Conference on Machine Learning}, pages
  6316--6326.

\bibitem[Lundborg et~al., 2024]{lundborg2024projected}
Lundborg, A.~R., Kim, I., Shah, R.~D., and Samworth, R.~J. (2024).
\newblock The projected covariance measure for assumption-lean variable
  significance testing.
\newblock {\em The Annals of Statistics}, 52(6):2851--2878.

\bibitem[Meinshausen et~al., 2011]{meinshausen2011asymptotic}
Meinshausen, N., Maathuis, M.~H., and B{\"u}hlmann, P. (2011).
\newblock Asymptotic optimality of the westfall–young permutation procedure
  for multiple testing under dependence.
\newblock {\em The Annals of Statistics}, 39(6):3369--3391.

\bibitem[Meng, 1994]{meng1994posterior}
Meng, X.-L. (1994).
\newblock {Posterior predictive $p$-values}.
\newblock {\em The Annals of Statistics}, 22(3):1142--1160.

\bibitem[Moran, 1973]{moran1973dividing}
Moran, P.~A. (1973).
\newblock Dividing a sample into two parts a statistical dilemma.
\newblock {\em Sankhy{\=a}: The Indian Journal of Statistics, Series A}, pages
  329--333.

\bibitem[Mun et~al., 2025]{mun2024minimax}
Mun, J., Kwak, S., and Kim, I. (2025).
\newblock Minimax optimal two-sample testing under local differential privacy.
\newblock {\em Journal of Machine Learning Research}, 26(252):1--79.

\bibitem[Paik et~al., 2025]{paik2025integral}
Paik, S., Celentano, M., Green, A., and Tibshirani, R.~J. (2025).
\newblock {Integral Probability Metrics Meet Neural Networks: The
  Radon-Kolmogorov-Smirnov Test}.
\newblock {\em Journal of Machine Learning Research}, 26(86):1--57.

\bibitem[Papadopoulos, 2008]{papadopoulos2008inductive}
Papadopoulos, H. (2008).
\newblock {\em {Inductive conformal prediction: Theory and application to
  neural networks}}.
\newblock INTECH Open Access Publisher Rijeka.

\bibitem[Pesarin and Salmaso, 2010]{pesarinsalmaso2010}
Pesarin, F. and Salmaso, L. (2010).
\newblock {\em Permutation Tests for Complex Data: Theory, Applications and
  Software}.
\newblock Wiley Series in Probability and Statistics. John Wiley \& Sons.

\bibitem[Pitman, 1937]{pitman1937significance}
Pitman, E.~J. (1937).
\newblock Significance tests which may be applied to samples from any
  populations.
\newblock {\em Supplement to the Journal of the Royal Statistical Society},
  4(1):119--130.

\bibitem[Pogodin et~al., 2024]{pogodin2024practical}
Pogodin, R., Schrab, A., Li, Y., Sutherland, D.~J., and Gretton, A. (2024).
\newblock {Practical Kernel Tests of Conditional Independence}.
\newblock {\em arXiv preprint arXiv:2402.13196}.

\bibitem[Ramdas et~al., 2023]{ramdas2023permutation}
Ramdas, A., Barber, R.~F., Cand{\`e}s, E.~J., and Tibshirani, R.~J. (2023).
\newblock Permutation tests using arbitrary permutation distributions.
\newblock {\em Sankhya A}, 85(2):1156--1177.

\bibitem[Ramdas and Wang, 2025]{ramdas2025hypothesis}
Ramdas, A. and Wang, R. (2025).
\newblock {Hypothesis testing with E-values}.
\newblock {\em Foundations and Trends{\textregistered} in Statistics},
  1(1-2):1--390.

\bibitem[Ribero et~al., 2026]{ribero2026regularized}
Ribero, M., Schrab, A., and Gretton, A. (2026).
\newblock Regularized $f$-divergence kernel tests.
\newblock In {\em The 29th International Conference on Artificial Intelligence
  and Statistics}.

\bibitem[Romano and Wolf, 2005]{romano2005exact}
Romano, J.~P. and Wolf, M. (2005).
\newblock Exact and approximate stepdown methods for multiple hypothesis
  testing.
\newblock {\em Journal of the American Statistical Association},
  100(469):94--108.

\bibitem[R{\"u}schendorf, 1982]{rueschendorf1982}
R{\"u}schendorf, L. (1982).
\newblock {Random Variables with Maximum Sums}.
\newblock {\em Advances in Applied Probability}, 14(3):623--632.

\bibitem[Rüger, 1978]{ruger1978}
Rüger, B. (1978).
\newblock {Das maximale Signifikanzniveau des Tests: ,,Lehne $H_0$ ab, wenn $k$
  unter $n$ gegebenen Tests zur Ablehnung führen."}.
\newblock {\em Metrika}, 25:171--178.

\bibitem[Schrab, 2025a]{schrab2025practical}
Schrab, A. (2025a).
\newblock {A practical introduction to kernel discrepancies: MMD, HSIC \& KSD}.
\newblock {\em arXiv preprint arXiv:2503.04820}.

\bibitem[Schrab, 2025b]{schrab2025optimal}
Schrab, A. (2025b).
\newblock {\em {Optimal Kernel Hypothesis Testing}}.
\newblock PhD thesis, UCL (University College London).

\bibitem[Schrab, 2025c]{schrab2025unified}
Schrab, A. (2025c).
\newblock A unified view of optimal kernel hypothesis testing.
\newblock {\em arXiv preprint arXiv:2503.07084}.

\bibitem[Schrab et~al., 2022a]{schrab2022ksd}
Schrab, A., Guedj, B., and Gretton, A. (2022a).
\newblock {KSD Aggregated Goodness-of-fit Test}.
\newblock In {\em Advances in Neural Information Processing Systems 35: Annual
  Conference on Neural Information Processing Systems 2022, NeurIPS 2022}.

\bibitem[Schrab et~al., 2023]{schrab2023mmd}
Schrab, A., Kim, I., Albert, M., Laurent, B., Guedj, B., and Gretton, A.
  (2023).
\newblock {MMD aggregated two-sample test}.
\newblock {\em Journal of Machine Learning Research}, 24(194):1--81.

\bibitem[Schrab et~al., 2022b]{schrab2022efficient}
Schrab, A., Kim, I., Guedj, B., and Gretton, A. (2022b).
\newblock {Efficient aggregated kernel tests using incomplete $ U
  $-statistics}.
\newblock {\em Advances in Neural Information Processing Systems},
  35:18793--18807.

\bibitem[Shah and B{\"u}hlmann, 2018]{shah2018goodness}
Shah, R.~D. and B{\"u}hlmann, P. (2018).
\newblock Goodness-of-fit tests for high dimensional linear models.
\newblock {\em Journal of the Royal Statistical Society Series B: Statistical
  Methodology}, 80(1):113--135.

\bibitem[Shekhar et~al., 2022]{shekhar2022permutation}
Shekhar, S., Kim, I., and Ramdas, A. (2022).
\newblock A permutation-free kernel two-sample test.
\newblock {\em Advances in Neural Information Processing Systems},
  35:18168--18180.

\bibitem[Shekhar et~al., 2023]{shekhar2023permutation}
Shekhar, S., Kim, I., and Ramdas, A. (2023).
\newblock A permutation-free kernel independence test.
\newblock {\em Journal of Machine Learning Research}, 24(369):1--68.

\bibitem[Shekhar and Ramdas, 2024]{shekhar2023nonparametric}
Shekhar, S. and Ramdas, A. (2024).
\newblock Nonparametric two-sample testing by betting.
\newblock {\em IEEE Transactions on Information Theory}, 70(2):1178--1203.

\bibitem[Solmi and Onghena, 2014]{solmi2014combining}
Solmi, F. and Onghena, P. (2014).
\newblock Combining p-values in replicated single-case experiments with
  multivariate outcome.
\newblock {\em Neuropsychological Rehabilitation}, 24(3-4):607--633.

\bibitem[Stouffer et~al., 1949]{stouffer1949american}
Stouffer, S.~A., Suchman, E.~A., DeVinney, L.~C., Star, S.~A., and Williams~Jr,
  R.~M. (1949).
\newblock {\em {The American Soldier: Adjustment during Army Life (Vol. 1)}}.
\newblock Princeton University Press, Princeton, NJ.

\bibitem[Tansey et~al., 2022]{tansey2022holdout}
Tansey, W., Veitch, V., Zhang, H., Rabadan, R., and Blei, D.~M. (2022).
\newblock The holdout randomization test for feature selection in black box
  models.
\newblock {\em Journal of Computational and Graphical Statistics},
  31(1):151--162.

\bibitem[Vovk et~al., 2005]{vovk2005algorithmic}
Vovk, V., Gammerman, A., and Shafer, G. (2005).
\newblock {\em {Algorithmic Learning in a Random World}}.
\newblock Springer.

\bibitem[Vovk et~al., 2022]{vovk2022admissible}
Vovk, V., Wang, B., and Wang, R. (2022).
\newblock Admissible ways of merging p-values under arbitrary dependence.
\newblock {\em The Annals of Statistics}, 50(1):351--375.

\bibitem[Vovk and Wang, 2020]{vovk2020combining}
Vovk, V. and Wang, R. (2020).
\newblock Combining p-values via averaging.
\newblock {\em Biometrika}, 107(4):791--808.

\bibitem[Vovk and Wang, 2021]{vovk2021values}
Vovk, V. and Wang, R. (2021).
\newblock {E-values: Calibration, combination, and applications}.
\newblock {\em The Annals of Statistics}, 49(3):1736--1754.

\bibitem[Westfall and Young, 1993]{westfall1993resampling}
Westfall, P.~H. and Young, S.~S. (1993).
\newblock {\em {Resampling-based multiple testing: Examples and methods for
  p-value adjustment}}.
\newblock John Wiley \& Sons.

\bibitem[Yang and Kuchibhotla, 2025]{yang2024selection}
Yang, Y. and Kuchibhotla, A.~K. (2025).
\newblock {Selection and Aggregation of Conformal Prediction Sets}.
\newblock {\em Journal of the American Statistical Association},
  120(549):435--447.

\bibitem[Zhou et~al., 2025]{zhou2026dual}
Zhou, Z., Tian, X., Peng, L., Lei, C., Schrab, A., Sutherland, D.~J., and Liu,
  F. (2025).
\newblock Dual: Learning diverse kernels for aggregated two-sample and
  independence testing.
\newblock In {\em Advances in Neural Information Processing Systems},
  volume~38.

\end{thebibliography}

\clearpage

\makeatletter
\renewcommand{\hyper@natlinkstart}[1]{%
  \Hy@backout{#1}%
  \hyper@linkstart{cite}{supp.cite.#1}%
  \def\hyper@nat@current{#1}%
}
\renewcommand{\hyper@natanchorstart}[1]{%
  \Hy@raisedlink{\hyper@anchorstart{supp.cite.#1}}%
}
\makeatother

\begin{appendix}

\crefalias{section}{appendix}
\crefalias{subsection}{appendix}

\setcounter{lemma}{0}
\setcounter{proposition}{0}
\setcounter{theorem}{0}
\setcounter{corollary}{0}
\setcounter{definition}{0}
\setcounter{example}{0}
\setcounter{assumption}{0}
\setcounter{remark}{0}

\renewcommand{\thelemma}{S.\arabic{lemma}}
\renewcommand{\theHlemma}{appendix.\arabic{lemma}}

\renewcommand{\theproposition}{S.\arabic{proposition}}
\renewcommand{\theHproposition}{appendix.\arabic{proposition}}

\renewcommand{\thetheorem}{S.\arabic{theorem}}
\renewcommand{\theHtheorem}{appendix.\arabic{theorem}}

\renewcommand{\thecorollary}{S.\arabic{corollary}}
\renewcommand{\theHcorollary}{appendix.\arabic{corollary}}

\renewcommand{\thedefinition}{S.\arabic{definition}}
\renewcommand{\theexample}{S.\arabic{example}}
\renewcommand{\theassumption}{S.\arabic{assumption}}
\renewcommand{\theremark}{S.\arabic{remark}}

\begin{center}
\emph{\LARGE Supplementary material for \\[.2em] Aggregation of Statistical Evidence under Exchangeability} 
\end{center}

\begin{figure}[H]
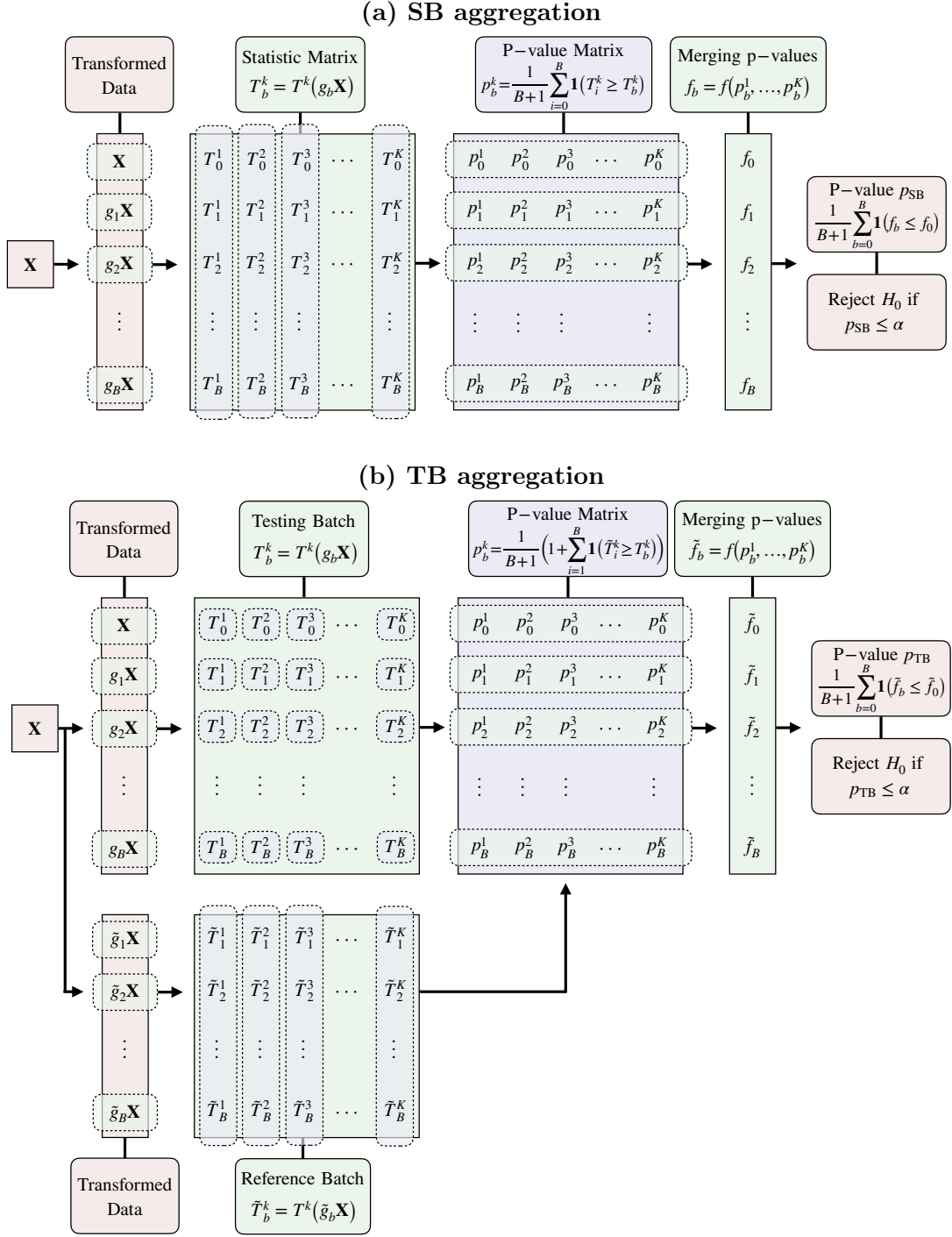

	\centering
	\textbf{(a) SB aggregation}
	\par\smallskip
	\includegraphics[width=0.86\textwidth,height=0.34\textheight,keepaspectratio]{figures/diagram_SB.pdf}
	\vspace{1.25em}
	
	\textbf{(b) TB aggregation}
	\par\smallskip
	\noindent
	\includegraphics[width=0.86\textwidth,keepaspectratio]{figures/diagram_TB.pdf}
	\caption{Schematic illustrations of (a) the SB aggregation procedure of \Cref{alg:SB} and (b) the TB aggregation procedure of \Cref{alg:TB}, using statistics $T^1,\dots,T^K$ and merging function $f$.}
	\label{fig:diagram-SB}
	\label{fig:diagram-TB}
\end{figure}

The supplementary material contains schematic illustrations (\Cref{sec: diagrams}),
background and methodological details (\Cref{sec: deferred main material}),
additional simulation results (\Cref{sec: additional simulation results}),
extensions of the aggregation framework (\Cref{sec: additional theoretical extensions}),
power refinements (\Cref{sec: additional results}), adaptive testing details
(\Cref{sec: adaptive testing details}), technical lemmas and auxiliary tools
(\Cref{sec: useful facts on quantiles}), and proofs
(\Cref{sec: proofs of main results,sec: proofs of appendices}).

\section{Schematic illustrations}
\label{sec: diagrams}

This section presents schematic illustrations of the SB and TB aggregation procedures, which can be found in \Cref{fig:diagram-SB}.

\section{Background and methodological details}
\label{sec: deferred main material}

This section collects background material, methodological reformulations, and auxiliary examples used to situate the proposed procedures.

\subsection{Classical p-value merging families}
\label{sec:p-merging-families}

This subsection recalls the two classical p-merging families referenced in \Cref{sec:merging}. Let $p_1,\ldots,p_K$ denote $K$ super-uniform p-values.

\parheading{O-family.}
R\"{u}ger's O-family \citep{ruger1978} is based on scaled order statistics and is defined as
\begin{align*}
f(p_1,\ldots,p_K)
= \min\!\left\{\frac{K}{k}\, p_{(k)},\, 1\right\},
\end{align*}
where $p_{(k)}$ denotes the $k$-th smallest p-value among $(p_1,\ldots,p_K)$. This family includes several classical procedures as special cases, such as the Bonferroni method ($k=1$), the median rule ($k=\lceil K/2\rceil$), and the maximum p-value rule ($k=K$). All members of the O-family are precise p-merging functions, in the sense that their super-uniformity property is tight.

\parheading{M-family.}
The M-family \citep{vovk2020combining} is parameterized by $r \in [-\infty,\infty]$ and is defined as
\begin{align*}
f(p_1,\ldots,p_K)
= \min\!\left\{ a_{r,K}\, M_{r,K}(p_1,\ldots,p_K),\, 1 \right\},
\end{align*}
where $M_{r,K}(p_1,\ldots,p_K)$ denotes the generalized mean
\begin{align*}
M_{r,K}(p_1,\ldots,p_K)
= \left( \frac{1}{K} \sum_{k=1}^K p_k^r \right)^{1/r}.
\end{align*}
This family encompasses a wide range of commonly used aggregation rules, including the minimum $(r=-\infty)$, maximum $(r=\infty)$, harmonic $(r=-1)$, geometric $(r=0)$, and arithmetic $(r=1)$ means. The scaling constant $a_{r,K}$ is chosen to ensure that $f$ is a precise p-merging function; see, for example, \citep[][Table~1]{vovk2020combining}. A notable special case is the arithmetic mean ($r=1$), for which the optimal scaling constant is $a_{1,K}=2$, a result that dates back to classical work~\citep{rueschendorf1982,meng1994posterior}. More generally, the asymptotically sharp scaling constant for the M-family is $(r+1)^{1/r}$ for $r>0$, $e$ for $r=0$, and $\frac{r}{r+1} K^{1+1/r}$ for $r < -1$; see, for example, \citep[][Table~1]{vovk2020combining}.

Recent work by \citep{vovk2022admissible} further investigates the structure of \emph{admissible} p-merging functions, namely those that cannot be uniformly improved while preserving universal validity under arbitrary dependence. Their results provide a complete characterization of admissibility within the O- and M-families. These procedures require deterministic worst-case calibration under arbitrary dependence. 

Under the group-invariance hypothesis, our main result (\Cref{corollary: single-batch power O and M}) shows that SB aggregation uniformly dominates deterministically calibrated p-merging rules, including the O- and M-families.

\subsection{MaxT p-value formulation and comparison with SB and TB}
\label{sec: deferred maxT comparison}

This subsection provides additional details on the MaxT aggregation procedure reviewed in \Cref{sec: maxT aggregation}, including a closed-form calibration threshold, a p-value formulation, and a comparison with the SB and TB procedures.

\parheading{Closed-form MaxT threshold.}
The estimated critical value $\tilde{u}_\alpha$ in \eqref{eq: estimated u alpha} can be written in closed form. This avoids the bisection procedure employed in prior work~\citep{schrab2023mmd,albert2022adaptive} and reveals the p-value structure of the MaxT calibration.

\begin{proposition}
\label{prop: u_alpha expression}
For each $b \in [B]$, define
\begin{align*}
	u_b \coloneqq \min_{k \in [K]} \frac{1}{B+1} \sum_{j=0}^B \mathds{1}(T_{B+b}^k \leq T_j^k).
\end{align*}	
Then, for $\alpha \in (0,1)$, the estimated critical value $\tilde{u}_\alpha$ defined in \eqref{eq: estimated u alpha} admits the representation
\begin{align*}
	\tilde{u}_\alpha = u_{(\lfloor{B\alpha\rfloor} + 1)},
\end{align*}
where $u_{(1)} \leq u_{(2)} \leq \cdots \leq u_{(B)}$ are the order statistics of $\{u_1,u_2,\ldots,u_B\}$.
\begin{proof}
	See \Cref{sec: proof of u_alpha expression}.
\end{proof}
\end{proposition}

The quantity $u_b$ is the minimum, over coordinates $k$, of the permutation p-value that the calibration statistic $T_{B+b}^k$ would receive when ranked against the testing batch $T_0^k,\ldots,T_B^k$. Thus, $\tilde{u}_\alpha$ is the empirical $\alpha$-quantile of these calibration-batch minimum p-values. This gives the following equivalent p-value view of the MaxT decision rule.

\parheading{P-value formulation of MaxT.}
Let
\begin{align*}
	p(T_{0}^k)
	\coloneqq
	\frac{1}{B+1}
	\sum_{b=0}^B
	\mathds{1}(T_{b}^k \geq T_{0}^k)
\end{align*}
denote the permutation p-value associated with $T_0^k$. Using the equivalent formulations explained in \Cref{sec: permutation test}, the decision rule in \eqref{Eq. max-type test} can be written as
\begin{align*}
	\min_{k \in [K]} p(T_0^k) \leq \tilde{u}_\alpha.
\end{align*}
Hence the MaxT procedure can be interpreted as a minimum p-value test with a Monte Carlo-calibrated correction factor. This viewpoint makes the comparison with SB and TB minimum aggregation transparent: all three procedures aggregate through the minimum p-value, but they differ in how the calibration set is constructed.

\parheading{Comparison with SB minimum aggregation.}
By \Cref{prop: single-batch equivalence} and \Cref{Lemma: permutation p-value}, the SB minimum threshold admits the following equivalent representation. Write
\begin{align*}
Q_k^{\mathrm{SB}}(u)
\coloneqq
\mathrm{Quantile}_{1-u}\bigl\{T_0^k,\ldots,T_B^k\bigr\}.
\end{align*}
Then
\begin{align*}
\hat u_{\alpha,\min}^{\mathrm{SB}}
=
\sup\biggl\{
u\in(0,1):
\frac{1}{B+1}\sum_{b=0}^B
\mathds{1}\biggl(
\max_{k\in[K]}
\Bigl\{
T_b^k-Q_k^{\mathrm{SB}}(u)
\Bigr\}
>0
\biggr)
\le \alpha
\biggr\}.
\end{align*}
This representation makes explicit the key distinction between SB minimum aggregation and the existing MaxT aggregation reviewed in \Cref{sec: maxT aggregation}: the SB threshold $\hat u_{\alpha,\min}^{\mathrm{SB}}$ relies on a single batch of transformed statistics $T_0^k,\ldots,T_B^k$ both to approximate the conditional rejection probability and to compute the relevant quantiles. In contrast, MaxT procedures employ an additional batch $T_{B+1}^k,\ldots,T_{2B}^k$ solely for calibration. This modification resolves the type~I error control issue identified in \Cref{sec: maxT aggregation}, while remaining data-dependent and adaptive to the dependence structure among the p-values.

\parheading{Comparison with TB minimum aggregation.}
For the minimum merge $f(p_1,\ldots,p_K)=\min_{k\in[K]}p_k$, define
\begin{align*}
Q_{b,k}^{\mathrm{TB}}(u)
\coloneqq
\mathrm{Quantile}_{1-u}
\bigl\{T_{b}^k,\tilde{T}_{1}^k,\ldots,\tilde{T}_{B}^k\bigr\}.
\end{align*}
The TB critical value can be written explicitly as
\begin{align*}
\uTB = \sup\biggl\{
u \in (0,1)
:
\frac{1}{B+1}
\sum_{b=0}^B
\mathds{1}
\biggl(
\max_{k \in [K]}
\Bigl\{
T_{b}^k-Q_{b,k}^{\mathrm{TB}}(u)
\Bigr\}
> 0
\biggr)
\leq \alpha
\biggr\}.
\end{align*}
The difference from the MaxT critical value in \eqref{eq: estimated u alpha} lies in how the calibration is constructed. In the TB procedure, each testing statistic $T_b^k$ is ranked against a reference batch that explicitly includes $T_b^k$ itself, and the randomization average in the definition of $\uTB$ is taken over $b=0,\ldots,B$. In contrast, the MaxT calibration does not include $T_b^k$ in the ranking set and averages only over the calibration indices $b=1,\ldots,B$. Although seemingly minor, these differences are structural: by including $T_b^k$ in the ranking set and averaging over $b=0,\ldots,B$, the TB construction is always well-defined, restores exact conditional exchangeability and achieves finite-sample type~I error control.

\subsection{General sequential aggregation principle}
\label{sec: general sequential aggregation}

This subsection records the general sequential aggregation principle underlying the SeqSB construction in \Cref{sec:sequential-alpha-spending-aggregation}. Let $I_1,\ldots,I_J \subseteq [K]$ be prespecified index sets, not necessarily nested, and for each $j\in[J]$ let $h_j\colon[0,1]^{|I_j|}\to\mathbb{R}$ be a measurable function, where smaller values indicate stronger evidence against the null. For $b\in[B]_0$ and $j\in[J]$, define the stage-$j$ score
\begin{align}
\label{eq: general seq score}
z_{b,j}\coloneqq h_j\bigl((p(T_b^k))_{k\in I_j}\bigr).
\end{align}
Let $\alpha_1,\ldots,\alpha_J\in[0,1]$ be stage-wise significance budgets with $\sum_{j=1}^J\alpha_j\le\alpha$. Set $S_0\coloneqq[B]_0$, and define recursively, for $j\in[J]$,
\begin{equation}
\label{eq: general seq threshold}
\begin{aligned}
c_j
&\coloneqq
\sup\biggl\{
u\in\mathbb{R}:
\frac{1}{B+1}\sum_{b\in S_{j-1}}
\mathds{1}(z_{b,j}\le u)
\le \alpha_j
\biggr\},\\
A_j&\coloneqq \{\,b\in S_{j-1}: z_{b,j}<c_j\,\},
\qquad
S_j\coloneqq S_{j-1}\setminus A_j.
\end{aligned}
\end{equation}
The sequential aggregation test rejects whenever $0\in\bigcup_{j=1}^J A_j$.

\begin{proposition}
\label{prop: general sequential aggregation}
Fix $\alpha\in(0,1)$, $B\ge 1$, and stage-wise budgets $\alpha_1,\ldots,\alpha_J\in[0,1]$ satisfying $\sum_{j=1}^J\alpha_j\le\alpha$. Let $(A_j)_{j=1}^J$ denote the sequence of eliminated sets produced by the recursive procedure~\eqref{eq: general seq threshold}. Under the group-invariance hypothesis,
\begin{align*}
\mP\biggl(0\in\bigcup_{j=1}^J A_j\biggr)
\le \frac{1}{B+1} \sum_{j=1}^J \lfloor (B+1)\alpha_j \rfloor
\le \alpha.
\end{align*}
Moreover, when $\{z_{b,j}: b\in S_{j-1}\}$ are almost surely distinct for every $j\in[J]$, the first inequality is tight.
\end{proposition}
\begin{proof}
See \Cref{sec: proof of general sequential aggregation}.
\end{proof}

\subsection{PCA-based aggregation using the reference batch}
\label{sec:pca-reference-batch}

This subsection gives a concrete example of a data-dependent aggregation rule that is enabled by the TB split. The idea is to learn a low-dimensional representation of the coordinate-wise p-value vector using only the reference batch, and then use this learned representation when aggregating the testing batch.

For each testing index $b\in[B]_0$, let
\begin{align*}
P_b \coloneqq
\bigl(
p_{\mathrm{HO}}(T_b^1),\ldots,p_{\mathrm{HO}}(T_b^K)
\bigr),
\end{align*}
and let $\tilde P_1,\ldots,\tilde P_B$ denote the corresponding reference-batch vectors, defined by
\begin{align*}
\tilde P_b \coloneqq
\biggl(
\frac{1}{B} \sum_{j=1}^B \mathds{1}\big(\tilde T_j^1\ge \tilde T_b^1\big),\ldots,
\frac{1}{B} \sum_{j=1}^B \mathds{1}\big(\tilde T_j^K\ge \tilde T_b^K\big)
\biggr),
\qquad b\in[B].
\end{align*}
Using only the reference batch, let $\hat\Sigma$ denote the empirical covariance matrix of $\{\tilde P_i\}_{i=1}^B$, and let $\hat v\in\mathbb R^K$ be a unit leading eigenvector of $\hat\Sigma$, selected according to a fixed deterministic tie-breaking and sign convention. Define the aggregation map
\begin{align*}
f_{\mathrm{PC}}(p_1,\ldots,p_K)
\coloneqq
\langle \hat v,\,(p_1,\ldots,p_K)\rangle .
\end{align*}
The merged testing-batch statistics are then given by $\tilde f_b = f_{\mathrm{PC}}(P_b)$ for $b\in[B]_0$.

Intuitively, this construction retains only the dominant direction of variation of the p-value vector under the null, thereby removing redundancy across coordinates. This is useful when the number of coordinates $K$ is large and the coordinate-wise p-values exhibit redundancy or an approximately low-dimensional dependence structure. In such regimes, fixed aggregation rules such as uniform averaging can overweight redundant directions. Because the projection vector $\hat v$ is measurable with respect to the reference batch, it is fixed after conditioning on the reference information, and the testing-batch rows remain exchangeable. As a result, the corresponding TB permutation p-value remains finite-sample valid. Unlike standard sample-splitting schemes that partition the data sample $\mathbf{X}$, the split here is performed only over the transformations, so all observations are used for both learning the aggregation map and performing the test.

\subsection{SB and TB conformal prediction algorithms}
\label{sec:alg-conformal}

Recall from \Cref{sec:conformal-prediction}
that the SB conformal prediction set is defined as
$C_{\mathrm{SB}}(X_{n+1})\coloneqq\{y\in\mathcal Y:\ p_{\mathrm{SB}}(y)>\alpha\}$
for the SB p-value
$
p_{\mathrm{SB}}(y)
\coloneqq
\frac{1}{n+1}\sum_{i=1}^{n+1}
\mathds{1}\!\left\{M_i(y)\le M_{n+1}(y)\right\},
$
with merged values
$
M_i(y) \coloneqq f\bigl(P_i^1(y),\ldots,P_i^K(y)\bigr)
$
of rank-transformed values
$
P_i^k(y)
\coloneqq
\frac{1}{n+1}\sum_{j=1}^{n+1}
\mathds{1}\{S_j^k(y)\ge S_i^k(y)\}
$
with residual scores
$S_i^k \coloneqq s_k(X_i,Y_i)$ for $i\in[n]$ and $S_{n+1}^k(y)\coloneqq s_k(X_{n+1},y)$.
When using the absolute residual scores $S_i^k \coloneqq |Y_i-\widehat{\mu}_k(X_i)|$, $i\in[n]$ and $S_{n+1}^k(y)\coloneqq |y-\widehat{\mu}_k(X_{n+1})|$, we note that the indicator $\mathds{1}(S_i^k \ge |y - \widehat{\mu}_k(X_{n+1})|)$ changes its truth value only when the test-point score equals a calibration residual, i.e., $|y - \widehat{\mu}_k(X_{n+1})| = S_i^k$. 
Between any two such points, none of the indicators change, meaning all p-values, merged scores $M_i$, and the final aggregated p-value are completely constant. 
As such, $p_\mathrm{SB}$ is a step function of $y$ with breakpoints at $\{ \widehat{\mu}_k(X_{n+1}) \pm S_i^k : k\in[K], i\in[n] \}$, and 
the exact continuous prediction set can be perfectly recovered by evaluating the functions at exactly one representative point in each open cell and at the boundary breakpoints themselves.
This allows for exact efficient computation of the SB conformal prediction set, entirely bypassing the need for discretization of $\mathcal Y=\mathbb R$, as detailed below.

First, sort the set $\{ \widehat{\mu}_k(X_{n+1}) \pm S_i^k : k\in[K], i\in[n] \}$, removing duplicates, to obtain ordered breakpoints $\tilde y_1<\ldots<\tilde y_L$. We need to evaluate $p_\mathrm{SB}$ at these breakpoints to determine whether each of the singleton sets $\{\tilde y_\ell\}, \ell\in[L]$ should be included in the SB conformal prediction set $C_{\mathrm{SB}}(X_{n+1})$. 
We also need to evaluate $p_\mathrm{SB}$ at reference points within the intervals $(\tilde y_\ell,\tilde y_{\ell+1}), \ell\in[L]_0$, with $\tilde y_0\coloneqq -\infty$ and $\tilde y_{L+1}\coloneqq \infty$ to determine whether each of these should be included in $C_{\mathrm{SB}}(X_{n+1})$.
Since the conformal p-value function is piecewise constant, evaluating the merged scores at exactly one representative point perfectly determines the acceptance or rejection of the entirety of the interval. Furthermore, stepping onto a breakpoint, and subsequently stepping off of it into the next open cell, triggers incremental updates for the exact same subset of coordinate-index pairs.

To this end, we construct the evaluation points as 
\begin{equation}
\label{eq:conf_y}
y_\ell =
\begin{cases}
\tilde y_1 - 1 & \text{if } \ell = 0, \\
\tilde y_j & \text{if } \ell = 2j - 1 \text{ for } j \in \{1, \dots, L\}, \\
(\tilde y_j + \tilde y_{j+1})/2 & \text{if } \ell = 2j \text{ for } j \in \{1, \dots, L-1\}, \\
\tilde y_L +1 & \text{if } \ell = 2L,
\end{cases}
\end{equation}
with corresponding intervals
\begin{equation}
\label{eq:conf_intervals}
\mathcal I_\ell =
\begin{cases}
(-\infty, \tilde y_1) & \text{if } \ell = 0, \\
\{\tilde y_j\} & \text{if } \ell = 2j - 1 \text{ for } j \in \{1, \dots, L\}, \\
(\tilde y_j, \tilde y_{j+1}) & \text{if } \ell = 2j \text{ for } j \in \{1, \dots, L-1\}, \\
(\tilde y_L, \infty) & \text{if } \ell = 2L,
\end{cases}
\end{equation}
and corresponding $k$-indices
\begin{equation}
\label{eq:conf_k}
\mathcal K_\ell =
\begin{cases}
[K] & \text{if } \ell = 0, \\
\mathcal K(\tilde y_j) & \text{if } \ell = 2j - 1 \text{ for } j \in \{1, \dots, L\}, \\
\mathcal K(\tilde y_j) & \text{if } \ell = 2j \text{ for } j \in \{1, \dots, L-1\}, \\
\mathcal K(\tilde y_L) & \text{if } \ell = 2L,
\end{cases}
\end{equation}
where
$
\mathcal K(y)\coloneqq \{k\in[K]:\exists i\in[n]\text{ s.t. } y = \widehat{\mu}_k(X_{n+1}) \pm S_i^k\}
$, and letting 
$
\mathcal J^k(y)\coloneqq \{i\in[n]:y = \widehat{\mu}_k(X_{n+1}) \pm S_i^k\} 
$,
for $k\in\mathcal K_\ell$
we define 
\begin{equation}
\label{eq:conf_i}
\mathcal J^k_\ell =
\begin{cases}
[n] & \text{if } \ell = 0, \\
\mathcal J^k(\tilde y_j) & \text{if } \ell = 2j - 1 \text{ for } j \in \{1, \dots, L\}, \\
\mathcal J^k(\tilde y_j) & \text{if } \ell = 2j \text{ for } j \in \{1, \dots, L-1\}, \\
\mathcal J^k(\tilde y_L) & \text{if } \ell = 2L.
\end{cases}
\end{equation}
Here, $y_\ell$ is the representative evaluation point, $\mathcal I_{\ell}$ is its associated interval or singleton, and $\mathcal{K}_{\ell}, \mathcal{J}_{\ell}^k$ are the subsets of coordinates and indices that require incremental updates at step $\ell$. 

For each evaluation point $y_\ell$ for $\ell\in[2L]_0$, we need to evaluate $p_\mathrm{SB}(y_\ell)$, which requires computing the p-values $P_i^k$ and the merged values $M_i$ for all $i\in[n+1]$ and $k\in[K]$.
To evaluate this sequence efficiently without incurring the full $\mathcal{O}(n K)$ cost at every step, we update the p-values and merged scores on the fly. 
At the leftmost evaluation point $y_0$, we initialize the p-values $P_i^k(y_0)$ and merged scores $M_i(y_0)$ from scratch. 
As we sweep left to right for subsequent steps $\ell > 0$, only the specific coordinates $\mathcal{K}_{\ell}$ and indices $\mathcal{J}_{\ell}^k$ for $k\in\mathcal{K}_{\ell}$ triggered by the breakpoint undergo a change, computable in at most $\mathcal{O}(\log n)$ time. If the test score crosses a calibration residual, the respective indicator sum either increases or decreases by $1$. We then re-evaluate only the affected merged scores $M_i$ and dynamically adjust their position in the sorted dynamic data structure $\mathbb{M}$ (e.g., an order-statistic tree, balanced binary search tree, or Fenwick tree). Merging functions such as the minimum and median can be updated incrementally in time $\mathcal{O}(\log K)$ (time $\mathcal{O}(1)$ for the mean), without requiring the full $\mathcal{O}(K)$ cost of applying the merging function. However, for a generic merging function, this cost is inevitable.

The exact SB conformal prediction procedure with absolute residuals is presented in \Cref{alg:SB-inversion-unified}, along with a detailed analysis of the computational complexity for each step.
The TB conformal prediction procedure can be exactly and efficiently evaluated in a similar manner, with a much simpler implementation of the update for each evaluation point, as presented in \Cref{alg:TB-inversion-unified}.
Finally, we detail in \Cref{alg:TB-intersection-shortcut} the exact implementation of the expression in \Cref{prop:cp_intersection} for the minimum merging function, which drastically reduces the computational complexity of the algorithm.
While the methods are presented in \Cref{sec:conformal-prediction} for the case of a single test point $X_{n+1}$, in practice we consider a batch of $M$ test points $\{X_{n+1}^m\}_{m=1}^M$ simultaneously, as conformal prediction is typically deployed to generate valid prediction sets for an entire batch of unlabelled observations, with $M$ often taking very large values in real-world applications.

This practical reality of evaluating large test batches highlights a critical computational dichotomy between the approaches. For the SB procedure (\Cref{alg:SB-inversion-unified}), the necessity of dynamically updating the calibration merged scores $M_i$ and $M_{n+1}$ across all evaluation points imposes a total computational runtime that scales as $\mathcal{O}(MKn \log(Kn))$ for minimum, mean and median merging functions. The generic TB inversion (\Cref{alg:TB-inversion-unified}) relaxes this burden by relying on a static reference threshold, which avoids the need to update $M_i$ but still yields a comparable test-time cost of $\mathcal{O}(MKn_1 \log(Kn_1))$ assuming $n_1$ and $n_2$ are of the same order. However, the true computational advantage emerges under the minimum merging rule with the implementation of \Cref{prop:cp_intersection} (\Cref{alg:TB-intersection-shortcut}). By entirely bypassing the grid-free sweep in favor of directly intersecting fixed residual quantiles, the per-test-point complexity collapses from $\mathcal{O}(Kn_1\log(Kn_1))$ to merely $\mathcal{O}(K)$. In modern deployment regimes where $M \gg n$, this reduction to a strictly $\mathcal{O}(Kn_1 \log(Kn_1) + MK)$ bottleneck isolates the test-time cost from the calibration size, transforming an otherwise prohibitive search into a highly scalable procedure.

\begin{algorithm}[htbp]
\caption{Exact SB Conformal Prediction for General Merging Function}
\label{alg:SB-inversion-unified}
\begin{algorithmic}[1]
	\Require Calibration data $\{(X_i, Y_i)\}_{i=1}^n$; test points $\{X_{n+1}^m\}_{m=1}^M$; prediction functions $\{\widehat{\mu}_k\}_{k=1}^K$; merging function $f$; level $\alpha$.
\State Compute $S_i^k \gets |Y_i - \widehat{\mu}_k(X_i)|$ for $k\in[K]$ and $i\in[n]$.  \Comment{$\mathcal{O}(Kn)$}
\State Sort $\{S_i^k : i \in [n]\}$ for each $k \in [K]$. \Comment{$\mathcal{O}(Kn \log n)$}
\For{$m=1,\ldots,M$}
    \State Initialize $C_{\mathrm{SB}}(X_{n+1}^m) \gets \emptyset$.
    \State Sort $\{ \widehat{\mu}_k(X_{n+1}^m) \pm S_i^k \!:\! k\!\in\![K], i\!\in\![n] \}$, remove duplicates, get $\tilde y_1 < \dots < \tilde y_L$. $\!\triangleright\ \mathcal{O}(Kn \log(Kn))$
\State Construct the evaluation tuples $(y_\ell, \mathcal{I}_{\ell}, \mathcal{K}_{\ell}, \mathcal{J}^k_{\ell})_{\ell=0}^{2L}$ as in \eqref{eq:conf_y}, \eqref{eq:conf_intervals}, \eqref{eq:conf_k} and \eqref{eq:conf_i}. \Comment{$\mathcal{O}(Kn)$}
    \State Initialize $S_{n+1}^k \gets |y_0 - \widehat{\mu}_k(X_{n+1}^m)|$ for $k\in[K]$. \Comment{$\mathcal{O}(K)$}
    \State Initialize $p_i^k \gets P_i^k(y_0)$ for $k\in[K]$ and $i\in[n+1]$ using sorted $\{S_j^k:j\in[n]\}$. \Comment{$\mathcal{O}(Kn\log n)$}
    \State Initialize $M_i \gets f\bigl(p_i^1, \dots, p_i^K\bigr)$ for $i\in[n+1]$. \Comment{$\mathcal{O}(Kn)$}
    \State Initialize $\mathbb{M} \gets \mathtt{sort}\bigl(M_1, \dots, M_n\bigr)$ using a dynamic data structure. \Comment{$\mathcal{O}(n \log n)$}
    \State Compute $p_{\mathrm{SB}}(y_0) \gets \frac{1}{n+1}\left(1+\sum_{i=1}^n \mathds{1}\bigl(M_i \le M_{n+1}\bigr)\right)$ via binary search on $\mathbb{M}$. \Comment{$\mathcal{O}(\log n)$}
	\State \textbf{if } {$p_{\mathrm{SB}}(y_0) > \alpha$} \textbf{ then } $C_{\mathrm{SB}}(X_{n+1}^m) \gets C_{\mathrm{SB}}(X_{n+1}^m) \cup \mathcal I_{0}$.
    \For{$\ell=1, \ldots, 2L$}
        \For{$k \in \mathcal{K}_{\ell}$}
            \State Update $S_{n+1}^k \gets |y_\ell - \widehat{\mu}_k(X_{n+1}^m)|$.
	        \State Update $p_{n+1}^k \gets P_{n+1}^k(y_\ell)$ using sorted $\{S_j^k:j\in[n]\}$. \Comment{$\mathcal{O}(\log n)$}
        \State Update $M_{n+1} \!\gets\! f\bigl(p_{n+1}^1, \dots, p_{n+1}^K\bigr)$. $\triangleright\, \mathcal{O}(\log K)$ for min/mean/med, \!$\mathcal{O}(K)$ for generic
	\For{$i \in \mathcal{J}_{\ell}^k$}
	\State Update $p_{i}^k \gets P_{i}^k(y_\ell)$ by updating only $\mathds{1}(S_{n+1}^k \geq S_i^k)$. \Comment{$\mathcal{O}(1)$}
                \State Remove $M_i$ from sorted $\mathbb{M}$. \Comment{$\mathcal{O}(\log n)$}
            \State Update $M_i \gets f\bigl(p_i^1, \dots, p_i^K\bigr)$. \Comment{$\mathcal{O}(\log K)$ for min/mean/med, $\mathcal{O}(K)$ for generic}
            \State Insert $M_i$ into $\mathbb{M}$ preserving the sorted order. \Comment{$\mathcal{O}(\log n)$}
        \EndFor
        \EndFor
        \State Compute $p_{\mathrm{SB}}(y_\ell) \gets \frac{1}{n+1}\left(1+\sum_{i=1}^n \mathds{1}\bigl(M_i \le M_{n+1}\bigr)\right)$ via binary search on $\mathbb{M}$. \Comment{$\mathcal{O}(\log n)$}
	\State \textbf{if } {$p_{\mathrm{SB}}(y_\ell) > \alpha$} \textbf{ then } $C_{\mathrm{SB}}(X_{n+1}^m) \gets C_{\mathrm{SB}}(X_{n+1}^m) \cup \mathcal I_{\ell}$.
    \EndFor
\EndFor
\State \textbf{Return:} $\{C_{\mathrm{SB}}(X_{n+1}^m):m\in[M]\}$.
\State \textbf{Complexity:} $\mathcal{O}(Kn \log n + MKn \log(Kn))$ for min, mean, and median $f$,
\State \phantom{\textbf{Complexity:} }$\mathcal{O}(Kn \log n + MKn(\log(Kn) + K))$ for generic $f$.
\end{algorithmic}
\end{algorithm}

\begin{algorithm}[htbp]
\caption{Exact TB Conformal Prediction for General Merging Function}
\label{alg:TB-inversion-unified}
\begin{algorithmic}[1]
	\Require Calibration data $\{(X_i, Y_i)\}_{i=1}^n$; disjoint index batches $I_{\mathrm{ref}}\cup I_{\mathrm{agg}}=[n]$ with $|I_{\mathrm{ref}}|=n_1$ and $|I_{\mathrm{agg}}|=n_2$; test points $\{X_{n+1}^m\}_{m=1}^M$; prediction functions $\{\widehat{\mu}_k\}_{k=1}^K$; merging function $f$; level $\alpha$.
\State Compute $S_j^k \gets |Y_j - \widehat{\mu}_k(X_j)|$ for $k\in[K]$ and $j\in I_{\mathrm{ref}}$.  \Comment{$\mathcal{O}(Kn_1)$}
\State Sort $\{S_j^k : j \in I_{\mathrm{ref}}\}$ for each $k \in [K]$. \Comment{$\mathcal{O}(Kn_1 \log n_1)$}
\State Compute $p_i^{k,\mathrm{ref}} \gets P_i^{k,\mathrm{ref}}$ for $k\in[K]$ and $i\in I_{\mathrm{agg}}$ using sorted $\{S_j^k:j\in I_{\mathrm{ref}}\}$. \Comment{$\mathcal{O}(Kn_2 \log n_1)$}
\State Compute $M_i \gets f\bigl(p_i^{1,\mathrm{ref}},\ldots,p_i^{K,\mathrm{ref}}\bigr)$ for $i\in I_{\mathrm{agg}}$. \Comment{$\mathcal{O}(Kn_2)$}
\State Compute $u_{\alpha}^{\mathrm{TB}} \gets -\mathrm{Quantile}_{(1-\alpha)(1+1/n_2)}\bigl\{-M_i : i\in I_{\mathrm{agg}}\bigr\}$. \Comment{$\mathcal{O}(n_2 \log n_2)$}
\For{$m=1,\ldots,M$}
    \State Initialize $C_{\mathrm{TB}}(X_{n+1}^m) \gets \emptyset$.
    \State Sort $\{ \widehat{\mu}_k(X_{n+1}^m) \!\pm\! S_j^k \!:\! k\!\in\![K], j\!\in\! I_{\mathrm{ref}} \}$, remove duplicates, get $\tilde y_1 \!< \!\dots \!<\! \tilde y_L$ {$\triangleright\ \mathcal{O}(Kn_1\!\log(Kn_1))$}
    \State Construct $(y_\ell, \mathcal I_{\ell}, \mathcal{K}_{\ell})_{\ell=0}^{2L}$ as in \eqref{eq:conf_y}, \eqref{eq:conf_intervals} and \eqref{eq:conf_k} using $I_\mathrm{ref}$ instead of $[n]$ for $\mathcal K_\ell$. \Comment{$\mathcal{O}(Kn_1)$}
    \State Initialize $S_{n+1}^k \gets |y_0 - \widehat{\mu}_k(X_{n+1}^m)|$ for $k\in[K]$. \Comment{$\mathcal{O}(K)$}
    \State Initialize $p_{n+1}^{k,\mathrm{ref}} \gets P_{n+1}^{k,\mathrm{ref}}(y_0)$ for $k\in[K]$ using sorted $\{S_j^k:j\in I_{\mathrm{ref}}\}$. \Comment{$\mathcal{O}(K\log n_1)$}
    \State Initialize $M_{n+1} \gets f\bigl(p_{n+1}^{1,\mathrm{ref}}, \dots, p_{n+1}^{K,\mathrm{ref}}\bigr)$. \Comment{$\mathcal{O}(K)$}
    \State \textbf{if } {$M_{n+1} \ge u_{\alpha}^{\mathrm{TB}}$} \textbf{ then } $C_{\mathrm{TB}}(X_{n+1}^m) \gets C_{\mathrm{TB}}(X_{n+1}^m) \cup \mathcal I_{0}$.
    \For{$\ell=1, \ldots, 2L$}
        \For{$k \in \mathcal{K}_{\ell}$}
            \State Update $S_{n+1}^k \gets |y_\ell - \widehat{\mu}_k(X_{n+1}^m)|$.
            \State Update $p_{n+1}^{k,\mathrm{ref}} \gets P_{n+1}^{k,\mathrm{ref}}(y_\ell)$ using sorted $\{S_j^k:j\in I_{\mathrm{ref}}\}$. \Comment{$\mathcal{O}(\log n_1)$}
        \EndFor
        \State Update $M_{n+1} \gets f\bigl(p_{n+1}^{1,\mathrm{ref}}, \dots, p_{n+1}^{K,\mathrm{ref}}\bigr)$. $\triangleright\ \mathcal{O}(\log K)$ for min/mean/med, $\mathcal{O}(K)$ for generic
        \State \textbf{if } {$M_{n+1} \ge u_{\alpha}^{\mathrm{TB}}$} \textbf{ then } $C_{\mathrm{TB}}(X_{n+1}^m) \gets C_{\mathrm{TB}}(X_{n+1}^m) \cup \mathcal I_{\ell}$.
    \EndFor
\EndFor
\State \textbf{Return:} $\{C_{\mathrm{TB}}(X_{n+1}^m):m\in[M]\}$.
\State \textbf{Complexity:} $\!\mathcal{O}(K(n_1\!+\!n_2)\log n_1 \!+\! n_2 \log n_2 \!+\! MKn_1 \log(Kn_1))$ for min, mean, and median $f$,
\State \phantom{\textbf{Complexity:} }$\!\mathcal{O}(K(n_1\!+\!n_2)\log n_1 \!+\! n_2 \log n_2 \!+\! MKn_1(\log(Kn_1) + K))$ for generic $f$.
\end{algorithmic}
\end{algorithm}

\begin{algorithm}[htbp]
\caption{Efficient Exact TB Conformal Prediction for Minimum Merging (\Cref{prop:cp_intersection})}
\label{alg:TB-intersection-shortcut}
\begin{algorithmic}[1]
	\Require Calibration data $\{(X_i, Y_i)\}_{i=1}^n$; disjoint index batches $I_{\mathrm{ref}}\cup I_{\mathrm{agg}}=[n]$ with $|I_{\mathrm{ref}}|=n_1$ and $|I_{\mathrm{agg}}|=n_2$; test points $\{X_{n+1}^m\}_{m=1}^M$; prediction functions $\{\widehat{\mu}_k\}_{k=1}^K$; level $\alpha$.
\State Compute $S_j^k \gets |Y_j - \widehat{\mu}_k(X_j)|$ for $k\in[K]$ and $j\in I_{\mathrm{ref}}$.  \Comment{$\mathcal{O}(Kn_1)$}
\State Sort $\{S_j^k : j \in I_{\mathrm{ref}}\}$ to get $S_{(1)}^k \le \dots \le S_{(n_1)}^k$ for each $k \in [K]$. \Comment{$\mathcal{O}(Kn_1 \log n_1)$}
\State Set $S_{(n_1+1)}^k \gets \infty$ for each $k \in [K]$. \Comment{$\mathcal{O}(K)$}
\State Compute $p_i^{k,\mathrm{ref}} \gets P_i^{k,\mathrm{ref}}$ for $k\in[K]$ and $i\in I_{\mathrm{agg}}$ using sorted $\{S_j^k:j\in I_{\mathrm{ref}}\}$. \Comment{$\mathcal{O}(Kn_2 \log n_1)$}
\State Compute $M_i \gets \min_{k\in[K]} p_i^{k,\mathrm{ref}}$ for $i\in I_{\mathrm{agg}}$. \Comment{$\mathcal{O}(Kn_2)$}
\State Compute $u_{\alpha}^{\mathrm{TB}} \gets -\mathrm{Quantile}_{(1-\alpha)(1+1/n_2)}\bigl\{-M_i : i\in I_{\mathrm{agg}}\bigr\}$. \Comment{$\mathcal{O}(n_2 \log n_2)$}
\State \textbf{if } {$u_{\alpha}^{\mathrm{TB}}=-\infty$} \textbf{ then } $\ell\gets n_1+1$  \textbf{ else } $\ell \gets n_1 + 2 - \lceil u_{\alpha}^{\mathrm{TB}}(n_1+1) \rceil$. \Comment{$\mathcal{O}(1)$}
\For{$m=1,\ldots,M$}
    \State $C_{\mathrm{TB}}(X_{n+1}^m) \gets \bigcap_{k=1}^K \Bigl[ \widehat{\mu}_k(X_{n+1}^m) - S_{(\ell)}^k,\, \widehat{\mu}_k(X_{n+1}^m) + S_{(\ell)}^k \Bigr]$. \Comment{$\mathcal{O}(K)$}
\EndFor
\State \textbf{Return:} $\{C_{\mathrm{TB}}(X_{n+1}^m):m\in[M]\}$.
\State \textbf{Complexity:} $\mathcal{O}(K(n_1+n_2)\log n_1 + n_2 \log n_2 + MK)$.
\end{algorithmic}
\end{algorithm}

\subsection{Powerful tie-breaking strategies}
\label{sec:tie-breaking}

When the underlying data distributions are discrete, identical evaluations across permutations occur with positive probability. While such ties may arise naturally at the level of individual marginal statistics $T_b^k$, this degeneracy is severely compounded when these evaluations are aggregated via merging functions, particularly those with a low-cardinality image space such as the minimum merging function. Because such merging functions project multi-dimensional vectors onto highly constrained supports, structural ties on the final merged test statistics occur with substantially higher frequency. The standard p-value transformation, evaluated conservatively without tie-breaking, assigns the maximum of the tied ranks to all identical evaluations, inducing a conservative bias that artificially limits statistical power. Resolving these ties is therefore practically necessary. 

As briefly explained in \Cref{sec:SB-finite-sample-validity}, tie-breaking can be implemented by augmenting each transformed row with auxiliary variables $U_b^1,\dots,U_b^L$. For example, the marginal tie-broken p-values can be written as
\begin{align*}
    p_U(T_b^k)
    \coloneqq
    \frac{1}{B+1}\sum_{i=0}^B
    \mathds{1}\!\left\{
    (T_i^k,U_i^1,\dots,U_i^L)
    \geq_{\mathrm{lex}}
    (T_b^k,U_b^1,\dots,U_b^L)
    \right\},
\end{align*}
where $\ge_{\mathrm{lex}}$ denotes the usual lexicographical order. The final merged statistic can be ranked analogously by
\begin{align*}
    p_{\mathrm{SB},U}
    \coloneqq
    \frac{1}{B+1}\sum_{i=0}^B
    \mathds{1}\!\left\{
    (f_i,-U_i^1,\dots,-U_i^L)
    \leq_{\mathrm{lex}}
    (f_0,-U_0^1,\dots,-U_0^L)
    \right\}.
\end{align*}
Provided that the augmented sequence of tuples $(T_b^1,\dots,T_b^K,f_b,U_b^1,\dots,U_b^L)_{b=0}^B$ remains exchangeable under the null hypothesis, these lexicographical transformations preserve finite-sample exchangeability and hence level control. While ties among the statistics $T_0^k,\dots,T_B^k$ are typically rare due to their often continuous support, ties among the merged values $f_0,\dots,f_B$ occur far more frequently because the merging function often projects these values onto a highly-constrained low-cardinality support (e.g., minimum merging function). Importantly, when breaking $f_b$-ties, the lexicographically tie-broken p-value is never larger than its conservative counterpart obtained without tie-breaking. Consequently, any power guarantee established for the conservative procedure immediately carries over to the tie-broken procedure. We formalize specific strategies for choosing the auxiliary variables within the SB testing framework, which readily generalize to TB and data-driven SB.

\begin{itemize}
    \item \textbf{No tie-breaking (conservative baseline).} Setting all auxiliary variables equal, i.e.,
	\begin{align*}
		U_b^1=\cdots=U_b^L=c,	
	\end{align*}
	recovers the standard conservative ranking in which tied evaluations receive the maximum tied rank. This choice is computationally trivial and guarantees a valid test. However, failing to separate the transformed or merged statistics yields an upwardly biased p-value, resulting in an overly conservative procedure with suboptimal power. 
	\\[-0.8em]
    
    \item \textbf{Uniformly random tie-breaking.} Ties may be broken by appending an independent and identically distributed random sequence, e.g., $U_b^1 \sim \mathrm{Uniform}(0, 1)$, to each permutation and evaluating the lexicographical ordering of $(T_b^k, U_b^1)$ for each $k$, and of $(f_b, -U_b^1)$. 
	Although this approach typically improves empirical power relative to the conservative no-tie-breaking baseline, it does not exploit the information contained in the test statistics themselves. This motivates the following data-dependent strategies, which construct the auxiliary variables to leverage this structural information. \\[-0.8em]

	\item \textbf{Data-dependent tie-breaking.} Rather than breaking ties at random, the auxiliary variables can instead be constructed systematically to favor permutations that exhibit stronger evidence against the null. The exchangeability framework developed here permits such data-dependent tie-breaking rules while preserving finite-sample validity. We propose two concrete constructions below, although the underlying principle is not restricted to these examples. \\[-0.8em]

    \begin{enumerate}[label=(\roman*)]
	\item \textbf{Data-dependent tie-breaking via studentised statistics.} To implement this idea, we define the primary auxiliary sequence as the logsumexp, which is a smooth approximation to the maximum of the studentized statistics, i.e.,
\begin{equation}
    \label{eq:student}
    U_b^1 \coloneqq \log \sum_{k=1}^K \exp\left( (T^k_b - \bar{T}_b^k) / \sigma^k_b \right).
\end{equation}
Here, $\bar{T}_b^k$ and $\sigma^k_b$ denote centering and scaling parameters, which may, for example, be taken as the empirical mean and standard deviation of $\{T_0^k,\dots,T_B^k\}\setminus\{T_b^k\}$ in a leave-one-out fashion \citep[Eq. (4)]{shah2018goodness}. This preserves the exchangeability of the augmented sequence and hence finite-sample validity. This auxiliary variable systematically breaks ties in favor of permutations exhibiting stronger studentized evidence against the null. This can improve power relative to purely random tie-breaking.
	 \\[-0.8em]
	        
	\item \textbf{Data-dependent tie-breaking via p-value aggregation.} Alternatively, ties can first be structurally resolved by deriving auxiliary sequences from the conservative (no-tie-breaking) marginal p-values $\check{p}_b^1,\dots,\check{p}_b^K$. For instance, we may compute 
$$U_b^1\coloneqq -\sum_{k=1}^K \log(\check{p}_b^k), \qquad \textrm{and} \qquad U_b^2\coloneqq -\sum_{k=1}^K \check{p}_b^k$$ 
relying on Fisher's and Edgington's methods. The primary variable $U_b^1$ breaks all ties for which the product of the marginal p-values differs, while the secondary variable $U_b^2$ subsequently breaks all ties for which their sum differs. Consequently, any potential ties surviving both aggregations must necessarily have equal sums and products of marginal p-values (e.g., exact same multiset occurring in a different dimensional ordering). These remaining symmetric ties can then potentially be broken by evaluating $U_b^3$ as a softmax merging of studentized statistics as in \eqref{eq:student}. \\[-0.8em]
    \end{enumerate}
Nevertheless, inherent ties may still manifest if, for instance, the resampling procedure samples the exact same group transformation multiple times. Since identical transformations yield indistinguishable evaluations across all deterministic constructions, these remaining degeneracies can only be resolved by appending a terminal, independent uniform random variable, e.g., $U_b^L \sim \mathrm{Uniform}(0, 1)$, thereby preserving finite-sample validity while breaking the remaining ties.
\end{itemize}

\section{Additional simulation results}
\label{sec: additional simulation results}

This section reports simulation results that complement the experiments in \Cref{sec: simulations}.

\subsection{Type I error control}
\label{subsec: level control}

\begin{figure}[!t]
	\centering
	\includegraphics[width=\textwidth]{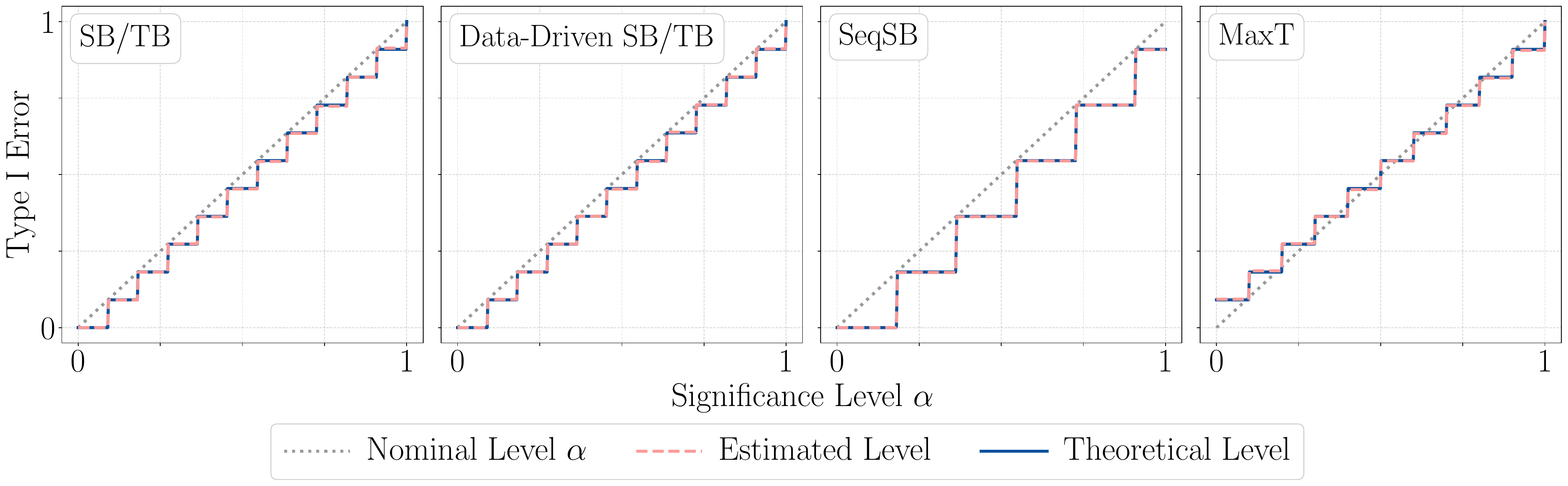}
	\caption{Type I error estimation for the SB/TB, Data-Driven SB/TB, SeqSB and MaxT tests.}
	\label{fig: figure_1}
\end{figure}

To estimate the type I error of the various tests in \Cref{fig: figure_1} we use the two-sample testing framework with $m=n=50$ i.i.d.~samples drawn from the uniform distribution on $(0,1)$, aggregating over statistics formed by the $L^p$ norms of the difference in means, and using permutations of the pooled samples.
Nonetheless, we stress that these level results hold more generally for any distribution, sample size and statistic in any framework that tests the group-invariance null hypothesis (\Cref{def:group-invariance hypothesis}), since they rely only on the exchangeability of the last $B+1$ values. 
For illustration purposes, only $B=10$ permutations are used in this setting, and uniformly random tie-breaking is implemented.
While the level achieved by (Data-Driven) SB/TB holds independently of the value of $K$, the results for SeqSB depend on $K$, and the level presented for MaxT only holds for $K=1$.
The estimated levels are averaged over 20,000 repetitions, and the theoretical levels are: 
${\lfloor (B+1)\alpha \rfloor}/{(B+1)}$ for SB, TB and their data-driven variants (\Cref{prop: SB level,prop: TB level}), 
${K\lfloor(B+1)\alpha/K\rfloor}/{(B+1)}$ for SeqSB with $\alpha_j=\alpha/K$, $j\in[K]$ (\Cref{prop: sequential minimum aggregation}), and ${(\lfloor B\alpha \rfloor + 1)}/{(B + 1)}$ for MaxT (\Cref{prop:failure of size control without bisection}).

As seen in \Cref{fig: figure_1}, the estimated levels of all implemented tests match their theoretical levels, confirming the validity of the theoretical analysis, as well as the correctness of our implementation. 
\Cref{fig: figure_1} highlights the fact that MaxT does not control the type I error at the desired level $\alpha$ for most values of $\alpha$ for a fixed number $B$ of transformations, even in the simplest case with $K=1$. 
Therefore, the MaxT procedure is not a valid test, and hence is not considered in our power experiments in the rest of this section to ensure fair comparisons across tests. 
We note that for $K>1$ the MaxT test exhibits a similar behavior (failing to control the type I error at level $\alpha$) but deviates from the theoretical level derived for $K=1$.
All other tests have type I error bounded above by $\alpha$ as desired. 
Due to the discrete nature of the permutation test, the test level is not always exactly $\alpha$, illustrating the importance of choosing $B$ appropriately; see \citep{cha2026more} for a formal study of the choice of $B$ in Monte Carlo permutation tests. 
In particular, the (Data-Driven) SB/TB tests of \Cref{sec: data-driven aggregation of merging functions} achieve exact level $\alpha$ whenever $(B+1)\alpha$ is an integer; in the following experiments we use $B=199$ transformations and level $\alpha=0.05$. 
The same parameter choice ensures that SeqSB with $\alpha_j=\alpha/K$ and $K=10$ also achieves exact level $\alpha$.
The classical worst-case tests of \Cref{sec:merging}, which are constructed to be valid under arbitrary dependence, are in general much more conservative, with much lower type I error, than the permutation-calibrated tests presented here.

\subsection{One-sample zero mean testing}
\label{subsec:one-sample-zero-mean}

\begin{figure}[!t]
	\centering
	\includegraphics[width=\textwidth]{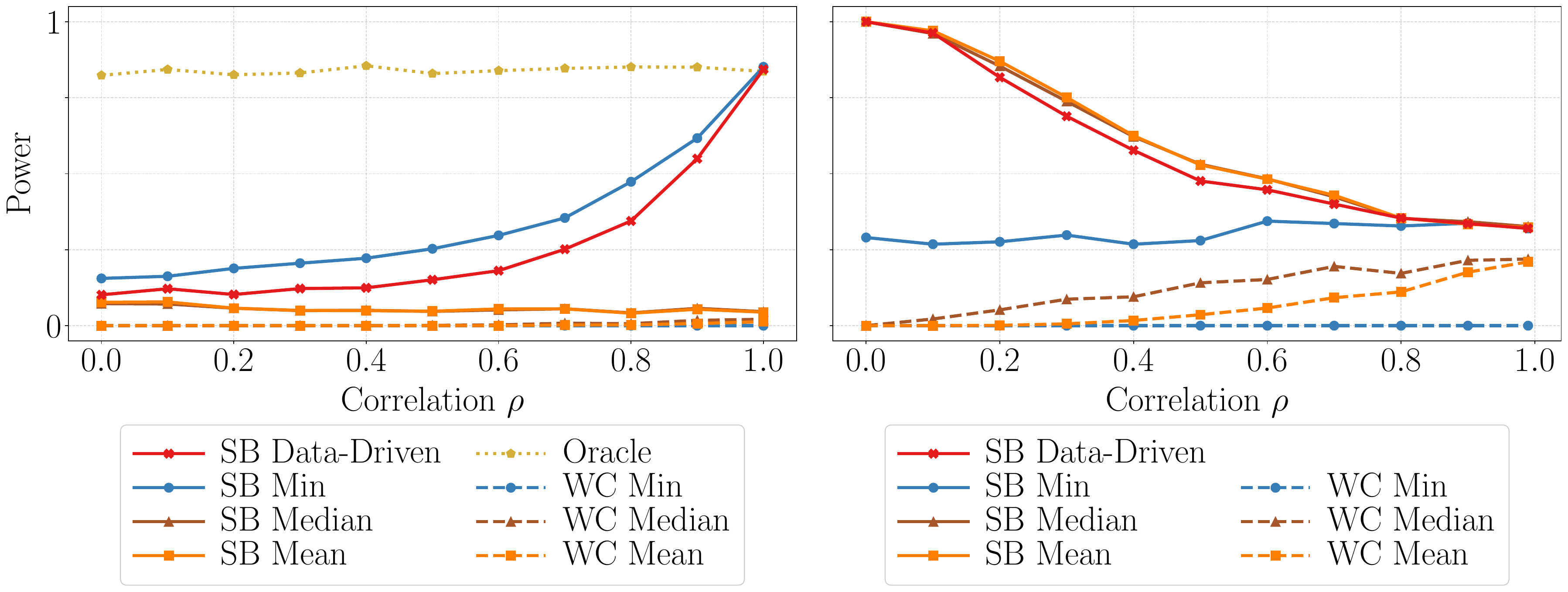}
	\caption{Power estimation for one-sample zero mean testing of an equicorrelated multivariate Gaussian.} 
	\label{fig: figure_3}
\end{figure}

To evaluate the proposed methods under varying degrees of dependence across statistics, we generate $n = 10,000$ samples from a $d$-dimensional multivariate normal distribution with ambient dimension $d = 40$. 
The covariance matrix is constructed such that all dimensions share a pairwise feature correlation parameter $\rho$, representing the off-diagonal entries, while the diagonal variances are set to $1$. 
We consider two distinct alternative hypotheses: a sparse setting in the left panel of \Cref{fig: figure_3} where only the first dimension contains a signal (a mean shift of $\mu_1 = 0.03$, with $\mu_{j} = 0$ for $j \neq 1$), and a dense setting for the right panel where the signal is distributed across all dimensions ($\mu_j = 0.015$ for all $j$). 
For each regime, we assess the performance of the worst-case (WC) and single-step (SB) procedures using the minimum, median, and mean aggregation functions over the absolute one-sample $t$-statistics computed for each dimension. 
Additionally, we evaluate the Data-Driven SB test, defined in \Cref{sec: data-driven aggregation of merging functions}, which adapts across the three merging functions. 
As a benchmark in the sparse regime, we include an Oracle test, which computes the classical permutation one-sample $t$-test strictly on the first dimension, assuming prior knowledge of the true signal location. 
For the group transformations, we multiply each data point by a random sign (i.e., Rademacher random variable) to simulate the case of zero mean, leveraging the symmetric property of the multivariate normal distribution.
All empirical power calculations are averaged over 1000 independent repetitions at a nominal significance level of $\alpha = 0.05$.

The empirical results, illustrated in \Cref{fig: figure_3}, reveal several key dynamics of the testing procedures. 
First, the worst-case (WC) tests are highly conservative, achieving near-zero empirical power across almost all settings. 
In the sparse signal regime (left panel), among the three single-step methods, only the SB Min test demonstrates strong performance. 
Because the mean shift is isolated to a single dimension, the minimum $p$-value effectively targets the only active statistic; taking the mean or median heavily dilutes this localized signal with the remaining $d-1$ noise dimensions. 
Conversely, in the dense signal regime (right panel), where the mean shift is applied across all dimensions, the SB Mean and SB Median tests successfully aggregate the distributed signal and consequently outperform the SB Min test.
Crucially, across both regimes, the SB Data-Driven test proves to be highly adaptive, consistently achieving statistical power that approximates the single best-performing aggregation method. 
Finally, the correlation $\rho$ significantly, yet oppositely, impacts test power depending on the signal structure: increasing $\rho$ in the sparse setting leads to higher power, whereas increasing $\rho$ in the dense setting leads to lower power. 
In the dense case, strong correlation reduces the effective number of independent signals available for aggregation. 
However, in the sparse case, highly correlated features effectively constrain the background noise, allowing the isolated signal to stand out. 
Notably, at $\rho = 1$, the SB Min and SB Data-Driven tests achieve the exact same empirical power as the single Oracle test, demonstrating that, unlike their WC counterparts, these adaptive procedures do not need to over-correct for perfectly redundant data (as explained in the discussion following \Cref{prop: minimum single-batch calibration statistic view}).

\subsection{Independence nonparametric testing}
\label{subsec: independence}

\begin{figure}[!t]
	\centering
	\includegraphics[width=\textwidth]{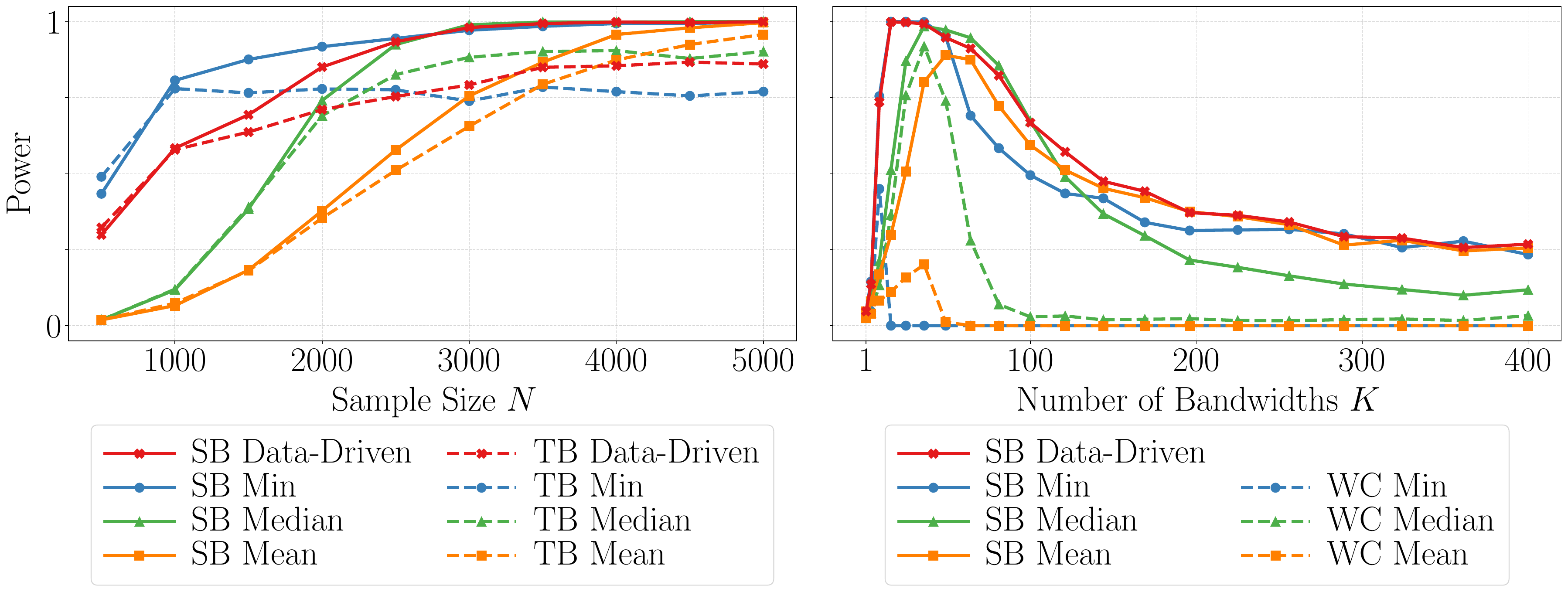}
	\caption{Power estimation for kernel-based HSIC independence nonparametric testing.}
	\label{fig: figure_4}
\end{figure}

We evaluate the proposed methods for independence testing using the Hilbert-Schmidt Independence Criterion (HSIC, \cite{gretton2005kernel}) in \Cref{fig: figure_4}. 
The data is generated from a Dirichlet distribution with a uniform concentration parameter vector $\mathbf{1}_d$. 
Specifically, we generate a sample of size $N$ in ambient dimension $d$, and partition each observation into two views, $X \in \mathbb{R}^{d-c}$ and $Y \in \mathbb{R}^c$, to evaluate the dependence between them. 
In our first experiment, we vary the sample size $N$ from 500 to 5000 while fixing the dimensions to $d=500$ and $c=200$. 
In the second experiment, we fix $N=500$, $d=200$, and $c=80$, and examine the impact of the bandwidth grid size. 
The HSIC V-statistic, with Gaussian kernels, is computed over a multi-scale grid where the total number of bandwidths evaluated, $K$, corresponds to the product of the number of candidate bandwidths for $X$ and $Y$ (i.e., $K = K_X \times K_Y$). 
Empirical power is averaged over 1000 independent repetitions, all tests are calibrated using $B=199$ permutations, and the nominal significance level is set to $\alpha = 0.05$. 
The paired data is transformed by permuting the data within one sample to break the dependence structure.
We compare the Single-Step (SB) methods against the Two-Batch (TB) and Worst-Case (WC) methods using three standard merging functions: the minimum, mean, and median of the individual test statistics. 
Furthermore, we assess the performance of Data-Driven SB/TB tests, defined analogously to \Cref{sec: data-driven aggregation of merging functions}, which are adaptive to the choice of these three merging functions.

The experimental results of \Cref{fig: figure_4} highlight the significant advantages of the proposed SB testing procedure. 
As shown in the left panel, the SB method consistently outperforms the TB method, with the power differential becoming particularly pronounced at larger sample sizes. 
The Data-Driven SB and TB procedures are observed to be adaptive, especially for larger sample sizes.
Furthermore, the right panel demonstrates that the SB procedure vastly outperforms the classical WC approach; this holds strictly true across all considered merging functions (minimum, mean, and median). 
When comparing the merging strategies themselves, the minimum function strictly achieves the highest power in the experiment varying the sample size (left panel). 
In the experiment varying the number of paired kernel bandwidths (right panel), each of the three merging functions dominates in a distinct regime depending on the grid density. 
In this setting, SB Data-Driven is seen to match the power of the best-performing merging function across all regimes.
The right panel also reveals critical dynamics regarding the bandwidth selection process: the initial, narrow bandwidth grid is poorly calibrated, yielding near-zero power. 
As we increase the total number of bandwidths $K$, the grid broadens to encompass better-calibrated bandwidths, resulting in a sharp, immediate increase in empirical power. 
However, continuing to increase $K$ inevitably introduces poorly calibrated, noisy bandwidths into the collection, which begins to penalize the overall power. 
Crucially, as $K$ becomes increasingly large, the power of the SB procedure stabilizes and reaches a robust plateau, whereas the power of the conservative WC procedure deteriorates completely to zero.

\begin{figure}[!t]
	\centering
	\includegraphics[width=\textwidth]{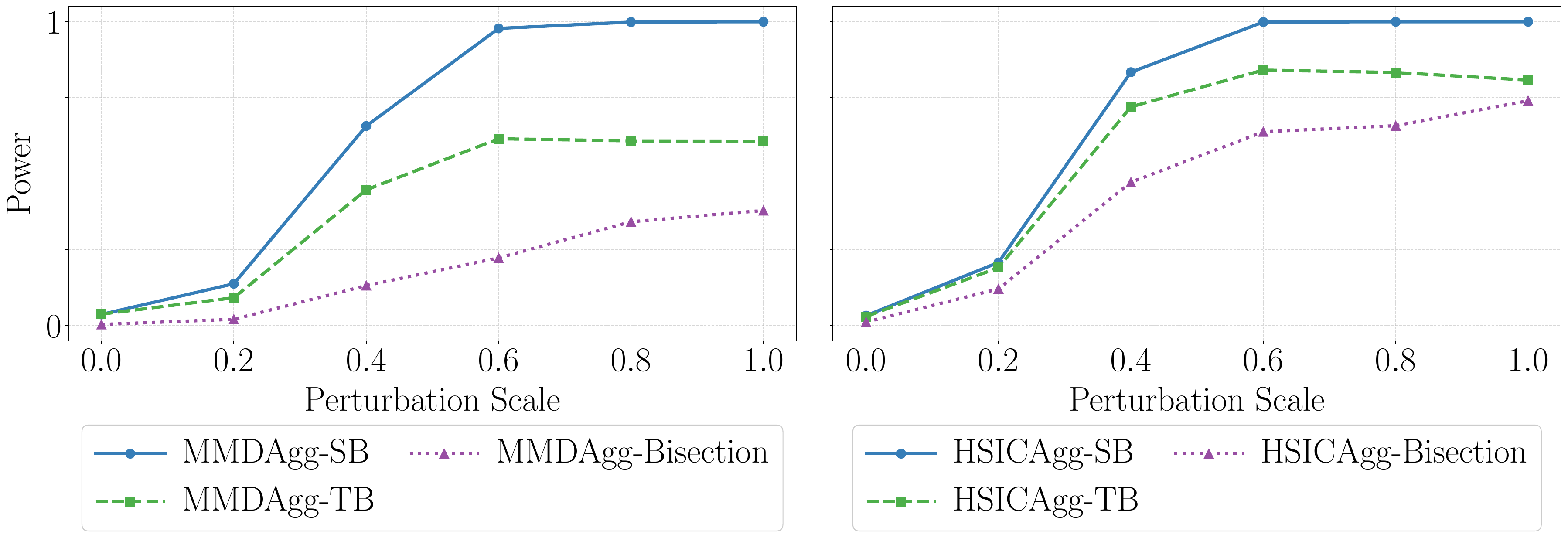}
	\caption{Power estimation for aggregated kernel-based MMDAgg and HSICAgg nonparametric testing.}
	\label{fig: figure_7}
\end{figure}

\subsection{Improved MMDAgg and HSICAgg optimal tests}
\label{subsec: mmd_hsic_power}

In the experimental setting of \Cref{fig: figure_7}, we implement the SB versions of the MMDAgg and HSICAgg tests \citep{schrab2023mmd,albert2022adaptive} based on unbiased U-statistics.
These implementations are exact: they correspond strictly to the SB tests analyzed in \Cref{sec: adaptive MMD and HSIC tests}, requiring no further approximations or empirical parameter tuning.
In particular, the collections of bandwidths for Gaussian kernels are implemented exactly as prescribed by theory (\ref{eq:optimal_bandwidths_mmd}, \ref{eq:optimal_bandwidths_hsic}).\footnote{The isotropic Sobolev assumption implies uniform smoothness across all dimensions; therefore, we use the same bandwidth for both kernels in the HSIC computation. We evaluate the more general setting of kernel-specific bandwidths using a two-dimensional grid in \Cref{subsec: independence}.}
These bandwidth collections are designed for the Sobolev smoothness assumption, which is satisfied by the data-generating process in this experiment as we use specially-constructed perturbed uniform distributions.
For two-sample testing, the aim is to detect the difference between the perturbed and unperturbed one-dimensional uniform distributions, while for independence testing the goal is to detect the dependence between the two dimensions of a joint perturbed uniform distribution (with uniform marginals).
We refer the reader to \cite[Eq. 17]{schrab2023mmd} and \cite[Eq. 4.3]{albert2022adaptive} for specific expressions of the perturbed uniform distributions.
We control the signal strength by varying the perturbation scale from 0 (the null hypothesis) to 1 (the theoretical maximum for non-negative densities), while maintaining a sample size of $500$, a test level of $\alpha=0.05$ and a permutation number of $B=199$ across both experiments.
We refer to \cite{domingo2025cheap} for a method to significantly reduce the computational cost of permutations while maintaining high power.

In \Cref{fig: figure_7}, we compare the empirical power (averaged over 1000 independent repetitions) of the SB tests and of their TB variants against the original MMDAgg and HSICAgg tests, which are defined via their statistic view with rejection event (\ref{Eq. max-type test}) and Monte Carlo calibration threshold (\ref{eq: estimated u alpha}) where the supremum is being further approximated via a bisection method.
As we show in \Cref{prop: u_alpha expression}, the threshold (\ref{eq: estimated u alpha}) can actually be computed exactly, bypassing the estimator error induced by the bisection approximation.
This corresponds to the MaxT variants of the tests (\Cref{sec: maxT aggregation}). However, as shown in \Cref{prop:failure of size control without bisection}, these do not guarantee finite-sample level control across all settings, and are hence excluded from our experiments, while the more conservative bisection-based MaxT variants are included for comparison with \cite{schrab2023mmd,albert2022adaptive}.
The proposed TB variants (with minimum merging function) of MMDAgg and HSICAgg resolve this lack of finite-sample validity while remaining closely related to their MaxT counterparts (see \Cref{sec: deferred maxT comparison}), benefit from a p-value formulation and achieve higher empirical power.
The computationally more efficient SB variants are observed to be statistically more powerful, validating our theory (see \Cref{sec: power properties of TB}).

Indeed, for both the MMDAgg and HSICAgg comparisons of \Cref{fig: figure_7}, the SB method strictly outperforms the TB approach, which in turn outperforms the standard MaxT-Bisection method.
While this power advantage is significant when using $B=199$ permutations, it is expected that all three tests perform comparably when $B$ is sufficiently large.
In practice, the implementations of \cite{schrab2023mmd,albert2022adaptive} use many more permutations (e.g., $B=2000$) to stabilize the power.
Consequently, the superiority of the SB procedure can be interpreted dually: it achieves strictly higher power for a fixed small number of permutations, or equivalently, it drastically reduces the computational cost required to match the power of the Bisection method.
We encourage practitioners to use the newly proposed MMDAgg-SB and HSICAgg-SB versions of the tests, whose minimax optimality and validity are established in \Cref{sec: adaptive MMD and HSIC tests}.

\section{Extensions of the aggregation framework}
\label{sec: additional theoretical extensions}

This section collects extensions of the basic SB and TB aggregation framework beyond the core p-value aggregation procedures.

\subsection{Exchangeable e-values}
\label{sec: deferred exchangeable e-values}

Although our primary focus is on p-value aggregation, the SB aggregation framework extends naturally to e-values \citep{vovk2021values}, which quantify evidence via non-negative random variables $E$ satisfying $\mathbb{E}[E]\le 1$ under the null. While simple averaging of e-values preserves validity, it can be conservative when the expectation of the average is strictly smaller than one. Under the group-invariance hypothesis, however, additional structure is available: the row-wise merged quantities $(f_0,\ldots,f_B)$ are exchangeable in the index $b$. Exploiting this exchangeability, for any measurable score map $\psi\colon\mathbb{R}\to[0,\infty)$, the self-normalized statistic
\begin{align*}
E_{\mathrm{SB}}
\coloneqq
\frac{(B+1)\psi(f_0)}{\sum_{b=0}^B \psi(f_b)},
\qquad
E_{\mathrm{SB}}=1 \text{ if } \sum_{b=0}^B \psi(f_b)=0,
\end{align*}
is an e-value satisfying $\mathbb{E}[E_{\mathrm{SB}}]=1$ under the null. Thus, exchangeability yields an exactly calibrated e-value that avoids the potential conservativeness of generic averaging schemes. This construction corresponds to the \emph{soft-rank e-value} \citep[][Chapter~1.7]{ramdas2025hypothesis}.

\subsection{Data-driven aggregation of merging functions} \label{sec: data-driven aggregation of merging functions}

This subsection extends the aggregation framework from combining statistics to combining the merging rules themselves. Building on this perspective, we develop a data-driven aggregation scheme that exploits redundancy among merging rules to retain the power of the best-performing rule while preserving finite-sample validity. A related strategy has been explored by \citep{guo2024rank} in the context of subsampling-based aggregation; however, a formal theoretical power analysis has been lacking.
While we formulate this data-driven scheme within the single-batch (SB) framework for clarity, the underlying methodology readily extends to the two-batch (TB) procedure.

\begin{figure}[htbp]
	\centering
	\includegraphics[width=\textwidth]{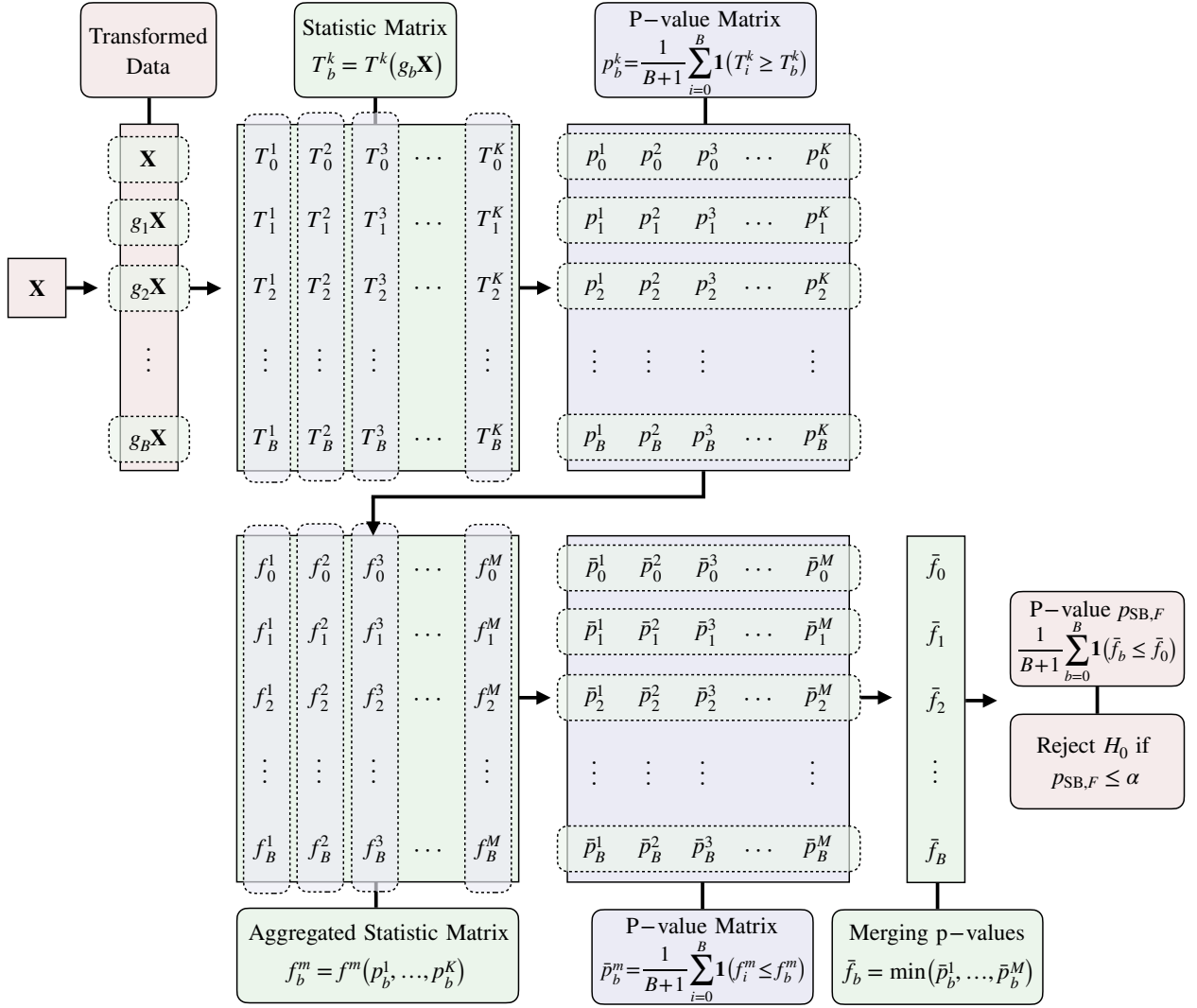}
	\caption{Schematic illustration of the Data-driven SB aggregation procedure of \Cref{sec: data-driven aggregation of merging functions} (\Cref{alg:DDSB}), using statistics $T^1,\dots,T^K$ and merging functions $f^1,\dots,f^M$.}
	\label{fig: diagram_SBF}
\end{figure}

\begin{algorithm}[t]
\caption{Data-Driven Single-Batch Aggregation}
\label{alg:DDSB}
\begin{algorithmic}[1]
\Require data $\mathbf{X}$; statistics $T^1,\ldots,T^K$; transformations $g_0,\ldots,g_B$ ($g_0=\mathrm{id}$); merging functions $f^1,\dots,f^M$; level $\alpha$.
\State \textbf{Transformed statistics:} For each $b\in[B]_0$ and $k\in[K]$, compute $T_b^k \gets T^k(g_b(\mathbf{X}))$.
\State \textbf{Standardization via p-values:} For each $b\in[B]_0$ and $k\in[K]$, compute the permutation p-value
\begin{align*}
p\big(T_b^k\big) \gets \frac{1}{B+1}\sum_{i=0}^B \mathds{1}\big(T_i^k \ge T_b^k\big).
\end{align*}
\State \textbf{Multiple aggregation row-wise:} For each $b\in[B]_0$ and $m\in[M]$, set
\begin{align*}
f_b^m \gets f^m\bigl(p\big(T_b^1\big),\ldots,p\big(T_b^K\big)\bigr).
\end{align*}
\State \textbf{Standardization via p-values:} For each $b\in[B]_0$ and $m\in[M]$, compute the permutation p-value
\begin{align*}
\ple\big(f_b^m\big) \gets \frac{1}{B+1}\sum_{i=0}^B \mathds{1}\big(f_i^m \le f_b^m\big).
\end{align*}
\State \textbf{Aggregation row-wise:} For each $b\in[B]_0$, set
\begin{align*}
p_b^\mathrm{min} \gets \min_{m\in[M]}\ple\big(f_b^m\big).
\end{align*}
\State \textbf{Calibration:} Compute the p-value
\begin{align*}
p_{\mathrm{SB},\mathcal{F}} \gets \frac{1}{B+1}\sum_{b=0}^B \mathds{1}\big(p_b^\mathrm{min} \le p_0^\mathrm{min}\big).
\end{align*}
\State \textbf{Decision rule:} Reject if $p_{\mathrm{SB},\mathcal{F}}\le\alpha$.
\end{algorithmic}
\end{algorithm}

\parheading{Procedure.}
Let $\mathcal{F} = \{f^1,\ldots,f^M\}$ denote a collection of merging functions, such as the minimum, median, mean, or other order-statistic--based rules. For each $m\in[M]$, apply the SB aggregation procedure of \Cref{sec:SB-procedure} to obtain merged statistics
\begin{align*}
f_0^m, f_1^m, \ldots, f_B^m.
\end{align*}
Importantly, for each $b\in[B]_0$, the values $f_b^1,\ldots,f_b^M$ are obtained by applying the merging functions in $\mathcal{F}$ to the same collection of permutation p-values computed from the $b$-th transformed dataset. With a slight abuse of notation, define the associated lower-tail permutation p-values
\begin{align*}
\ple(f_b^m)
\coloneqq
\frac{1}{B+1}\sum_{i=0}^B \mathds{1}\bigl(f_i^m \le f_b^m\bigr),
\qquad b\in[B]_0.
\end{align*}
For each transformation index $b$, aggregate across merging functions by taking the minimum p-value
\begin{align*}
p_b^{\min}
\coloneqq
\min_{m\in[M]} \ple(f_b^m),
\qquad b\in[B]_0.
\end{align*}
The data-driven aggregation test is then defined by
\begin{align*}
p_{\mathrm{SB},\mathcal{F}}
\coloneqq
\frac{1}{B+1}\sum_{b=0}^B
\mathds{1}\bigl(p_b^{\min} \le p_0^{\min}\bigr),
\end{align*}
and rejects the null hypothesis whenever $p_{\mathrm{SB},\mathcal{F}}\le\alpha$. Finite-sample validity follows directly from the exchangeability of $(p_b^{\min})_{b\in[B]_0}$ under the group-invariance hypothesis. This procedure is fully presented in \Cref{alg:DDSB} and is illustrated in \Cref{fig: diagram_SBF}.

To study the power of the proposed data-driven aggregation procedure, we first formalize the notion that, in many applications, only a small number of merging functions are effectively distinct. This captures redundancy across merging rules through the ranks they induce on the same transformed rows.

\parheading{Effective multiplicity.}
For $\varepsilon \ge 0$, define the (random) effective multiplicity
\begin{align*}
N_{\mathrm{eff}}(\varepsilon)
=
\min\Bigl\{|S|:\ S\subseteq[M],\
\min_{m\in S} \ple(f_b^m)\le (1+\varepsilon)\min_{m\in[M]} \ple(f_b^m)\ \ \forall b\in[B]_0
\Bigr\}.
\end{align*}
Thus, $N_{\mathrm{eff}}(\varepsilon)$ is the smallest number of representative merging functions whose minimum p-value approximates, up to a multiplicative factor $1+\varepsilon$, the minimum over all $M$ merging functions.

We assume that there exist deterministic constants $N \le M$ and $\eta \in [0,1]$ such that
\begin{equation}
\label{eq: Neff_highprob}
\mP( N_{\mathrm{eff}}(\varepsilon) \le N ) \ge 1-\eta.
\end{equation}
Condition~\eqref{eq: Neff_highprob} is satisfied whenever the collection $\{\ple(f_b^m)\}_{m\in[M]}$ exhibits redundancy, for instance when many merging functions induce nearly identical ranks on the same transformation. In the extreme case where all merging functions induce identical ranks, one has $N_{\mathrm{eff}}(0)=1$ almost surely. 

The following proposition quantifies the power of the data-driven aggregation procedure under this effective multiplicity condition.

\begin{proposition}
\label{prop: adaptive aggregation effective multiplicity}
Fix $\alpha\in(0,1)$ and $\varepsilon\ge0$. Suppose that the effective multiplicity condition \eqref{eq: Neff_highprob} holds for some $N\le M$ and $\eta\in[0,1]$. Then the type~II error of the data-driven aggregation procedure satisfies
\begin{align*}
\mP\bigl(p_{\mathrm{SB},\mathcal{F}}>\alpha\bigr)
\le
\eta
+ 
\mP\!\left( \min_{m\in[M]}
 \ple(f_0^{m})>a_{\varepsilon,N}
\right)
\le
\eta
+ 
\min_{m\in[M]}
\mP\!\left(
 \ple(f_0^{m})>a_{\varepsilon,N}
\right)
\end{align*}
for $a_{\varepsilon,N}\coloneqq {\alpha}/{(N(1+\varepsilon))}$.
\end{proposition}
\begin{proof} See \Cref{sec: proof of adaptive aggregation}. \end{proof}

\Cref{prop: adaptive aggregation effective multiplicity} shows that the power of the data-driven aggregation procedure is comparable to that of the best merging function in the collection, up to a multiplicative factor $(1+\varepsilon)N$ in the significance level and an additive error $\eta$. In particular, the procedure behaves as if only $N$ effectively distinct merging functions were considered. As a special case, taking $\varepsilon=0$, $\eta=0$, and $N=M$ yields
\begin{align*}
\mP\bigl(p_{\mathrm{SB},\mathcal{F}}>\alpha\bigr)
\le
\mP\!\left(\min_{m\in[M]} \ple(f_0^m)>\frac{\alpha}{M}\right)
\le
\min_{m\in[M]}
\mP\!\left(\ple(f_0^m)>\frac{\alpha}{M}\right),
\end{align*}
which coincides with the classical Bonferroni-type bound obtained by calibrating each merging function at level $\alpha/M$. Existing sufficient conditions under which a single permutation p-value is powerful, such as those studied in \citep{kim2022minimax,kim2025differentially}, can be invoked to further instantiate and interpret the power bound above.

\parheading{Uniform power improvement over single merging rules.}
\Cref{prop: adaptive aggregation effective multiplicity} quantifies the power of the data-driven aggregation test in terms of redundancy among the candidate merging functions, as captured by the effective multiplicity $N_{\mathrm{eff}}(\varepsilon)$. This perspective is particularly informative when many rules are nearly equivalent (so that $N_{\mathrm{eff}}(\varepsilon)\ll M$). At the same time, even when the candidate rules are all effectively distinct (so that $N_{\mathrm{eff}}(0)=M$), the permutation-calibrated aggregation can strictly outperform \emph{every} fixed choice of a single merging function under heterogeneous (e.g., mixture) alternatives, where different rules are optimal in different sub-regimes. We delve into this phenomenon in more detail in \Cref{sec: uniform power improvement over single merging rules}.

\subsection{SB aggregation of data-asymmetric tests} \label{sec: aggregation of data-asymmetric tests}

We next specialize the preceding SB framework to settings in which the test statistic is
\emph{asymmetric} in the data. Such asymmetry arises when a portion of the sample is used to construct a data-dependent object that is subsequently evaluated on held-out data, so that the resulting test depends on a particular ordering or partition of the sample. Examples include both early data-splitting procedures \citep{cox1975note,moran1973dividing} and more recent nonparametric tests involving learned or data-adaptive components
\citep[e.g.,][]{jankova2020goodness,liu2020learning,kim2021classification,tansey2022holdout,shekhar2022permutation,shekhar2023permutation,chau2024credal,pogodin2024practical,lundborg2024projected,kim2024dimension,kim2024conditional}. In such settings, it is natural to repeat the test over multiple random (or predetermined) splits and to aggregate the resulting p-values to improve power and stability.

Let $\Delta_1,\ldots,\Delta_K$ denote tests obtained by applying a fixed asymmetric testing procedure to $K$ independent random orderings or partitions of the data, and let $p^{(k)}$ be the p-value associated with $\Delta_k$. Even when the vector $(p^{(1)},\ldots,p^{(K)})$ is exchangeable under the null hypothesis, the arithmetic mean $\bar p \coloneqq K^{-1}\sum_{k=1}^K p^{(k)}$ is not, in general, super-uniform \citep{choi2023averaging,gasparin2025combining}, and therefore cannot be used directly as a valid p-value. While a classical result of
\citep{rueschendorf1982} shows that the rescaled statistic $2\bar p$ is super-uniform, this worst-case correction is often conservative and may substantially reduce power in practice.

Assume now that the group-invariance hypothesis holds. For each transformation index $b\in[B]_0$ and repetition $k\in[K]$, let $p_b^{(k)}$ denote the permutation p-value obtained by applying the same transformation to all $K$ repetitions. Define the row-wise averages $\bar p_b \coloneqq K^{-1}\sum_{k=1}^K p_b^{(k)}$ for $b\in[B]_0$, and define the SB average p-value by $p_{\mathrm{SB},\mathrm{avg}}
\coloneqq
(B+1)^{-1}\sum_{b=0}^B \mathds{1}\!\left(\bar p_b \le \bar p_0\right)$. By \Cref{prop: SB level}, this p-value is valid under the group-invariance hypothesis. Moreover, \Cref{thm:single-batch power} and \Cref{prop:asymptotic-adaptation} imply that it is structure-adaptive and uniformly dominates the worst-case corrected average $2\bar p_0$ in terms of power. It is also worth noting that exploiting the additional finite-population structure inherent in the SB construction, the worst-case correction factor can be sharpened from $2$ to $2(B+1)/(B+2)$; see \Cref{cor: SB avg finite population}. Taken together, these results establish SB aggregation as a principled approach for aggregating asymmetric tests under the group-invariance hypothesis.

\parheading{Consistency transfer from individual tests.}
Crucially, the uniform power dominance established in \Cref{thm:single-batch power} has an immediate implication for power consistency. Since \Cref{thm:single-batch power} shows that the type~II error of the SB procedure is uniformly upper bounded by that of any deterministic worst-case calibrated test, any consistency guarantee established for the latter automatically transfers to the SB procedure as shown below.

\begin{corollary}
\label{cor: SB-consistency} Let $K=K_n$ be an arbitrary (possibly diverging) sequence, and let $p_{0,n}^1,\ldots,p_{0,n}^{K_n}$ be exchangeable permutation p-values indexed by $n$. Assume that, for some (and hence also for all) 
$k\in[K_n]$,
\begin{align*}
\sup_{P \in \mathcal{P}}\mP_P\bigl(p_{0,n}^k > \alpha\bigr) \;\rightarrow\; 0
\quad \text{for every } \alpha \in (0,1).
\end{align*}
Then the SB average aggregation test satisfies
\begin{align*}
\sup_{P \in \mathcal{P}}\mP_P\bigl(p_{\mathrm{SB},\mathrm{avg}} > \alpha\bigr)
\;\rightarrow\; 0
\quad \text{for every } \alpha \in (0,1),
\end{align*}
without any restriction on the growth rate of $K_n$.
\end{corollary}
\begin{proof} See \Cref{sec: proof of SB-consistency}. \end{proof}

The converse does not hold in general: when the individual tests are (nearly) independent, the SB average test can be consistent even if none of the individual tests are.

\section{Power refinements}
\label{sec: additional results}
This section collects refinements of the power and adaptivity analysis for the proposed aggregation procedures.

\subsection{Uniform asymptotic adaptation}\label{sec:uniform-asymptotic-adaptation}
Here, we present a version of the asymptotic adaptivity result in \Cref{prop:asymptotic-adaptation} which holds uniformly over a class of null distributions.

\begin{proposition}
\label{prop:uniform-asymptotic-adaptation}
Let $\mathcal P$ be a class of null distributions. Under the setup of \Cref{prop:asymptotic-adaptation}, assume that for each $P\in\mathcal P$, the pointwise convergence conditions hold with $F$ replaced by $F_P$, and that
\begin{align*}
\sup_{P\in\mathcal P}\sup_{t,s\in\mathbb R}
\bigl|
\mP_P(f_{1,n}\le t, \, f_{2,n}\le s)
-
F_P(t)F_P(s)
\bigr|
\;\rightarrow \; 0.
\end{align*}
Write $Q^\star_{\alpha,P}\coloneqq\inf\{u \in \mathbb{R}:\,F_P(u)\ge\alpha\}$, and assume the uniform margin condition: for every $\varepsilon>0$,
\begin{align*}
\inf_{P\in\mathcal P}
\min\!\left\{
\alpha-F_P(Q^\star_{\alpha,P}-\varepsilon),\, F_P(Q^\star_{\alpha,P}+\varepsilon)-\alpha
\right\}
>0.
\end{align*}
Then, for every $\varepsilon>0$,
\begin{align*}
\sup_{P\in\mathcal P}
\mP_P\bigl(|\uSB-Q^\star_{\alpha,P}|>\varepsilon\bigr)
\;\rightarrow\; 0.
\end{align*}
\end{proposition}
\begin{proof}
See \Cref{sec: proof of uniform asymptotic adaptation}.
\end{proof}

\subsection{Uniform power improvement over single merging rules} \label{sec: uniform power improvement over single merging rules}

In this subsection, we show that SB aggregation of multiple merging rules (\Cref{sec: data-driven aggregation of merging functions}) can achieve uniform power improvement over each individual rule. The key mechanism is that aggregation can succeed whenever at least one candidate rule provides sufficiently strong evidence, while incurring only a limited calibration penalty. We start with a simple lemma. 

\begin{lemma}
\label{lem:rank1_hit_implies_reject}
Suppose that the permutation p-values are defined by
\begin{align*}
\ple(f_b^m)=\frac{1}{B+1}\sum_{i=0}^B \mathds{1}\{f_i^m \le f_b^m\},
\qquad b\in[B]_0,\ m\in[M].
\end{align*}
Suppose that the minimum permutation p-value across merging rules
attains the smallest possible grid value $1/(B+1)$, that is,
\begin{equation}
\label{eq:rank1_hit_event}
p_0^{\min}=\min_{m\in[M]}\ple(f_0^m)=\frac{1}{B+1}.
\end{equation}
In this case, the SB aggregation p-value admits the representation
\begin{align*}
p_{\mathrm{SB},\mathcal{F}}=\frac{W}{B+1},
\qquad
W \le M,
\end{align*}
where
\begin{align*}
W \coloneqq \bigl|\{b\in[B]_0:\ p_b^{\min}=1/(B+1)\}\bigr|
\end{align*}
denotes the number of permutation indices attaining the minimal grid value. As a result, whenever $M \le \lfloor (B+1)\alpha \rfloor$, the aggregated test necessarily rejects.
\end{lemma}

\begin{proof}
 See \Cref{sec: proof of rank1_hit_implies_reject}.
\end{proof}

\Cref{lem:rank1_hit_implies_reject} shows that if \emph{any} candidate merging rule attains the smallest possible permutation p-value, namely $1/(B+1)$, on the observed data, then the second-stage permutation calibration incurs a penalty of at most $M/(B+1)$. In particular, if $B$ is large enough so that $M/(B+1)\le \alpha$, then this event alone guarantees rejection by the aggregated test.

The next proposition demonstrates that, under a mixture alternative in which different components favor different merging rules, the aggregated test can achieve power one, even though no single merging rule is uniformly powerful across all components. 

\begin{proposition}
\label{prop:aggregation_strictly_better_than_each_rule}
Fix $\alpha\in(0,1)$ and suppose that $M \le \lfloor (B+1)\alpha\rfloor$. Consider an alternative distribution of the data that is a finite mixture
\begin{align*}
P = \sum_{j=1}^J \pi_j P_j, \quad \text{where }\quad
\pi_j > 0, \; \sum_{j=1}^J \pi_j = 1.
\end{align*}
such that for each component $P_j$ there exists an index $m(j)\in[M]$ satisfying
\begin{equation}
\label{eq:component_rank1_hit}
\mP_{P_j}\!\left(\ple\bigl(f_0^{m(j)}\bigr)=\frac{1}{B+1}\right)=1,
\end{equation}
and moreover no single merging function succeeds on all components, in the sense that for every fixed $m\in[M]$
there exists some $j\in[J]$ with $\pi_j>0$ such that
\begin{equation}
\label{eq:component_failure_for_fixed_m}
\mP_{P_j}\!\left(\ple\bigl(f_0^{m}\bigr)\le \alpha\right) < 1.
\end{equation}
Then the aggregated test has power one under the mixture:
\begin{align*}
\mP_{P}\bigl(p_{\mathrm{SB},\mathcal{F}}\le \alpha\bigr)=1,
\end{align*}
whereas every fixed-rule SB test has strictly smaller mixture power:
\begin{align*}
\mP_{P}\!\left(\ple\bigl(f_0^{m}\bigr)\le \alpha\right) < 1
\qquad \text{for all } m\in[M].
\end{align*}
\end{proposition}

\begin{proof} 
	See \Cref{sec: proof of aggregation strictly better than each rule}.
\end{proof}

\subsection{Power properties of TB aggregation} 
\label{sec: power properties of TB}

The analysis of the TB aggregation procedure differs from that of SB aggregation since the testing and reference batches are distinct, and hence the deterministic super-uniformity argument of \Cref{lemma:super-uniform-property}, which underpins the power analysis of SB aggregation, no longer applies. Nevertheless, as the number of transformations $B$ increases, we show that the SB and TB thresholds become asymptotically close. This follows from the fact that, under i.i.d.~uniform sampling of transformations, both procedures converge to the idealized procedure that enumerates all transformations in $\mathcal{G}$. Although threshold convergence alone does not imply identical power, it plays a crucial asymptotic role. Under mild regularity conditions, such as continuity of the limiting distribution at the threshold, it ensures that the power guarantees established for SB aggregation carry over to TB aggregation as $B \to \infty$.

The following proposition formalizes the asymptotic threshold equivalence and provides an exponential-in-$B$ control.

\begin{proposition}
\label{prop: SB-TB equivalence rate}
Assume that the SB and TB procedures use the same continuous merging function $f:[0,1]^K\to\mathbb R$, and that the number of coordinates $K$ is fixed. Let $\uSB$ and $\uTB$ denote the SB and TB thresholds defined in
\eqref{eq: u_alpha general expression any function} and \eqref{eq: uTB_sup_form}, respectively. Assume that, conditional on the observed data $\mathbf X$, the transformations $g_1,\ldots,g_{2B}$ are i.i.d.\ uniform draws from $\mathcal{G}$. Then for every $\varepsilon>0$, there exist constants $C(\varepsilon,f,K)>0$ and $c(\varepsilon,f)>0$ such that
\begin{align*}
\mP\bigl(|\uTB-\uSB|>\varepsilon\bigr)
\le C(\varepsilon,f,K)\,e^{-c(\varepsilon,f) B},
\qquad \text{for all sufficiently large } B.
\end{align*}
In particular, $\uTB-\uSB \overset{p}{\longrightarrow} 0$ at an exponential rate in $B$ uniformly over all data-generating distributions for $\mathbf X$.
\end{proposition}
\begin{proof} See \Cref{sec: proof of SB-TB equivalence rate}. 
\end{proof}

Although the SB and TB thresholds become asymptotically equivalent as $B\to\infty$, their finite-$B$ power may differ. In particular, TB can be less powerful than SB even though it uses an additional batch of transformations. The loss arises from the fact that TB \emph{holdout} p-values are constructed without including the observed statistic $T_0^k$ in the ranking set, whereas the SB permutation p-values \emph{do} include $T_0^k$.  When $B$ is small, this distinction can be decisive: under a strong signal, the inclusion of $T_0^k$ systematically inflates the SB calibration p-values away from the smallest grid point $1/(B+1)$, thereby increasing the SB critical value and facilitating rejection. By contrast, TB calibration p-values can still hit the smallest grid point with non-vanishing probability, creating ties that obstruct rejection under a strict comparison.

We now formalize this mechanism via an explicit construction, showing that, for fixed and small $B$, SB aggregation can achieve asymptotic power one whereas TB aggregation cannot.

\begin{proposition}
\label{prop:SBvsTB:finiteB_gap_main}
Fix an integer $B \ge 1$ and let $\alpha \in (1/(B+1),1)$. There exist a sequence of alternatives and a merging function for which, with $B$ fixed, the SB aggregation test achieves asymptotic power one, whereas the TB aggregation test has asymptotic power strictly bounded away from one. More precisely, under a strong-signal regime and some regularity conditions specified in \Cref{app:SBvsTB_finiteB_gap},
\begin{align*}
\mP(p_{\mathrm{TB}} \le \alpha)
\rightarrow
\mP\!\left(
W \le \lfloor (B+1)\alpha \rfloor - 1
\right)
\qquad \text{as } n \to \infty,
\end{align*}
where $W$ follows a negative hypergeometric distribution with probability mass function
\begin{align*}
\mP(W = k)
=
\frac{\binom{2B-k-1}{B-1}}{\binom{2B}{B}},
\qquad k = 0,1,\ldots,B.
\end{align*}
In particular, this limit is strictly smaller than one for all $B \ge 1$ and $\alpha \in (1/(B+1),1)$. By contrast, under the same conditions,
\begin{align*}
\mP(p_{\mathrm{SB}} \le \alpha)
\rightarrow
1
\qquad \text{as } n \to \infty.
\end{align*}
\end{proposition}
\begin{proof} See \Cref{app:SBvsTB_finiteB_gap}. \end{proof}
For illustration, when $1/(B+1) < \alpha < 2/(B+1)$, we have $\lfloor (B+1)\alpha \rfloor = 1$ and hence the limiting TB power reduces to $\mP(W=0) = 1/2$, which holds for all $B \ge 1$. This shows that, for fixed $B$, the TB construction can suffer a non-vanishing power loss relative to SB aggregation, despite using an additional batch of transformations.

\subsection{Worst-case bound for quasi-arithmetic mean aggregation} \label{sec: Worst-case bound for quasi-arithmetic mean aggregation}

In this subsection, we derive finite-sample worst-case bounds for quasi-arithmetic mean aggregation under the SB construction, and quantify how the resulting calibration constants improve upon existing super-uniform bounds by exploiting the permutation structure.

\begin{lemma} 
\label{lem: SB quasi arithmetic finite population}
Let $\{p_b^{(k)} : b\in[B]_0,\, k\in[K]\}$ be permutation p-values constructed as in \eqref{Eq: individual p-value}. Let $\phi:(0,1]\to\mathbb{R}$ be a strictly increasing and continuous function. Define the row-wise quasi-arithmetic mean
\begin{align*}
\bar p_{\phi,b}
\coloneqq
\phi^{-1}\left(\frac{1}{K}\sum_{k=1}^K \phi\left(p_b^{(k)}\right)\right),
\qquad b\in[B]_0.
\end{align*}
For each $\ell\in\{1,\ldots,B+1\}$, define the deterministic thresholds
\begin{align*}
t_{\phi,B}(\ell)
\coloneqq
\phi^{-1}\!\left(\frac{1}{\ell}\sum_{j=1}^{\ell}\phi\left(\frac{j}{B+1}\right)\right).
\end{align*}
Define the calibration function $G_{\phi,B}:(0,1]\to[0,1]$ by
\begin{align*}
G_{\phi,B}(t)
\coloneqq
\frac{1}{B+1}\max\Bigl\{\ell\in [B+1]_0: t_{\phi,B}(\ell)\le t\Bigr\},
\end{align*}
with the convention $t_{\phi,B}(0)\coloneqq 0$. Then, for any $t\in(0,1]$, it holds deterministically that
\begin{align*}
\frac{1}{B+1}\sum_{b=0}^{B}\mathds{1}\!\left(\bar p_{\phi,b}\le t\right)
\le G_{\phi,B}(t).
\end{align*}
Consequently, under the group-invariance hypothesis,
\begin{align*}
\mP\!\left(\bar p_{\phi,0}\le t\right) \le G_{\phi,B}(t),\qquad t\in(0,1].
\end{align*}
In particular, $G_{\phi,B}(\bar p_{\phi,0})$ is a valid p-value.
\begin{proof}
	See \Cref{sec: proof of SB quasi arithmetic finite population}. 
\end{proof}
\end{lemma}

\Cref{lem: SB quasi arithmetic finite population} recovers the power-mean aggregation scheme by choosing $\phi(x)=x^r$ ($r>0$). In this case, it is useful to obtain an explicit \emph{linear} upper bound on the stepwise calibration function $G_{\phi,B}$, as such bounds directly yield simple and interpretable worst-case corrections for the aggregated p-values.

Accordingly, for $r>0$ we define
\begin{align*}
c_{r,B}
\coloneqq
\max_{1\le \ell\le B+1}
\frac{\ell}{\left(\frac{1}{\ell}\sum_{j=1}^{\ell} j^r\right)^{1/r}}.
\end{align*}
This constant is chosen so that
\begin{align*}
G_{\phi,B}(t) \le c_{r,B}\,t
\qquad\text{for all } t\in(0,1].
\end{align*}
Indeed, recall that the calibration function $G_{\phi,B}$ is defined by
\begin{align*}
G_{\phi,B}(t)
=
\frac{1}{B+1}\max\Bigl\{\ell\in [B+1]_0: t_{\phi,B}(\ell)\le t\Bigr\}.
\end{align*}
Thus, if $t_{\phi,B}(\ell)\le t$ for some $\ell$, then $\ell/(B+1)\le G_{\phi,B}(t)$, and for the maximal admissible index $\ell^\star(t)$ we have $G_{\phi,B}(t)=\ell^\star(t)/(B+1)$. Consequently, any inequality of the form
\begin{align*}
\frac{\ell}{B+1}\le C\,t
\qquad\text{whenever } t_{\phi,B}(\ell)\le t
\end{align*}
immediately implies the linear bound $G_{\phi,B}(t)\le C\,t$.

For $\phi(x)=x^r$ with $r>0$, we have
\begin{align*}
t_{\phi,B}(\ell)
=
\frac{1}{B+1}\left(\frac{1}{\ell}\sum_{j=1}^{\ell} j^r\right)^{1/r},
\end{align*}
and the smallest constant $C$ for which the above inequality holds uniformly over $\ell\in\{1,\ldots,B+1\}$ is precisely $c_{r,B}$.

It is instructive to contrast the constant $c_{r,B}$ with its worst-case counterpart in \cite{vovk2020combining}. When only super-uniformity of the aggregated p-value is assumed, without any additional structural information, the optimal worst-case linear bound for the power mean is given by the universal constant $c_r=(r+1)^{1/r}$, which is sharp in the class of all super-uniform random variables. In the present setting, however, the aggregated quantities arise from \emph{permutation p-values} and therefore inherit a finite-population structure that is not captured by super-uniformity alone. This additional structure allows the worst-case correction factor to be strictly improved at finite $B$. The constant $c_{r,B}$ quantifies the optimal linear bound that exploits this permutation structure. The following lemma shows that $c_{r,B}$ is strictly smaller than $c_r$ for any finite $B$, while converging to $c_r$ as $B\to\infty$.

\begin{lemma}
\label{lemma: compare cr and crB}
Fix $r>0$ and recall the constant 
$c_{r,B}$ defined above. Let $c_r\coloneqq (r+1)^{1/r}$. Then:
\begin{enumerate}
\item[(i)] For every $B\ge 1$, $c_{r,B}<c_r$.
\item[(ii)] The sequence $B\mapsto c_{r,B}$ is nondecreasing and
\begin{align*}
c_{r,B}\ \uparrow\ c_r
\qquad\text{as } B\to\infty.
\end{align*}
\end{enumerate}
\begin{proof}
	See \Cref{sec: proof of compare cr and crB}.
\end{proof}
\end{lemma}

As a direct consequence of \Cref{lem: SB quasi arithmetic finite population} and
\Cref{lemma: compare cr and crB}, we obtain the following worst-case linear bound for arithmetic-mean aggregation of permutation p-values.

\begin{corollary}
\label{cor: SB avg finite population}
Let $\{p_b^{(k)} : b\in[B]_0,\, k\in[K]\}$ be permutation p-values constructed as in \eqref{Eq: individual p-value}. For each $b\in[B]_0$, define the row-wise average
\begin{align*}
\bar p_b \coloneqq \frac{1}{K}\sum_{k=1}^K p_b^{(k)} .
\end{align*}
Then, for any $t\in[0,1]$, it holds deterministically that
\begin{align*}
\frac{1}{B+1}\sum_{b=0}^{B} \mathds{1}\!\left(\bar p_b \le t\right)
\;\le\;
\min\!\left\{1,\; \frac{2(B+1)}{B+2}\, t \right\}.
\end{align*}
\begin{proof}
Apply \Cref{lem: SB quasi arithmetic finite population} with $\phi(x)=x$, so that $\bar p_{\phi,b}=\bar p_b$ and $t_{\phi,B}(\ell)=\frac1\ell\sum_{j=1}^\ell \frac{j}{B+1}=\frac{\ell+1}{2(B+1)}$. Hence the associated linear bound constant is $c_{1,B}=\max_{1\le \ell\le B+1}\frac{\ell}{\frac1\ell\sum_{j=1}^\ell j} =\max_{1\le \ell\le B+1}\frac{2\ell}{\ell+1}=\frac{2(B+1)}{B+2}$, which yields the claim.
\end{proof}
\end{corollary}

This improves upon the classical result of \citep{rueschendorf1982} stating that twice the arithmetic mean of p-values is itself a valid p-value. Here, we show that, in the permutation setting, the scaled quantity $\bar p_b\,2\,(B+1)/(B+2)$ is valid. See \Cref{sec: aggregation of data-asymmetric tests} for details.

\section{Adaptive testing details}
\label{sec: adaptive testing details}
\label{sec: adaptive MMD and HSIC tests}
This section details how SB aggregation yields finite-sample valid adaptive kernel tests. We instantiate SB aggregation to adaptive kernel two-sample MMD and independence HSIC testing\footnote{We simply refer the reader to \cite{schrab2025practical} for introductory details on MMD \citep{gretton2012kernel} and HSIC \citep{gretton2005kernel}.} by aggregating permutation p-values over a dyadic bandwidth grid. As discussed in \Cref{sec: maxT aggregation}, the existing MaxT-based adaptive MMD and HSIC procedures of \citep{albert2022adaptive} and \citep{schrab2023mmd} either rely on oracle critical values that are not directly implementable in practice or use Monte Carlo calibration that does not guarantee finite-sample control of the type~I error. Consequently, the corresponding adaptive separation rates are not established under rigorous finite-sample validity. In contrast, the SB-based adaptive MMD and HSIC tests proposed here are provably finite-sample valid and straightforward to implement. Moreover, since the SB minimum test uniformly dominates Bonferroni calibration in type II error, existing adaptive separation bounds for Bonferroni/MaxT procedures~\citep{albert2022adaptive,schrab2023mmd,schrab2022efficient} transfer directly, yielding the same adaptive rates.

Throughout this section, we assume that $g_1,\ldots,g_B$ are drawn independently and uniformly from the collection of all permutations when implementing SB aggregation. We can also consider other groups of transformations, such as pairwise permutations, which correspond to a wild bootstrap~\citep[e.g.,][]{schrab2023mmd,schrab2022efficient} and can also be used to establish adaptivity. For simplicity, however, we focus on uniform permutations.

\subsection{Adaptive MMD test}
\label{sec: adaptive MMD test}
We start with the adaptive two-sample testing problem considered by \citep{schrab2023mmd}. Let $(X_1,\ldots,X_m)\iid P$ and $(Y_1,\ldots,Y_n)\iid Q$ be independent samples taking values in $\mathbb{R}^d$, and consider testing
\begin{align*}
H_0 : P = Q
\qquad\text{versus}\qquad H_1 : P \neq Q.
\end{align*}
Let $f_P$ and $f_Q$ denote the densities of $P$ and $Q$ with respect to the Lebesgue measure.

\parheading{Separation radius.}
To quantify testing difficulty, we adopt the notion of a uniform separation radius. Let $\Delta$ be a level-$\alpha$ test based on $(X_1,\ldots,X_m,Y_1,\ldots,Y_n)$. For a function class $\mathcal C$ and constants $\rho,M>0$, define the alternative class
\begin{align*}
\mathcal{F}_\rho^M(\mathcal C)
\coloneqq
\Bigl\{
(f_P,f_Q):
\max(\|f_P\|_\infty,\|f_Q\|_\infty)\le M,\ 
f_P-f_Q\in\mathcal C,\ 
\|f_P-f_Q\|_2\ge\rho
\Bigr\},
\end{align*}
where $\|\cdot\|_2$ denotes the $L^2$ norm. For $\beta\in(0,1)$, the (uniform) separation radius of $\Delta$ over $\mathcal C$ is defined as
\begin{align*}
\rho(\Delta;\mathcal C,M,\alpha,\beta)
\coloneqq
\inf\Bigl\{
\rho>0:
\sup_{(f_P,f_Q)\in\mathcal{F}_\rho^M(\mathcal C)}
\mP_{(f_P,f_Q)}(\Delta=0)\le\beta
\Bigr\}.
\end{align*}
Throughout this subsection, we write $\rho(\Delta)$ for $\rho(\Delta;\mathcal C,M,\alpha,\beta)$. For simplicity, we regard $\alpha,\beta\in(0,1)$ as fixed constants but one can make the dependence on $\alpha$ and $\beta$ explicit in the separation radius, as done in \citep{schrab2025unified}.

\parheading{Sobolev smoothness class.}
Following \citep{schrab2023mmd}, we model smooth alternatives using a Sobolev ball. For $s>0$ and $R>0$, define
\begin{align*}
\mathcal S_d^{s}(R)
\coloneqq
\Bigl\{
f \in L^1(\mathbb R^d)\cap L^2(\mathbb R^d)
:
\int_{\mathbb R^d}
\|\xi\|_2^{2s}\,
|\hat f(\xi)|^2
\,\mathrm d\xi
\le
(2\pi)^d R^2
\Bigr\},
\end{align*}
where $\hat f$ denotes the Fourier transform of $f$. In what follows, we take $\mathcal C=\mathcal S_d^{s}(R)$.

\parheading{Quadratic-time MMD statistic.}
To test equality of distributions, we employ the quadratic-time MMD statistic~\citep{gretton2012kernel}. 
Let $G_1,\ldots,G_d$ be one-dimensional kernels satisfying $G_i\in L^1(\mathbb R)\cap L^2(\mathbb R)$ and $\int G_i=1$. For a bandwidth vector $\lambda=(\lambda_1,\ldots,\lambda_d)\in(0,\infty)^d$, define the product kernel
\begin{align*}
k_{\lambda}(x,y)
\coloneqq
\prod_{i=1}^d
\frac{1}{\lambda_i}
G_i\!\left(\frac{x_i-y_i}{\lambda_i}\right),
\end{align*}
which we assume to be characteristic on $\mathbb R^d$. The squared population MMD is
\begin{align*}
\mathrm{MMD}^2_{\lambda}(P,Q)
\coloneqq
\|\mu_P-\mu_Q\|_{\mathcal H_{k_{\lambda}}}^2,
\end{align*}
which admits an unbiased U-statistic estimator given by 
\begin{align*}
\widehat{\mathrm{MMD}}^{\,2}_{\lambda}(X_m,Y_n)
&\coloneqq
\frac{1}{m(m-1)}
\sum_{1\le i\ne i'\le m} k_{\lambda}(X_i,X_{i'}) \\
&\quad+
\frac{1}{n(n-1)}
\sum_{1\le j\ne j'\le n} k_{\lambda}(Y_j,Y_{j'})
-
\frac{2}{mn}
\sum_{i=1}^m\sum_{j=1}^n k_{\lambda}(X_i,Y_j).
\end{align*}

\parheading{Bandwidth collection.}
Adaptivity with respect to the unknown smoothness parameter $s$ is achieved by aggregating tests across a suitably chosen collection of bandwidths. Following \citep[][Corollary~10]{schrab2023mmd}, we consider the dyadic grid
\begin{align}
\label{eq:optimal_bandwidths_mmd}
\Lambda
\coloneqq
\Bigl\{
(2^{-k},\ldots,2^{-k})\in(0,\infty)^d
: k=1,\ldots,
\Bigl\lceil
\tfrac{2}{d}
\log_2\!\Bigl(\tfrac{m+n}{\log\log(m+n)}\Bigr)
\Bigr\rceil
\Bigr\},
\end{align}
and let $K\coloneqq|\Lambda|$. This dyadic grid balances the bias--variance tradeoff across smoothness levels and ensures that the optimal bandwidth is approximated up to logarithmic factors.

For each $k\in[K]$, with associated bandwidth $\lambda^{(k)}\in\Lambda$, define the statistic
\begin{align*}
T^{k}(X_m,Y_n)
\coloneqq
\widehat{\mathrm{MMD}}^{\,2}_{\lambda^{(k)}}(X_m,Y_n).
\end{align*}
\parheading{Permutation tests and SB minimum aggregation.}
To obtain finite-sample valid inference, we employ permutation tests. Let $(U_1,\ldots,U_{m+n})$ denote the pooled sample obtained by concatenating $(X_1,\ldots,X_m)$ and $(Y_1,\ldots,Y_n)$. For each permutation $g$ of $\{1,\ldots,m+n\}$, define the permuted samples $X_m^{g}=(U_{g(i)})_{i=1}^m$ and $Y_n^{g}=(U_{g(m+j)})_{j=1}^n$, and let $T^{k}(X_m^{g},Y_n^{g})$ be the corresponding statistic. Let $g_1,\ldots,g_B$ be i.i.d.\ uniform permutations and set $g_0$ to be the identity. For $b\in[B]_0$ and $k\in[K]$, write 
\begin{align*}
T_b^{k} \coloneqq T^{k}(X_m^{g_b},Y_n^{g_b}).
\end{align*}
We form permutation p-values by ranking each $T_b^{k}$ within its permutation distribution,
\begin{align*}
p\bigl(T_b^{k}\bigr)
\coloneqq
\frac{1}{B+1}\sum_{i=0}^B \mathds{1}\bigl(T_i^{k}\ge T_b^{k}\bigr),
\end{align*}
and aggregate them row-wise using the minimum merging function
\begin{align*}
f_b^{\min}
\coloneqq
\min_{k\in[K]} p\bigl(T_b^{k}\bigr),
\qquad b\in[B]_0.
\end{align*}
The SB minimum-aggregated p-value is then defined as
\begin{align*}
p_{\mathrm{SB},\mathrm{MMD}}
\coloneqq
\frac{1}{B+1}\sum_{b=0}^B
\mathds{1}\bigl(f_b^{\min}\le f_0^{\min}\bigr).
\end{align*}
Under $H_0$, the vectors $(T_b^{1},\ldots,T_b^{K})$, $b\in[B]_0$, are row-wise exchangeable. Since the permutation p-value map and the merging function are applied identically to each row, the aggregated values $(f_0^{\min},\ldots,f_B^{\min})$ are exchangeable, and hence $p_{\mathrm{SB},\mathrm{MMD}}$ is super-uniform (\Cref{prop: SB level}). The resulting SB minimum test rejects $H_0$ when $p_{\mathrm{SB},\mathrm{MMD}}\le\alpha$.

\parheading{Adaptive separation rate.}
We now analyze the testing power of the SB minimum test. By \Cref{thm:single-batch power}, the SB minimum test, denoted by $\Delta_{\mathrm{SB},\mathrm{MMD}}$, enjoys uniform dominance in type~II error over the Bonferroni-corrected MMD procedure constructed on the same bandwidth grid. For each $k\in[K]$ with $\lambda^{(k)}\in\Lambda$, let $\Delta_{\lambda^{(k)}}$ denote the level-($\alpha/K$) permutation test based on the statistic $\widehat{\mathrm{MMD}}^2_{\lambda^{(k)}}$, and define the Bonferroni-aggregated test as
\begin{align*}
\Delta_{\mathrm{Bonf},\mathrm{MMD}}
\coloneqq
\mathds{1}\biggl\{
\min_{k\in[K]} p_{\lambda^{(k)}} \le \alpha/K
\biggr\}.
\end{align*}
Under the conditions of \citep[][Theorem~6]{schrab2023mmd} and for a sufficiently large $B$, the permutation MMD test $\Delta_\lambda$ with a fixed bandwidth $\lambda$ satisfies the separation bound 
\begin{align*}
\rho(\Delta_\lambda)^2
\;\lesssim\;
\sum_{i=1}^d \lambda_i^{2s}
+
\frac{\log(1/\alpha)}{(m+n)\sqrt{\lambda_1\cdots\lambda_d}}.
\end{align*}
Applying this bound at level $\alpha/K$ and invoking the uniform power dominance of SB aggregation yields
\begin{align*}
\rho(\Delta_{\mathrm{SB},\mathrm{MMD}})^2
\;\le\;
\rho(\Delta_{\mathrm{Bonf},\mathrm{MMD}})^2
\;\lesssim\;
\min_{k\in[K]}
\Biggl\{
\sum_{i=1}^d (\lambda_i^{(k)})^{2s}
+
\frac{\log(K/\alpha)}
{(m+n)\sqrt{\lambda_1^{(k)}\cdots\lambda_d^{(k)}}}
\Biggr\}.
\end{align*}
Balancing the upper bounds in terms of $\lambda^{(k)}$, we choose $\lambda^{(k^\ast)}=(2^{-k^\ast},\ldots,2^{-k^\ast})$ with
\begin{align*}
k^\ast
=
\Bigl\lceil
\frac{2}{4s+d}
\log_2\biggl(\frac{m+n}{\log\log(m+n)}\biggr)
\Bigr\rceil.
\end{align*}
Substituting this choice into the above bound yields
\begin{align*}
\rho(\Delta_{\mathrm{SB},\mathrm{MMD}})^2
\;\lesssim\;
\biggl(\frac{\log\log(m+n)}{m+n}\biggr)^{\frac{4s}{4s+d}},
\end{align*}
which coincides with the adaptive separation rate established in \citep[][Corollary~10]{schrab2023mmd}.

\subsection{Adaptive HSIC test}
\label{sec: adaptive HSIC test}
We consider adaptive independence testing via the Hilbert–Schmidt Independence Criterion (HSIC), building on the framework of \citep{albert2022adaptive} and its permutation extension in \citep{schrab2022efficient}. Let $\{(X_i,Y_i)\}_{i=1}^n$ be i.i.d.\ observations in $\mathbb R^{d_x}\times\mathbb R^{d_y}$ with joint distribution $P_{XY}$ and marginals $P_X$ and $P_Y$. We test
\begin{align*}
H_0 : P_{XY} = P_X \otimes P_Y
\qquad\text{versus}\qquad H_1 : P_{XY} \neq P_X \otimes P_Y.
\end{align*}
\parheading{HSIC statistic.}
Let $k$ and $\ell$ be characteristic kernels on $\mathbb R^{d_x}$ and $\mathbb R^{d_y}$, respectively, and define the product kernel $\kappa\bigl((x,y),(x',y')\bigr) \coloneqq k(x,x')\,\ell(y,y')$. The population HSIC~\citep{gretton2005kernel}
is given by
\begin{align*}
\mathrm{HSIC}_{k,\ell}(P_{XY})
=
\mathrm{MMD}^2_{\kappa}\bigl(P_{XY},\,P_X\otimes P_Y\bigr),
\end{align*}
and vanishes if and only if $P_{XY}=P_X\otimes P_Y$. An unbiased estimator takes the form of a fourth-order U-statistic \citep{kim2022minimax,albert2022adaptive,schrab2022efficient}:
\begin{align*}
\widehat{\mathrm{HSIC}}_{k,\ell}(Z_n)
=
\frac{1}{n(n-1)(n-2)(n-3)}
\sum_{(i,j,r,s)\in i_n^4}
h^{\mathrm{HSIC}}_{k,\ell}(Z_i,Z_j,Z_r,Z_s),
\end{align*}
where $Z_i=(X_i,Y_i)$ and $i_n^4$ denotes the set of all 4-tuples of distinct indices from $\{1,\ldots,n\}$. For a kernel $s$ on a Euclidean space, define
\begin{align*}
h^{\mathrm{MMD}}_{s}(a_1,a_2;a_3,a_4)
\coloneqq
 s(a_1,a_2)
 - s(a_1,a_4)
 - s(a_2,a_3)
 + s(a_3,a_4),
\end{align*}
and set
\begin{align*}
h^{\mathrm{HSIC}}_{k,\ell}(z_1,z_2,z_3,z_4)
=
\frac{1}{4}\,
 h^{\mathrm{MMD}}_{k}(x_1,x_2;x_3,x_4)\,
 h^{\mathrm{MMD}}_{\ell}(y_1,y_2;y_3,y_4),
\end{align*}
with $z_i=(x_i,y_i)$.

\parheading{Bandwidth collection.}
We consider translation-invariant product kernels indexed by bandwidth parameters, analogously to the adaptive MMD construction. Let $G_1,\ldots,G_{d_x}$ and $L_1,\ldots,L_{d_y}$ be one-dimensional kernels satisfying $G_i,L_j\in L^1(\mathbb R)\cap L^2(\mathbb R)$ and $\int G_i=\int L_j=1$. For bandwidth vectors $\lambda\in(0,\infty)^{d_x}$ and $\mu\in(0,\infty)^{d_y}$, define
\begin{align*}
k_\lambda(x,x')
\coloneqq
\prod_{i=1}^{d_x}\frac{1}{\lambda_i}
G_i\!\left(\frac{x_i-x_i'}{\lambda_i}\right),
\qquad
\ell_\mu(y,y')
\coloneqq
\prod_{j=1}^{d_y}\frac{1}{\mu_j}
L_j\!\left(\frac{y_j-y_j'}{\mu_j}\right).
\end{align*}
We assume that $k_\lambda$ and $\ell_\mu$ are characteristic, ensuring that the associated HSIC detects independence.\footnote{Following \cite{gretton2015simpler}, it is sufficient to assume that $G_1,\dots,G_{d_x}$ and $L_1,\dots,L_{d_y}$ are characteristic.} 
The Gaussian kernels employed in \citep{albert2022adaptive} correspond to a special case of this family, and the analysis of \citep[][Theorem~3]{schrab2022efficient} applies to the general product-kernel setting.
Consider the dyadic bandwidth collection formed by jointly varying the $X$- and $Y$-bandwidths. Specifically, we define $\Lambda$ as the collection of concatenated bandwidth pairs $(\lambda,\mu) \in (0,\infty)^{d_x+d_y}$ of the form
\begin{align}
\label{eq:optimal_bandwidths_hsic}
\Lambda
\coloneqq
\Bigl\{
(2^{-k},\ldots,2^{-k})\in(0,\infty)^{d_x+d_y}
: k=1,\ldots,
\Bigl\lceil
\tfrac{2}{d_x+d_y}
\log_2\Bigl(\tfrac{n}{\log\log n}\Bigr)
\Bigr\rceil
\Bigr\},
\end{align}
which coincides with the collection used in \citep[][Theorem~3]{schrab2022efficient}. For each $k\in[K]$ with $K = |\Lambda|$, let $(\lambda^{(k)},\mu^{(k)})$ denote the $k$-th bandwidth pair in $\Lambda$, and define
\begin{align*}
T^{k}(Z_n)
\coloneqq
\widehat{\mathrm{HSIC}}_{k_{\lambda^{(k)}},\ell_{\mu^{(k)}}}(Z_n).
\end{align*}
\parheading{Permutation tests and SB aggregation.}
Under $H_0$, the joint distribution is invariant under permutations of the $Y$-coordinates relative to the $X$-coordinates. Let $g_1,\ldots,g_B$ be i.i.d.\ uniform permutations of $\{1,\ldots,n\}$, with $g_0$ the identity. Define
\begin{align*}
T_b^{k}
\coloneqq
\widehat{\mathrm{HSIC}}_{k_{\lambda^{(k)}},\ell_{\mu^{(k)}}}
\bigl(\bigl(X_i,Y_{g_b(i)}\bigr)_{i=1}^n\bigr),
\qquad b\in[B]_0.
\end{align*}
The associated permutation p-values are
\begin{align*}
p\bigl(T_b^{k}\bigr)
\coloneqq
\frac{1}{B+1}\sum_{i=0}^B
\mathds{1}\bigl(T_i^{k} \ge T_b^{k}\bigr).
\end{align*}
Aggregating row-wise using the minimum merging function yields
\begin{align*}
f_b^{\min}
\coloneqq
\min_{k\in[K]} p\bigl(T_b^{k}\bigr),
\qquad p_{\mathrm{SB},\mathrm{HSIC}}
\coloneqq
\frac{1}{B+1}\sum_{b=0}^B
\mathds{1}\bigl(f_b^{\min}\le f_0^{\min}\bigr).
\end{align*}
The resulting test $\Delta_{\mathrm{SB},\mathrm{HSIC}}$ rejects $H_0$ whenever $p_{\mathrm{SB},\mathrm{HSIC}}\le\alpha$. By row-wise exchangeability under $H_0$, $p_{\mathrm{SB},\mathrm{HSIC}}$ is super-uniform (\Cref{prop: SB level}), and hence $\Delta_{\mathrm{SB},\mathrm{HSIC}}$ is a valid level-$\alpha$ test.

\parheading{Adaptive separation rate.}
Recall that $\rho(\Delta_{\lambda,\mu})$ denotes the uniform separation radius of the level-$\alpha$ HSIC test with bandwidth $(\lambda,\mu)$, defined with respect to the alternative class $\mathcal F_\rho^M\!\bigl(\mathcal S^{s}_{d_x+d_y}(R)\bigr)$. Let $f_{XY}$ denote the joint density of $(X,Y)$ and let $f_X\otimes f_Y$ denote the product of the marginal densities. In this formulation, separation from the null corresponds to the $L^2$ distance between $f_{XY}$ and $f_X\otimes f_Y$, with Sobolev smoothness imposed on this density difference through $\mathcal S^{s}_{d_x+d_y}(R)$.

Under the conditions of \citep[][Theorem~3]{schrab2022efficient} and for a sufficiently large $B$, the HSIC permutation test with fixed bandwidth $(\lambda,\mu)$ satisfies
\begin{align*}
\rho(\Delta_{\lambda,\mu})^2
\;\lesssim\;
\sum_{i=1}^{d_x} \lambda_i^{2s} + \sum_{i=1}^{d_y} \mu_i^{2s} 
+
\frac{\log(1/\alpha)}
{n\sqrt{\lambda_1\cdots\lambda_{d_x}\mu_1\cdots\mu_{d_y}}}
+
\frac{\log(1/\alpha)^{3/2}}
{n^{3/2}\lambda_1\cdots\lambda_{d_x}\mu_1\cdots\mu_{d_y}}.
\end{align*}
By the uniform power dominance of SB aggregation (\Cref{thm:single-batch power}), the SB-aggregated HSIC test satisfies
\begin{align*}
\begin{aligned}
\rho(\Delta_{\mathrm{SB},\mathrm{HSIC}})^2
\;\lesssim\;
\min_{k\in[K]}
\Biggl\{
&\sum_{i=1}^{d_x} (\lambda_i^{(k)})^{2s}
+ \sum_{i=1}^{d_y} (\mu_i^{(k)})^{2s} \\
&\quad
+ \frac{\log(K/\alpha)}
{n\sqrt{\lambda_1^{(k)}\cdots\lambda_{d_x}^{(k)}
\,\mu_1^{(k)}\cdots\mu_{d_y}^{(k)}}} + \frac{\log(K/\alpha)^{3/2}}
{n^{3/2}\lambda_1^{(k)}\cdots\lambda_{d_x}^{(k)}
\,\mu_1^{(k)}\cdots\mu_{d_y}^{(k)}}
\Biggr\}.
\end{aligned}
\end{align*}
Since the dyadic grid satisfies $K\asymp\log n$, the additional Bonferroni factor contributes only a $\log\log n$ term. Balancing the bias and variance terms and assuming $4s\ge d_x+d_y$,\footnote{The condition $4s \geq d_x + d_y$ ensures that the term involving $n^{-3/2}$ does not dominate at the optimal bandwidth.} choosing $\lambda^{(k^\ast)}=\mu^{(k^\ast)}=(2^{-k^\ast},\ldots,2^{-k^\ast})$ with
\begin{align*}
k^\ast
=
\Bigl\lceil
\frac{2}{4s+d_x+d_y}
\log_2\Bigl(\tfrac{n}{\log\log n}\Bigr)
\Bigr\rceil
\end{align*}
yields
\begin{align*}
\rho(\Delta_{\mathrm{SB},\mathrm{HSIC}})^2
\;\lesssim\;
\biggl(\frac{\log\log n}{n}\biggr)^{\frac{4s}{4s+d_x+d_y}},
\end{align*}
which matches the adaptive minimax separation rate for independence testing up to the standard iterated logarithmic factor.

\section{Technical lemmas and auxiliary tools} \label{sec: useful facts on quantiles}
This section presents several useful facts on empirical quantiles and p-values that are used throughout the proofs. We begin by collecting standard facts relating empirical quantiles, rank-based p-values, and their equivalence under exchangeability. Here and throughout, $Z_{(1)} \le \cdots \le Z_{(n)}$ denote the order statistics, with the convention $Z_{(0)} = -\infty$ and $Z_{(n+1)} = +\infty$.

 \begin{lemma}[Lemmas S.14 and S.16, \cite{kim2025differentially}] \label{Lemma: permutation p-value}
Let $\alpha \in [0,1]$. Let $Z_1,\ldots,Z_n \in \mathbb{R}$. Then, for each $i \in [n]$,
\begin{align*}
	\frac{1}{n} \sum_{j=1}^n \mathds{1}(Z_j \geq Z_i) \leq \alpha \ \Longleftrightarrow \ Z_i > \mathrm{Quantile}_{1-\alpha}\bigl\{ Z_1,\ldots,Z_n\bigr\} .
\end{align*}
Furthermore when $Z_1,\ldots,Z_n$ are exchangeable, it holds that 
\begin{align*}
	\mP\biggl( \frac{1}{n} \sum_{j=1}^n \mathds{1}(Z_j \geq Z_i) \leq \alpha \biggr) = \mP\biggl( Z_i > \mathrm{Quantile}_{1-\alpha}\bigl\{ Z_1,\ldots,Z_n\bigr\} \biggr)\leq \alpha. 
\end{align*}
\end{lemma}

\begin{lemma}[Fact 2.10, \cite{angelopoulos2024theoretical}] \label{Lemma: quantile 4}
Let $\alpha \in [0,1]$. For any $Z_1,\ldots,Z_n \in \mathbb{R}$, 
\begin{align*}
	\mathrm{Quantile}_{\alpha}\bigl\{ Z_1,\ldots,Z_n \bigr\} = Z_{(\lceil \alpha n \rceil)}.
\end{align*}
\end{lemma}

\begin{lemma}[Lemma 3.4, \cite{angelopoulos2024theoretical}] \label{Lemma: replacement lemma}
Let $\alpha \in [0,1]$. For any $Z_1,\ldots,Z_n,Z_{n+1} \in \mathbb{R}$,
\begin{align*}
	Z_{n+1} \leq \mathrm{Quantile}_{\alpha}\bigl\{ Z_1,\ldots,Z_n,Z_{n+1} \bigr\}\, \Longleftrightarrow \, Z_{n+1} \leq \mathrm{Quantile}_{\alpha(1 + 1/n)} \bigl\{ Z_1,\ldots,Z_n \bigr\}.
\end{align*}
\end{lemma}

\begin{lemma} \label{Lemma: sup equivalence}
	For any $Z_1,\ldots,Z_n \in \mathbb{R}$ and $\alpha \in [0,1]$, the following holds:
		\begin{align*}
		\sup\biggl\{t : \frac{1}{n} \sum_{i=1}^n \mathds{1}(Z_i \leq t) \leq \alpha  \biggr\} = \sup\biggl\{t : \frac{1}{n} \sum_{i=1}^n \mathds{1}(Z_i < t) \leq \alpha  \biggr\}.
	\end{align*}
	Moreover, both sides are equal to $Z_{(\lfloor n\alpha \rfloor + 1)}$ where $Z_{(1)} \leq Z_{(2)} \leq \cdots \leq Z_{(n)}$ are the order statistics of $Z_1,\ldots,Z_n$.
 	\begin{proof} See \Cref{sec: proof of sup equivalence}.
	\end{proof}
\end{lemma}

The following lemma establishes a new general super-uniformity property for weighted rank functionals on arbitrary measure spaces.

\begin{lemma} \label{Lemma: generalized superuniform lemma}
Let $(\Omega,\mathcal F,\mu)$ be a measure space. Let $t:\Omega\to \overline{\mathbb R}\coloneqq \mathbb R\cup\{-\infty,+\infty\}$ be measurable and $w:\Omega\to[0,\infty)$ be measurable with $\int_\Omega w\,d\mu<\infty$. Then for every $\alpha\ge 0$,
\begin{align*}
\int_\Omega w(\omega)\,
\mathds 1\!\left\{
\int_\Omega w(\omega')\,
\mathds 1\{t(\omega')\ge t(\omega)\}\,d\mu(\omega')
\le \alpha
\right\}
\,d\mu(\omega)
\le\alpha.
\end{align*}
\end{lemma}

\begin{proof} 
	See \Cref{sec: proof of generalized superuniform lemma}.
\end{proof}

The above lemma immediately implies the following corollary, which corresponds to \citep[][Lemma A1]{harrison2012conservative}.

\begin{corollary} \label{cor: weighted superuniform lemma}
Let $t_0,\ldots,t_n\in[-\infty,\infty]$, $w_0,\ldots,w_n\in[0,\infty)$, and $\alpha\ge 0$. Then
\begin{align*}
\sum_{k=0}^n w_k\,
\mathds 1\!\left\{
\sum_{i=0}^n w_i\,\mathds 1\{t_i\ge t_k\}
\le \alpha
\right\}
\le \alpha.
\end{align*}
\end{corollary}

\begin{proof}
Apply \Cref{Lemma: generalized superuniform lemma} with $\Omega=\{0,1,\dots,n\}$, $\mu$ the counting measure, $t(k)=t_k$, and $w(k)=w_k$.
\end{proof}

\section{Proofs of the results in the main text}
\label{sec: proofs of main results}

This section collects the proofs of the results stated in the main text.

\subsection{Proof of \Cref{prop: SB level}}\label{sec: proof of SB level}
Under the group-invariance hypothesis, the row vectors
\(
(T_0^1,\ldots,T_0^K),\ldots,(T_B^1,\ldots,T_B^K)
\)
are exchangeable. Since the p-value map
\( (T_0^k,\ldots,T_B^k)\mapsto (p(T_0^k),\ldots,p(T_B^k)) \)
and the merging map
\(
(p(T_b^1),\ldots,p(T_b^K))\mapsto f_b
\)
are applied identically for each index \(b\in[B]_0\), the aggregated values
\((f_0,\ldots,f_B)\) are exchangeable as well. Hence the rank p-value
\begin{align*}
p_{\mathrm{SB}}=\frac{1}{B+1}\sum_{b=0}^B \mathds{1}(f_b\le f_0)
\end{align*}
is super-uniform by \Cref{Lemma: permutation p-value}, which yields
\(\mP(p_{\mathrm{SB}}\le\alpha)\le\alpha\). If, moreover, \(f_0,\ldots,f_B\) are distinct a.s., then \(p_{\mathrm{SB}}\) is uniform on the grid \(\{1/(B+1),\ldots,1\}\), and therefore
\(
\mP(p_{\mathrm{SB}}\le\alpha)=\lfloor (B+1)\alpha\rfloor/(B+1).
\)

\subsection{Proof of \Cref{prop: single-batch equivalence}}
\label{proof: Proposition: single-batch equivalence}
We start from the definition of the SB threshold,
\(\uSB=-\mathrm{Quantile}_{1-\alpha}\{-f_0,\ldots,-f_B\}\), and expand the empirical quantile using our convention
\(\mathrm{Quantile}_{q}\{x_0,\ldots,x_B\}=\inf\{t:\frac{1}{B+1}\sum_{b=0}^B\mathds{1}(x_b\le t)\ge q\}\). Then
\begin{align*}
\uSB
&= -\inf\Bigl\{t:\frac{1}{B+1}\sum_{b=0}^B\mathds{1}(-f_b\le t)\ge 1-\alpha\Bigr\}\\
&= -\inf\Bigl\{t:\frac{1}{B+1}\sum_{b=0}^B\mathds{1}(-f_b> t)\le \alpha\Bigr\}\\
&= \sup\Bigl\{u:\frac{1}{B+1}\sum_{b=0}^B\mathds{1}(-f_b>-u)\le \alpha\Bigr\}\\
&= \sup\Bigl\{u:\frac{1}{B+1}\sum_{b=0}^B\mathds{1}(f_b<u)\le \alpha\Bigr\}\\
&= \sup\Bigl\{u:\frac{1}{B+1}\sum_{b=0}^B\mathds{1}(f_b\le u)\le \alpha\Bigr\},
\end{align*}
where the second line uses the identity \(\mathds{1}(x\le t)=1-\mathds{1}(x>t)\), the third line is the change of variables \(u=-t\), and the last line follows from \Cref{Lemma: sup equivalence} (which shows that replacing $\le$ by $<$ inside the supremum does not change its value). This completes the proof.

\subsection{Proof of \Cref{lemma:super-uniform-property}}\label{sec: proof of super-uniform-property}
We begin by noting that the second inequality is an immediate consequence of a general weighted result due to \citep[][Lemma A1]{harrison2012conservative} recalled in \Cref{cor: weighted superuniform lemma}. Specifically, for arbitrary $T_0,\ldots,T_B\in[-\infty,\infty]$ and nonnegative weights $w_0,\ldots,w_B$, that result implies
\begin{align*}
\sum_{i=0}^B w_i \mathds{1}\biggl(\sum_{j=0}^B w_j \mathds{1}(T_j \ge T_i) \le \alpha\biggr)
\le \alpha.
\end{align*}
Specializing to uniform weights $w_0=\cdots=w_B=1/(B+1)$ yields the second claim of \Cref{lemma:super-uniform-property}. We therefore concentrate on establishing the first inequality, for which we provide a direct and sharper argument in the uniform-weights setting.

Let $R_j = \sum_{i=0}^B \mathds{1}(T_i \geq T_j)$ so that $p_j = R_j/(B+1)$. We claim that $\sum_{j=0}^B \mathds{1}(R_j \leq k) \leq k$ for all $k \in [B]_0$. To see this, write $T_{(0)} \leq T_{(1)} \leq T_{(2)} \leq \ldots \leq T_{(B)}$ for the order statistics and observe
\begin{align*}
	\sum_{j=0}^B \mathds{1}(R_j \leq k)  = \sum_{j=0}^B \mathds{1}\biggl( \sum_{i=0}^B \mathds{1}(T_i \geq T_j) \leq k  \biggr) = \sum_{j=0}^B \mathds{1}\biggl( \sum_{i=0}^B \mathds{1}(T_i \geq T_{(j)}) \leq k  \biggr).
\end{align*}
Now note that for each $j \in [B]_0$,
\begin{align*}
		\sum_{i=0}^B \mathds{1}(T_i \geq T_{(j)}) \geq B - j +1.
\end{align*}
Therefore
\begin{align*}
	\sum_{j=0}^B \mathds{1}(R_j \leq k) \leq \sum_{j=0}^B \mathds{1}\bigl(  B -j +1 \leq k \bigr) = \sum_{j=0}^B \mathds{1}\bigl(  B - k +1 \leq j \bigr) = \sum_{j = B-k+1}^B 1 = k,
\end{align*}
as desired. To complete the proof, observe that 
\begin{align*}
	\sum_{j=0}^B \mathds{1}(p_j \leq \alpha) = \sum_{j=0}^B \mathds{1}(R_j \leq (B+1)\alpha)  = \sum_{j=0}^B \mathds{1}(R_j \leq \lfloor (B+1)\alpha \rfloor) \leq  \lfloor (B+1)\alpha \rfloor.
\end{align*}
Consequently, we prove the first claim that
\begin{align*}
	\frac{1}{B+1} \sum_{j=0}^B \mathds{1}(p_j \leq \alpha) \leq \frac{\lfloor (B+1)\alpha \rfloor}{B+1} \leq \alpha. 
\end{align*}
When $T_0,\ldots,T_B$ are distinct, we have 
\begin{align*}
	\sum_{j=0}^B \mathds{1}(R_j \leq \lfloor (B+1)\alpha \rfloor) = \lfloor (B+1)\alpha \rfloor,
\end{align*}
which proves the second claim. This completes the proof of \Cref{lemma:super-uniform-property}.

\subsection{Proof of \Cref{thm:single-batch power}}\label{sec: proof of single-batch power}
Fix the observed data $\mathbf{X}$ and transformations $g_0,\dots,g_B$ with $g_0$ being the identity.
For each $b\in[B]_0$, define
\begin{align*}
f_b = f\bigl(p(T_b^1),\ldots,p(T_b^K)\bigr).
\end{align*}
By \Cref{lemma:super-uniform-property}, the p-values $p(T_b^1),\ldots,p(T_b^K)$ are super-uniform conditional on $\mathbf{X}$, with randomness arising only through $b \sim \mathrm{Unif}([B]_0)$. Hence, by the defining property of $c_{\alpha,K}$,
\begin{align}
\label{Eq: validity bound for c}
\frac{1}{B+1}\sum_{b=0}^B \mathds{1}(f_b \le c_{\alpha,K}) \le \alpha.
\end{align}
Let $F(u) \coloneqq \frac{1}{B+1}\sum_{b=0}^B \mathds{1}(f_b \le u)$ denote the empirical distribution function of $\{f_0,\ldots,f_B\}$. Since $F$ is right-continuous and non-decreasing from 0 to 1, the alternative characterization of $\uSB$ in \Cref{prop: single-batch equivalence} implies $F(\uSB)>\alpha$, while \eqref{Eq: validity bound for c} gives $F(c_{\alpha,K})\le\alpha$. Therefore $c_{\alpha,K} < \uSB$ almost surely, which implies the claimed bound on the type~II error.

\subsection{Proof of \Cref{corollary: single-batch power O and M}}\label{sec: proof of corollary single-batch power}

By \Cref{thm:single-batch power}, the SB aggregation test calibrated at level $\alpha$ uniformly dominates any deterministic test obtained by worst-case calibration of the same merging function. For the O- and M-families, the prior work \citep[e.g.,][]{vovk2020combining} shows that the corresponding merging functions produce super-uniform aggregated p-values under arbitrary dependence, so that the worst-case calibration constant satisfies $c_{\alpha,K}=\alpha$. Combining these facts yields the stated uniform power dominance.

\subsection{Proof of \Cref{prop: single-batch perfect rank alignment}}\label{sec: proof of single-batch perfect rank alignment}

Recall that the merging function $f$ satisfies the diagonal monotonicity condition
\eqref{eq:diagonal_monotone}, namely
\begin{align*}
f(x,\ldots,x) \le f(y,\ldots,y) \quad \Longleftrightarrow \quad x \le y .
\end{align*}
Under the rank alignment condition \eqref{eq:rank_alignment}, the vectors $(T_0^k,\ldots,T_B^k)$ induce the same ordering for all $k\in[K]$. Consequently, the permutation p-values satisfy
\begin{align*}
p(T_b^1)=\cdots=p(T_b^K) \eqqcolon p_b,
\qquad b\in[B]_0 .
\end{align*}
Therefore, for each $b\in[B]_0$,
\begin{align*}
f\bigl(p(T_b^1),\ldots,p(T_b^K)\bigr)
=
f(p_b,\ldots,p_b),
\end{align*}
and by diagonal monotonicity, for any $b,b'\in[B]_0$,
\begin{align*}
f(p_b,\ldots,p_b) \le f(p_{b'},\ldots,p_{b'})
\quad \Longleftrightarrow \quad p_b \le p_{b'}.
\end{align*}
Using the p-value representation of the SB test, we obtain 
\begin{align*}
p_{\mathrm{SB}}
&=
\frac{1}{B+1}\sum_{b=0}^B
\mathds{1}\Bigl(
f\bigl(p(T_b^1),\ldots,p(T_b^K)\bigr)
\le f\bigl(p(T_0^1),\ldots,p(T_0^K)\bigr)
\Bigr) \\
&=
\frac{1}{B+1}\sum_{b=0}^B \mathds{1}(p_b \le p_0).
\end{align*}
Finally, recall that by definition of the permutation p-values,
\begin{align*}
p_b = \frac{1}{B+1}\sum_{i=0}^B \mathds{1}(T_i^k \ge T_b^k)
\end{align*}
for any fixed $k\in[K]$. Since the ordering of $T_i^k$ does not depend on $k$ under \eqref{eq:rank_alignment}, the right-hand side is the same for all $k$. Hence $p_b \le p_0$ if and only if $T_b^k \ge T_0^k$ for any $k\in[K]$. In particular, fixing an arbitrary representative coordinate $k^\star\in[K]$ and writing $T_b \coloneqq T_b^{k^\star}$ for all $b\in[B]_0$, we obtain the equivalent condition $T_b \ge T_0$.

Combining the above identities yields
\begin{align*}
p_{\mathrm{SB}}
=
\frac{1}{B+1}\sum_{b=0}^B \mathds{1}(T_b \ge T_0),
\end{align*}
which coincides with the usual permutation test based on a single statistic. This completes the proof.

\subsection{Proof of \Cref{prop:asymptotic-adaptation}}
\label{proof:asymptotic-adaptation}

The argument follows the standard analysis of empirical quantiles for a permutation distribution (e.g.,
\citep[][Theorem~17.2.3]{lehmann2022}). The only difference is that, instead of enumerating all transformations in a finite group, we work with statistics computed from i.i.d.\ random transformations. We therefore verify directly that the Monte Carlo empirical distribution converges to the limiting null distribution, and then apply a standard quantile-continuity argument.

\begin{proof}
For each $(B,n)$, define the empirical distribution functions
\begin{align*}
\hat F_{B,n}(t)\coloneqq \frac{1}{B+1}\sum_{b=0}^B \mathds{1}(f_{b,n}\le t),
\qquad
\tilde F_{B,n}(t)\coloneqq \frac{1}{B}\sum_{b=1}^B \mathds{1}(f_{b,n}\le t),
\quad t\in\mathbb{R}.
\end{align*}
Since
\begin{align*}
\hat F_{B,n}(t)
=
\frac{B}{B+1}\tilde F_{B,n}(t)
+ 
\frac{1}{B+1}\mathds{1}(f_{0,n}\le t),
\end{align*}
we have the uniform bound
\begin{align*}
\sup_{t\in\mathbb{R}}
\bigl|
\hat F_{B,n}(t)-\tilde F_{B,n}(t)
\bigr|
\le \frac{1}{B+1}.
\end{align*}
Thus it suffices to establish convergence of $\tilde F_{B,n}$, and then transfer it to $\hat F_{B,n}$ and to the corresponding empirical quantiles.

\parheading{Step 1: Pointwise convergence of the empirical CDF.}
Fix a distribution $P$ and a continuity point $t$ of $F_{\infty,P}(t)\coloneqq \mP_P(V_{\infty,P}\le t)$. Consider
\begin{align*}
\tilde F_{B,n}(t)=\frac{1}{B}\sum_{b=1}^B \mathds{1}(f_{b,n}\le t).
\end{align*}
Since $(f_{1,n},\ldots,f_{B,n})$ are exchangeable (under both the null and the alternative),
\begin{align*}
\mE_P\bigl[\tilde F_{B,n}(t)\bigr]=\mP_P(f_{1,n}\le t)\eqqcolon p_{n,P}(t).
\end{align*}
Moreover,
\begin{align*}
\mE_P\bigl[\tilde F_{B,n}(t)^2\bigr]
&=
\frac{1}{B^2}\sum_{b=1}^B \mE_P\bigl[\mathds{1}(f_{b,n}\le t)\bigr]
+ 
\frac{1}{B^2}\sum_{\substack{b,b'=1\\ b\neq b'}}^B
\mE_P\bigl[\mathds{1}(f_{b,n}\le t)\mathds{1}(f_{b',n}\le t)\bigr] \\
&=
\frac{1}{B}p_{n,P}(t)+\frac{B-1}{B}q_{n,P}(t),
\end{align*}
where $q_{n,P}(t)\coloneqq \mP_P(f_{1,n}\le t,\ f_{2,n}\le t)$. Therefore,
\begin{align*}
\Var_P\bigl(\tilde F_{B,n}(t)\bigr)
=
\frac{1}{B}\bigl(p_{n,P}(t)-p_{n,P}(t)^2\bigr)
+\frac{B-1}{B}\bigl(q_{n,P}(t)-p_{n,P}(t)^2\bigr).
\end{align*}
Recall that $t$ is a continuity point of $F_{\infty,P}$. By the assumed joint convergence $(f_{1,n},f_{2,n})\overset{d}{\longrightarrow}(V_{\infty,P},V'_{\infty,P})$ with $V_{\infty,P},V'_{\infty,P}$ i.i.d., the Portmanteau theorem applied to the continuity sets $(-\infty,t]$ and $(-\infty,t]^2$ yields
\begin{align*}
 p_{n,P}(t)\to F_{\infty,P}(t),
\qquad
 q_{n,P}(t)\to F_{\infty,P}(t)^2,
\end{align*}
so $q_{n,P}(t)-p_{n,P}(t)^2\to 0$. Hence $\Var_P(\tilde F_{B,n}(t))\to 0$ as $B,n\to\infty$. Chebyshev's inequality gives
\begin{align*}
\tilde F_{B,n}(t)-p_{n,P}(t)\ \overset{p}{\longrightarrow}\ 0,
\qquad\text{and thus}\qquad
\tilde F_{B,n}(t)\ \overset{p}{\longrightarrow}\ F_{\infty,P}(t).
\end{align*}
\parheading{Step 2: Pointwise quantile convergence.}
Recall that the SB threshold is
\begin{align*}
\uSB
=
\sup\Bigl\{u\in\mathbb R:
\hat F_{B,n}(u)\le \alpha
\Bigr\},
\qquad Q_{\alpha,P}^\star
\coloneqq
\inf\{u\in\mathbb{R}:F_{\infty,P}(u)\ge \alpha\}.
\end{align*}
Fix $P$. By the assumption, for every $\varepsilon>0$,
\begin{align*}
F_{\infty,P}(Q_{\alpha,P}^\star-\varepsilon)<\alpha
\quad\text{and}\quad F_{\infty,P}(Q_{\alpha,P}^\star+\varepsilon)>\alpha.
\end{align*}
Using Step~1 (pointwise) and the $O(B^{-1})$ bound between $\hat F_{B,n}$ and $\tilde F_{B,n}$,
\begin{align*}
\hat F_{B,n}(Q_{\alpha,P}^\star-\varepsilon)\ \overset{p}{\longrightarrow}\ F_{\infty,P}(Q_{\alpha,P}^\star-\varepsilon),
\qquad
\hat F_{B,n}(Q_{\alpha,P}^\star+\varepsilon)\ \overset{p}{\longrightarrow}\ F_{\infty,P}(Q_{\alpha,P}^\star+\varepsilon),
\end{align*}
where $Q_{\alpha,P}^\star-\varepsilon$ and $Q_{\alpha,P}^\star+\varepsilon$ are assumed to be continuity points of $F_{\infty,P}$ without loss of generality. Consequently,
\begin{align*}
\mP_P\!\left(\hat F_{B,n}(Q_{\alpha,P}^\star-\varepsilon)\le \alpha\right)\to 1,
\qquad
\mP_P\!\left(\hat F_{B,n}(Q_{\alpha,P}^\star+\varepsilon)> \alpha\right)\to 1.
\end{align*}
On the intersection of these events we have $Q_{\alpha,P}^\star-\varepsilon\le \uSB\le Q_{\alpha,P}^\star+\varepsilon$, hence $\uSB\overset{p}{\longrightarrow}Q_{\alpha,P}^\star$.
\end{proof}

\subsection{Proof of \Cref{prop: minimum single-batch calibration statistic view}}\label{sec: proof of minimum single-batch calibration}
For each $k\in[K]$ and $b\in[B]_0$, define
\begin{align*}
p_b^k \coloneqq p(T_b^k)
=
\frac{1}{B+1}\sum_{i=0}^{B}\mathds{1}(T_i^k\ge T_b^k),
\qquad f_b \coloneqq f_b^{\min}=\min_{k\in[K]}p_b^k .
\end{align*}
Recall from Proposition~\ref{prop: single-batch equivalence} that the SB critical value for the minimum merge satisfies
\begin{equation}\label{eq:usbmin_sup_def}
\uSBmin
=
\sup\Bigl\{u\in\mathbb R:\ 
\frac{1}{B+1}\sum_{b=0}^{B}\mathds{1}(f_b\le u)\le \alpha
\Bigr\}.
\end{equation}
All statements below are deterministic conditional on the realized matrix $\{T_b^k\}_{b,k}$, hence they hold almost surely.

\parheading{Step 1: Bonferroni rejection implies SB rejection.}
For the minimum merge $f^{\min}(u_1,\ldots,u_K)=\min_{k\in[K]}u_k$, the Bonferroni bound implies that the deterministic worst-case calibration constant satisfies $c_{\alpha,K}=\alpha/K$. Hence, by \Cref{thm:single-batch power}, the SB critical value obeys $\uSBmin>c_{\alpha,K}=\alpha/K$ almost surely. Therefore, $\min_k p(T_0^k)\le \alpha/K$ implies $\min_k p(T_0^k)<\uSBmin$, proving the claim.

\parheading{Step 2: SB rejection implies unadjusted minimum rejection.}
We prove the contrapositive. Suppose that $f_0=\min_k p(T_0^k)>\alpha$, and let $k_\star\in[K]$ be an index attaining the minimum, i.e., $p(T_0^{k_\star})=f_0$. Consider the set of transformations whose $k_\star$-statistic is at least as large as the observed one,
\begin{align*}
A\coloneqq\{b\in[B]_0: T_b^{k_\star}\ge T_0^{k_\star}\}.
\end{align*}
By definition of the permutation p-value, $|A|=(B+1)\,p(T_0^{k_\star})=(B+1)f_0$. Moreover, for any $b\in A$, the statistic $T_b^{k_\star}$ is more extreme than $T_0^{k_\star}$, so its permutation p-value in coordinate $k_\star$ cannot be larger:
\begin{align*}
p(T_b^{k_\star})\le p(T_0^{k_\star})=f_0.
\end{align*}
Since $f_b=\min_k p(T_b^k)\le p(T_b^{k_\star})$, we obtain $f_b\le f_0$ for every $b\in A$. Consequently,
\begin{align*}
\frac{1}{B+1}\sum_{b=0}^B \mathds{1}(f_b\le f_0)
\ \ge\ \frac{|A|}{B+1}
\ =\ f_0
\ >\ \alpha.
\end{align*}
Thus $u=f_0$ is not feasible in \eqref{eq:usbmin_sup_def}, which implies $\uSBmin\le f_0$. In particular, under $f_0>\alpha$ the event $f_0<\uSBmin$ is impossible, proving
\begin{align*}
\mathds{1}\!\left(\min_{k}p(T_0^k)< \uSBmin\right)
\le
\mathds{1}\!\left(\min_{k}p(T_0^k)\le \alpha\right).
\end{align*}
\parheading{Step 3: Lower bound on $\uSBmin$.}
By \Cref{thm:single-batch power}, we have
\begin{align*}
\uSBmin > \frac{\alpha}{K}.
\end{align*}
Since each permutation p-value lies on the grid
\begin{align*}
\left\{\frac{1}{B+1}, \frac{2}{B+1}, \ldots, 1\right\},
\end{align*}
the merged values $f_b^{\mathrm{min}}$ also lie on this grid. By the order-statistic representation of the SB threshold, $\uSBmin$ belongs to the same grid. Therefore, the smallest grid value strictly larger than $\alpha/K$ is
\begin{align*}
\frac{\lfloor (B+1)\alpha/K \rfloor + 1}{B+1},
\end{align*}
which implies
\begin{align*}
\uSBmin 
\ge 
\frac{\lfloor (B+1)\alpha/K \rfloor + 1}{B+1}.
\end{align*}

\begin{example}
\label{ex:SBmin_equals_Bonf} \normalfont
We construct a finite-sample example in the standard randomization setup (statistics computed from transformed data) such that
\begin{equation}
\label{eq:SB_equals_Bonf_indicator}
\mathds{1}\!\left(
\min_{k\in[K]} p(T_0^k) \le \frac{\alpha}{K}
\right)
=
\mathds{1}\!\left(
\min_{k\in[K]} p(T_0^k) < \hat u^{\mathrm{SB}}_{\alpha,\min}
\right)
\quad\text{almost surely}.
\end{equation}

\parheading{Step 1: Parameters.}
Fix integers $B\ge 1$, $K\ge 2$, and $\alpha\in(0,1)$, and define
\begin{align*}
m \coloneqq \left\lfloor \frac{(B+1)\alpha}{K} \right\rfloor .
\end{align*}
Assume $(B+1)\alpha/K\notin\mathbb{Z}$ so that
\begin{align*}
\frac{m}{B+1}
\;<\;
\frac{\alpha}{K}
\;<\;
\frac{m+1}{B+1}.
\end{align*}
\parheading{Step 2: Randomization setup.}
Let the data vector $\mathbf X$ have size $B+1$ and define
\begin{align*}
\mathbf X = (U_0,\ldots,U_B),
\qquad U_0,\ldots,U_B \stackrel{\mathrm{i.i.d.}}{\sim} \mathrm{Unif}(0,1).
\end{align*}
Thus the sample size in this construction is $n=B+1$. Let $\pi:\{0,\ldots,B\}\to\{0,\ldots,B\}$ be the cyclic permutation $\pi(j)\coloneqq j+1 \pmod{B+1}$. For each $b\in\{0,\ldots,B\}$, define $g_b$ as the coordinate permutation induced by the $b$-fold composition $\pi^b$ (i.e., $\pi$ applied $b$ times), namely
\begin{align*}
(g_b(\mathbf X))_j \coloneqq \mathbf X_{\pi^b(j)} \qquad \text{for all } j\in\{0,\ldots,B\}.
\end{align*}
In particular, $(g_b(\mathbf X))_0 = U_b$, and therefore
\begin{align*}
\{(g_b(\mathbf X))_0 : b=0,\ldots,B\} = \{U_0,\ldots,U_B\}.
\end{align*}
\parheading{Step 3: Coordinate statistics.}
Let $R_1,\ldots,R_K \subset \{1,\ldots,B+1\}$ be disjoint sets of size $m$. Define the rank of the first coordinate by
\begin{align*}
r(x) \coloneqq 1 + \sum_{j=0}^B \mathds{1}\{x_j < x_0\}.
\end{align*}
For each $k\in[K]$, define
\begin{align*}
T^k(x) \coloneqq \mathds{1}\{ r(x) \in R_k \},
\qquad T_b^k \coloneqq T^k(g_b (\mathbf X)).
\end{align*}
Since the $U_b$ are continuous, all ranks are distinct almost surely.

\parheading{Step 4: Permutation p-values.}
For each $k$ and $b$, define the permutation p-value 
\begin{align*}
p(T_b^k)
=
\frac{1}{B+1}\sum_{i=0}^B \mathds{1}(T_i^k \ge T_b^k). 
\end{align*}
Because exactly $m$ ranks belong to $R_k$, we have almost surely
\begin{align*}
p(T_b^k)
=
\begin{cases}
\displaystyle \frac{m}{B+1}, & r(g_b(\mathbf X))\in R_k,\\[6pt]
1, & \text{otherwise}.
\end{cases}
\end{align*}
\parheading{Step 5: Minimum merged values.}
Define
\begin{align*}
f_b^{\min} \coloneqq \min_{k\in[K]} p(T_b^k).
\end{align*}
Then almost surely
\begin{align*}
f_b^{\min}
=
\begin{cases}
\displaystyle \frac{m}{B+1}, & r(g_b(\mathbf X))\in \bigcup_{k=1}^K R_k,\\[6pt]
1, & \text{otherwise}.
\end{cases}
\end{align*}
Since the sets $R_k$ are disjoint,
\begin{align*}
\frac{1}{B+1}\sum_{b=0}^B \mathds{1}\!\left(f_b^{\min} \le \frac{m}{B+1}\right)
=
\frac{Km}{B+1}
\le \alpha,
\end{align*}
whereas
\begin{align*}
\frac{1}{B+1}\sum_{b=0}^B \mathds{1}\!\left(f_b^{\min} \le 1\right)
= 1 > \alpha .
\end{align*}
\parheading{Step 6: SB threshold and equivalence.}
By definition,
\begin{align*}
\hat u^{\mathrm{SB}}_{\alpha,\min}
=
\sup\Bigl\{
u:\frac{1}{B+1}\sum_{b=0}^B \mathds{1}(f_b^{\min}\le u)\le\alpha
\Bigr\}
= 1 .
\end{align*}
Therefore,
\begin{align*}
\min_k p(T_0^k) < \hat u^{\mathrm{SB}}_{\alpha,\min}
\;\Longleftrightarrow\;
\min_k p(T_0^k) = \frac{m}{B+1}.
\end{align*}
Using $\frac{m}{B+1}<\alpha/K<1$, we conclude that
\begin{align*}
\min_{k} p(T_0^k) \le \frac{\alpha}{K}
\quad\Longleftrightarrow\quad
\min_{k} p(T_0^k) < \hat u^{\mathrm{SB}}_{\alpha,\min},
\end{align*}
which proves \eqref{eq:SB_equals_Bonf_indicator}.
\end{example}

\subsection{Proof of \Cref{prop: sequential minimum aggregation}}\label{sec: proof of sequential minimum aggregation}

This is the special case of \Cref{prop: general sequential aggregation} obtained by taking $J=K$, $I_j=[j]$, and $h_j(x_1,\ldots,x_j)=\min_{1\le k\le j}x_k$.

\subsection{Proof of \Cref{prop: TB level}}\label{sec: proof of TB level}

Condition on the reference batch $\{(\tilde T_i^1,\ldots,\tilde T_i^K)\}_{i=1}^B$. Under the group-invariance hypothesis, the testing-batch row vectors
\begin{align*}
(T_0^1,\ldots,T_0^K),\ldots,(T_B^1,\ldots,T_B^K)
\end{align*}
are exchangeable given the reference batch. Since the holdout p-value map $p_{\mathrm{HO}}(\cdot)$ and the merging function $f$ are applied row-wise and identically across $b\in[B]_0$, the merged values $(\tilde f_0,\ldots,\tilde f_B)$ are conditionally exchangeable as well. Therefore the rank p-value
\begin{align*}
p_{\mathrm{TB}}=\frac{1}{B+1}\sum_{b=0}^B \mathds{1}(\tilde f_b \le \tilde f_0)
\end{align*}
is conditionally super-uniform by \Cref{Lemma: permutation p-value}, implying $\mP(p_{\mathrm{TB}}\le\alpha)\le\alpha$ after removing the conditioning. If $\tilde f_0,\ldots,\tilde f_B$ are distinct a.s., then $p_{\mathrm{TB}}$ is uniform on the grid $\{1/(B+1),\ldots,1\}$ conditional on the reference batch, and hence
$\mP(p_{\mathrm{TB}}\le\alpha)=\lfloor(B+1)\alpha\rfloor/(B+1)$.

\subsection{Proof of \Cref{prop:cp_intersection}}\label{sec: proof of cp intersection}

Fix the split with $|I_{\mathrm{ref}}|=n_1$ and write $u\coloneqq u_{\alpha}^{\mathrm{TB}}$. Under the absolute residual score
\(
s_k(x,y)=|y-\widehat{\mu}_k(x)|,
\)
the reference residuals are
\(
R_j^k \coloneqq |Y_j-\widehat{\mu}_k(X_j)|
\)
for $j\in I_{\mathrm{ref}}$, and for any candidate label $y\in\mathbb R$ the reference-based p-value \eqref{eq:cp_TB_column_pvalues_test} becomes
\begin{align*}
P_{n+1}^{k,\mathrm{ref}}(y)
=
\frac{1+\sum_{j\in I_{\mathrm{ref}}}\mathds{1}\!\left\{R_j^k
\ge |y-\widehat{\mu}_k(X_{n+1})|\right\}}{n_1+1}.
\end{align*}
With the minimum merge $f(u_1,\ldots,u_K)=\min_{k\in[K]}u_k$, we have
\begin{align*}
M_{n+1}(y)=\min_{k\in[K]} P_{n+1}^{k,\mathrm{ref}}(y).
\end{align*}
By definition of the TB set in \eqref{eq:cp_TB_threshold},
\begin{align*}
y\in C_{\mathrm{TB}}(X_{n+1})
\quad\Longleftrightarrow\quad M_{n+1}(y)\ge u
\quad\Longleftrightarrow\quad P_{n+1}^{k,\mathrm{ref}}(y)\ge u\ \ \text{for all }k\in[K].
\end{align*}
If $u_{\alpha}^{\mathrm{TB}}=-\infty$, we trivially have $C_{\mathrm{TB}}(X_{n+1})=\mathcal{Y}=\mathbb{R}$, so assume $u_{\alpha}^{\mathrm{TB}}$ is finite for the remainder of the proof.
Now fix $k\in[K]$ and set
\(
r \coloneqq |y-\widehat{\mu}_k(X_{n+1})|.
\)
The condition $P_{n+1}^{k,\mathrm{ref}}(y)\ge u$ is equivalent to
\begin{align*}
1+\sum_{j\in I_{\mathrm{ref}}}\mathds{1}\{R_j^k\ge r\}
\;\ge\; u(n_1+1).
\end{align*}
Since the left-hand side is integer-valued, this is equivalent to
\begin{align*}
\sum_{j\in I_{\mathrm{ref}}}\mathds{1}\{R_j^k\ge r\}
\;\ge\; \left\lceil u(n_1+1)\right\rceil-1.
\end{align*}
Let
\begin{align*}
m \coloneqq \left\lceil u(n_1+1)\right\rceil-1,
\qquad\text{so that}\qquad
\ell \coloneqq n_1-m+1 = n_1+2-\left\lceil u(n_1+1)\right\rceil,
\end{align*}
which matches the definition in the proposition. The inequality above therefore says that at least $m$ reference residuals are
\(\ge r\), i.e., at most \(n_1-m=\ell-1\) reference residuals are \(<r\). Equivalently, \(r\) is no larger than the \(\ell\)-th smallest reference residual:
\begin{align*}
\sum_{j\in I_{\mathrm{ref}}}\mathds{1}\{R_j^k\ge r\}\ge m
\quad\Longleftrightarrow\quad r \le R_{(\ell)}^k,
\end{align*}
with the convention \(R_{(n_1+1)}^k=\infty\) covering the case \(m=0\). Thus, $y\in C_{\mathrm{TB}}(X_{n+1})$ is equivalent to the following condition holding for all $k\in[K]$:
\begin{align*}
P_{n+1}^{k,\mathrm{ref}}(y)\ge u
&\quad\Longleftrightarrow\quad
|y-\widehat{\mu}_k(X_{n+1})|\le R_{(\ell)}^k
\\
&\quad\Longleftrightarrow\quad
y\in
\Bigl[
\widehat{\mu}_k(X_{n+1})-R_{(\ell)}^k,\
\widehat{\mu}_k(X_{n+1})+R_{(\ell)}^k
\Bigr].
\end{align*}
Intersecting over $k=1,\ldots,K$ yields
\begin{align*}
C_{\mathrm{TB}}(X_{n+1})
=
\bigcap_{k=1}^K
\Bigl[
\widehat{\mu}_k(X_{n+1})-R_{(\ell)}^k,\ 
\widehat{\mu}_k(X_{n+1})+R_{(\ell)}^k
\Bigr],
\end{align*}
as claimed.

\section{Proofs of the results in the supplementary material}
\label{sec: proofs of appendices}

This section collects the proofs of the results stated in the supplementary material.

\subsection{Finite-sample limitation of Monte Carlo MaxT calibration}
\label{sec: proof of failure of size control}

Consider the simplified setting $K=1$. For notational convenience, write $T_b = T_b^k$ for $b=0,1,\ldots,2B$ and $k = 1$. In this case, the critical value $\tilde{u}_\alpha$ in \eqref{eq: estimated u alpha} reduces to
\begin{align*}
	\tilde{u}_\alpha
	=
	\sup \Biggl\{
	u \in (0,1)
	:
	\frac{1}{B}
	\sum_{b=1}^{B}
	\mathds{1}
	\Bigl(
	T_{B+b}
	>
	\mathrm{Quantile}_{1-u}
	\bigl\{T_{0}, T_{1}, \ldots, T_{B}\bigr\}
	\Bigr)
	\leq \alpha
	\Biggr\}.
\end{align*}
The resulting MaxT test~\eqref{Eq. max-type test} becomes equivalent to rejecting the null hypothesis when
\begin{align}
\label{Eq: single test}
	T_{0}
	>
	\mathrm{Quantile}_{1-\tilde{u}_\alpha}
	\bigl\{T_{0}, T_{1}, \ldots, T_{B}\bigr\}.
\end{align}
The same intuition applies when the coordinate-wise statistics are perfectly rank-aligned, so that the coordinate-wise rejection events coincide.

\begin{proposition}
\label{prop:failure of size control without bisection}
	Suppose that $T_{0}, T_{1}, \ldots, T_{2B}$ are exchangeable and take distinct
	values with probability one. Fix $\alpha \in (0,1)$. Then the type~I error
	probability of the test in \eqref{Eq: single test} satisfies
	\begin{align*}
	\mP\!\left(
	T_{0}
	>
	\mathrm{Quantile}_{1-\tilde{u}_\alpha}
	\{T_{0},\ldots,T_{B}\}
	\right)
	=
	\frac{\lfloor B\alpha \rfloor + 1}{B + 1}.
	\end{align*}
\end{proposition}

\begin{proof}

Let $S \coloneqq \{T_0,\ldots,T_B\}$ denote the calibration set and $H \coloneqq \{T_{B+1},\ldots,T_{2B}\}$ the holdout set. Then
\begin{align*}
    \frac{1}{B} \sum_{b=1}^B \mathds{1}\{T_{B+b} > \mathrm{Quantile}_{1-u}(S)\} \leq \alpha \   \Longleftrightarrow \ \sum_{b=1}^B \mathds{1}\{T_{B+b} > \mathrm{Quantile}_{1-u}(S)\} \leq \lfloor B \alpha \rfloor. 
\end{align*}
Set $k = \lfloor B \alpha \rfloor$. The event above is equivalent to stating that at most $k$ holdouts exceed the $(1-u)$-quantile of $S$. This is the same as stating that at least $B-k$ holdouts are less than or equal to this quantile. This is equivalent to
\begin{align*}
    H_{(B-k)} \leq \mathrm{Quantile}_{1-u}(S),
\end{align*}
where $H_{(B-k)}$ denotes the $(B-k)$-th order statistic of the holdout set $H$. The definition of $\tilde{u}_\alpha$ in \eqref{eq: estimated u alpha} (with $K=1$ and $w_k=1$) is thus
\begin{align*}
    \tilde{u}_\alpha = \sup\bigl\{u \in (0,1) : H_{(B-k)} \leq \mathrm{Quantile}_{1-u}(S)\bigr\}.
\end{align*}
Let $S_{(1)} \leq S_{(2)} \leq \ldots \leq S_{(B+1)}$ denote the order statistics of $S$. By \Cref{Lemma: quantile 4}, we have $\mathrm{Quantile}_{1-u}(S) = S_{(\lceil (1-u)(B+1) \rceil)}$ and define $S_{(0)} = -\infty$ by convention.

We partition the sample space into $B+2$ disjoint events:
\begin{align*}
    & E_1  = \{H_{(B-k)} \leq S_{(1)}\}, \\
    & E_r = \{S_{(r-1)} < H_{(B-k)} \leq S_{(r)}\}, \quad \text{for } r \in \{2,\ldots,B+1\}, \\
    & E_{B+2} = \{S_{(B+1)} < H_{(B-k)}\}.
\end{align*}
We now compute $\tilde{u}_\alpha$ on each event:
\begin{itemize}
    \item On $E_1$, $H_{(B-k)} \leq S_{(1)}$. For any $u \in (0,1)$, we have $\lceil (1-u)(B+1) \rceil \geq 1$, so $\mathrm{Quantile}_{1-u}(S) = S_{(\lceil (1-u)(B+1) \rceil)} \geq S_{(1)}$. Thus, $H_{(B-k)} \leq \mathrm{Quantile}_{1-u}(S)$ holds for all $u \in (0,1)$. The supremum is $\tilde{u}_\alpha = 1$.
    This gives $\mathrm{Quantile}_{1-\tilde{u}_\alpha}(S) = \mathrm{Quantile}_{0}(S) = S_{(\lceil 0 \cdot (B+1) \rceil)} = S_{(0)} = -\infty$.
    
    \item On $E_r$ for $r \in \{2,\ldots,B+1\}$, we have $S_{(r-1)} < H_{(B-k)} \leq S_{(r)}$. The condition $H_{(B-k)} \leq \mathrm{Quantile}_{1-u}(S) = S_{(\lceil (1-u)(B+1) \rceil)}$ holds if and only if $r \leq \lceil (1-u)(B+1) \rceil$. This inequality is equivalent to $r-1 < (1-u)(B+1)$, or $u < 1 - \frac{r-1}{B+1}$.
    Therefore, the supremum is $\tilde{u}_\alpha = 1 - \frac{r-1}{B+1} = \frac{B+2 - r}{B+1}$.
    This gives $\mathrm{Quantile}_{1-\tilde{u}_\alpha}(S) = \mathrm{Quantile}_{\frac{r-1}{B+1}}(S) = S_{(\lceil \frac{r-1}{B+1}(B+1) \rceil)} = S_{(r-1)}$.
    
    \item On $E_{B+2}$, we have $S_{(B+1)} < H_{(B-k)}$. Since $\mathrm{Quantile}_{1-u}(S) \leq S_{(B+1)}$ for all $u \in (0,1)$, the condition $H_{(B-k)} \leq \mathrm{Quantile}_{1-u}(S)$ never holds. The set of valid $u$ is empty. By convention for this type of test, we define the supremum over an empty set (as a subset of $[0,1]$) to be $0$. Thus, $\tilde{u}_\alpha = 0$.
    This gives $\mathrm{Quantile}_{1-\tilde{u}_\alpha}(S) = \mathrm{Quantile}_{1}(S) = S_{(\lceil 1 \cdot (B+1) \rceil)} = S_{(B+1)}$.
\end{itemize}
Using these calculations, we can express the type I error probability as
\begin{align*}
    \mP\bigl(T_0 > \mathrm{Quantile}_{1-\tilde{u}_\alpha}(S)\bigr) & = \sum_{r=1}^{B+2} \mP\bigl(\{T_0 > \mathrm{Quantile}_{1-\tilde{u}_\alpha}(S)\} \cap E_r \bigr) \\ 
    & = \sum_{r=1}^{B+2} \mP\bigl(T_0 > S_{(r-1)} \cap E_r \bigr) \\
    & = \sum_{r=1}^{B+2} \mP\bigl(T_0 > S_{(r-1)} \mid E_r \bigr) \mP(E_r).
\end{align*}
Let $A_r = \{T_0 > S_{(r-1)}\}$. We claim that $A_r$ and $E_r$ are independent; the proof is deferred to \Cref{lem:Ar_Er_independence}.

Using this independence, we have
\begin{align*}
    \mP\bigl(T_0 > \mathrm{Quantile}_{1-\tilde{u}_\alpha}(S)\bigr) & = \sum_{r=1}^{B+2} \mP\bigl(T_0 > S_{(r-1)}\bigr) \mP(E_r).
\end{align*}
As shown above, we have the identity
\begin{align*}
    \mP\bigl(T_0 > S_{(r-1)}\bigr) = \frac{B+2-r}{B+1}.
\end{align*}
Substituting this in, the type I error is
\begin{align*}
    \frac{1}{B+1} \sum_{r=1}^{B+2} (B+2 - r) \mP(E_r).
\end{align*}
Note that $E_r$ occurs if and only if $V \coloneqq \sum_{b=0}^B \mathds{1}\{T_{b} < H_{(B-k)}\} = r-1$ under the distinctness assumption. By re-indexing the sum using $v = r-1$, we have
\begin{align*}
    \frac{1}{B+1} \sum_{r=1}^{B+2} (B+2 - r) \mP(E_r) & = \frac{1}{B+1} \sum_{v=0}^{B+1} (B+1 - v) \mP(V = v) \\
    & = \frac{1}{B+1} \left( \sum_{v=0}^{B+1} (B+1)\mP(V = v) - \sum_{v=0}^{B+1} v \mP(V = v) \right) \\
    & = \frac{1}{B+1} \left( (B+1) \cdot 1 - \mE[V] \right) = \frac{B+1 - \mE[V]}{B+1}.
\end{align*}
By linearity of expectation, $\mE[V] = \sum_{b=0}^B \mE[\mathds{1}\{T_{b} < H_{(B-k)}\}] = \sum_{b=0}^B \mP(T_b < H_{(B-k)})$. By exchangeability of $T_0, \ldots, T_B$, all these probabilities are equal:
\begin{align*}
    \mE[V] = (B+1) \mP(T_0 < H_{(B-k)}).
\end{align*}
To compute $\mP(T_0 < H_{(B-k)})$, consider the set $U = \{T_0\} \cup H = \{T_0, T_{B+1}, \ldots, T_{2B}\}$. This set has $B+1$ exchangeable elements. Let $R_0$ be the rank of $T_0$ within $U$. By exchangeability, $R_0 \sim \mathrm{Uniform}\{1, \ldots, B+1\}$. The event $T_0 < H_{(B-k)}$ means that $T_0$ is smaller than at least $k+1$ elements of $H$ (namely, $H_{(B-k)}, \ldots, H_{(B)}$). This implies that the rank $R_0$ of $T_0$ in $U$ can be at most $(B+1) - (k+1) = B-k$. Thus, $\mP(T_0 < H_{(B-k)}) = \mP(R_0 \leq B-k) = \frac{B-k}{B+1}$. Substituting this into the expression for $\mE[V]$:
\begin{align*}
    \mE[V] = (B+1) \cdot \frac{B-k}{B+1} = B-k.
\end{align*}
Finally, the type I error is
\begin{align*}
    \mP\bigl(T_0 > \mathrm{Quantile}_{1-\tilde{u}_\alpha}(S)\bigr) = \frac{B+1 - \mE[V]}{B+1} = \frac{B+1 - (B-k)}{B+1} = \frac{k+1}{B+1}.
\end{align*}
Substituting $k = \lfloor B\alpha \rfloor$, the type I error is $\frac{\lfloor B\alpha \rfloor+1}{B+1}$.
\end{proof}

\begin{lemma}\label{lem:Ar_Er_independence}
For each $r\in\{1,\ldots,B+2\}$, the events $A_r\coloneqq\{T_0>S_{(r-1)}\}$ and $E_r$ are independent.
\end{lemma}

\begin{proof}
Let $\mathcal{G}_S=\sigma(S_{(1)},\ldots,S_{(B+1)})$ be the sigma-field generated by the order statistics of $S$, and let $\mathcal{G}_H=\sigma(H_{(1)},\ldots,H_{(B)})$ be the sigma-field generated by the order statistics of $H$. Let $J\in\{1,\ldots,B+1\}$ denote the rank of $T_0$ within $S$. By exchangeability and distinctness, conditional on $\mathcal{G}_S$, $J$ is uniformly distributed on $\{1,\ldots,B+1\}$. Moreover, block-exchangeability (i.e., invariance under permutations within each block $\{0,\ldots,B\}$ and $\{B+1,\ldots,2B\}$) implies that
\begin{align*}
\mP(J=j\mid \mathcal{G}_S,\mathcal{G}_H)=\mP(J=j\mid \mathcal{G}_S)=\frac{1}{B+1},
\qquad j=1,\ldots,B+1.
\end{align*}
Hence $J$ is conditionally independent of $\mathcal{G}_H$ given $\mathcal{G}_S$.

Recall that for $r\in\{1,\ldots,B+2\}$,
\begin{align*}
A_r=\{T_0>S_{(r-1)}\}=\{J\ge r\},
\qquad E_r=\{S_{(r-1)}<H_{(B-k)}\le S_{(r)}\}.
\end{align*}
Using iterated conditioning,
\begin{align*}
\mP(A_r\cap E_r\mid \mathcal{G}_S,\mathcal{G}_H)
&= \mP(A_r\mid \mathcal{G}_S,\mathcal{G}_H)\,\mathds{1}(E_r) \\
&= \mP(J\ge r\mid \mathcal{G}_S,\mathcal{G}_H)\,\mathds{1}(E_r) \\
&= \mP(J\ge r\mid \mathcal{G}_S)\,\mathds{1}(E_r).
\end{align*}
Taking expectations first with respect to $\mathcal{G}_H$ conditional on $\mathcal{G}_S$ yields
\begin{align*}
\mP(A_r\cap E_r\mid \mathcal{G}_S)=\mP(J\ge r\mid \mathcal{G}_S)\,\mP(E_r\mid \mathcal{G}_S).
\end{align*}
Finally, taking expectations over $\mathcal{G}_S$ and using $\mP(J\ge r\mid \mathcal{G}_S)=(B+2-r)/(B+1)$, we obtain
\begin{align*}
\mP(A_r\cap E_r)=\frac{B+2-r}{B+1}\,\mP(E_r)=\mP(A_r)\,\mP(E_r),
\end{align*}
which proves independence.
\end{proof}

\subsection{Proof of \Cref{prop:uniform-asymptotic-adaptation}}
\label{sec: proof of uniform asymptotic adaptation}

We use the notation from \Cref{proof:asymptotic-adaptation}. For each $(B,n)$, define
\begin{align*}
\hat F_{B,n}(t)\coloneqq \frac{1}{B+1}\sum_{b=0}^B \mathds{1}(f_{b,n}\le t),
\qquad
\tilde F_{B,n}(t)\coloneqq \frac{1}{B}\sum_{b=1}^B \mathds{1}(f_{b,n}\le t).
\end{align*}
As before, $\sup_t|\hat F_{B,n}(t)-\tilde F_{B,n}(t)|\le 1/(B+1)$.

\begin{proof}
\parheading{Step 1': Uniform convergence of the empirical CDF.}
Assume the uniform marginal and joint approximation conditions used in the uniform statement:
\begin{align*}
\sup_{P\in\mathcal P}\sup_{t\in\mathbb R}
\bigl|\mP_P(f_{1,n}\le t)-F_P(t)\bigr|\to 0,
\qquad
\sup_{P\in\mathcal P}\sup_{t\in\mathbb R}
\bigl|\mP_P(f_{1,n}\le t,f_{2,n}\le t)-F_P(t)^2\bigr|\to 0,
\end{align*}
where $F_P(t)=\mP_P(V_{\infty,P}\le t)$. These can be deduced from the more general pointwise convergence condition in the statement of \Cref{prop:uniform-asymptotic-adaptation}. Fix $\eta>0$ and write $X_{b,n}(t)=\mathds{1}\{f_{b,n}\le t\}$. For the Monte Carlo CDF,
$\tilde F_{B,n}(t)=\frac1B\sum_{b=1}^B X_{b,n}(t)$, exchangeability yields
\begin{align*}
\Var_P\bigl(\tilde F_{B,n}(t)\bigr)
&=
\frac{1}{B}\Var_P(X_{1,n}(t))
+\frac{B-1}{B}\mathrm{Cov}_P\!\bigl(X_{1,n}(t),X_{2,n}(t)\bigr)\\
&\le
\frac{1}{4B}
+
\bigl|\mathrm{Cov}_P(X_{1,n}(t),X_{2,n}(t))\bigr|,
\end{align*}
where the $1/4$ factor arises from the fact that $X_{1,n}(t)$ is a Bernoulli random variable, whose variance is at most $1/4$. 
Moreover,
\begin{align*}
\bigl|\mathrm{Cov}_P(X_{1,n}(t),X_{2,n}(t))\bigr|
&=
\bigl|\mP_P(f_{1,n}\le t,f_{2,n}\le t)-\mP_P(f_{1,n}\le t)^2\bigr| \\
&\le
\bigl|\mP_P(f_{1,n}\le t,f_{2,n}\le t)-F_P(t)^2\bigr|
+2\bigl|\mP_P(f_{1,n}\le t)-F_P(t)\bigr|.
\end{align*}
Taking $\sup_{P\in\mathcal P}\sup_{t\in\mathbb R}$, the right-hand side tends to $0$ by assumption. Hence
\begin{align*}
\sup_{P\in\mathcal P}\sup_{t\in\mathbb R}\Var_P\bigl(\tilde F_{B,n}(t)\bigr)\to 0.
\end{align*}
Fix $\eta>0$. By Chebyshev's inequality,
\begin{align*}
\sup_{P\in\mathcal P}\sup_{t\in\mathbb R}
\mP_P\bigl(|\tilde F_{B,n}(t)-\mP_P(f_{1,n}\le t)|>\eta\bigr)\to 0.
\end{align*}
Combining with uniform marginal convergence gives
\begin{align*}
\sup_{P\in\mathcal P}\sup_{t\in\mathbb R}
\mP_P\bigl(|\tilde F_{B,n}(t)-F_P(t)|>2\eta\bigr)\to 0.
\end{align*}
Finally, since $\sup_{t \in \mathbb{R}}|\hat F_{B,n}(t)-\tilde F_{B,n}(t)|\le 1/(B+1)$, the same conclusion holds with $\hat F_{B,n}$ in place of $\tilde F_{B,n}$.

\parheading{Step 2': Uniform quantile convergence.}
For the uniform statement, use the assumed uniform margin condition: for each $\varepsilon>0$,
\begin{align*}
\delta(\varepsilon)
\coloneqq
\inf_{P\in\mathcal P}
\min\bigl\{\alpha-F_P(Q^\star_{\alpha,P}-\varepsilon),\
F_P(Q^\star_{\alpha,P}+\varepsilon)-\alpha\bigr\}
>0.
\end{align*}
Then, for any $P\in\mathcal P$,
\begin{align*}
\{\uSB>Q_{\alpha,P}^\star+\varepsilon\}
&\subseteq
\{\hat F_{B,n}(Q_{\alpha,P}^\star+\varepsilon)\le \alpha\}\\
&\subseteq
\Bigl\{
|\hat F_{B,n}(Q_{\alpha,P}^\star+\varepsilon)
-F_P(Q_{\alpha,P}^\star+\varepsilon)|
\ge \delta(\varepsilon)
\Bigr\},
\end{align*}
and similarly
\begin{align*}
\{\uSB<Q_{\alpha,P}^\star-\varepsilon\}
&\subseteq
\{\hat F_{B,n}(Q_{\alpha,P}^\star-\varepsilon)>\alpha\}\\
&\subseteq
\Bigl\{
|\hat F_{B,n}(Q_{\alpha,P}^\star-\varepsilon)
-F_P(Q_{\alpha,P}^\star-\varepsilon)|
\ge \delta(\varepsilon)
\Bigr\}.
\end{align*}
Therefore,
\begin{align*}
\sup_{P\in\mathcal P}\mP_P\bigl(|\uSB-Q_{\alpha,P}^\star|>\varepsilon\bigr)
&\le
\sup_{P\in\mathcal P}\mP_P\Bigl(|\hat F_{B,n}(Q_{\alpha,P}^\star+\varepsilon)-F_P(Q_{\alpha,P}^\star+\varepsilon)|\ge \delta(\varepsilon)\Bigr)\\
&\quad+
\sup_{P\in\mathcal P}\mP_P\Bigl(|\hat F_{B,n}(Q_{\alpha,P}^\star-\varepsilon)-F_P(Q_{\alpha,P}^\star-\varepsilon)|\ge \delta(\varepsilon)\Bigr)\\
&\le
2\sup_{P\in\mathcal P}\sup_{t\in\mathbb R}\mP_P\bigl(|\hat F_{B,n}(t)-F_P(t)|\ge \delta(\varepsilon)\bigr)
\, \rightarrow\, 0
\end{align*}
by Step~1'. This proves uniform consistency of $\uSB$.
\end{proof}

\subsection{Proof of \Cref{prop: adaptive aggregation effective multiplicity}}\label{sec: proof of adaptive aggregation}

Write
\begin{align*}
f_b \coloneqq p_b^{\min} = \min_{m\in[M]} \ple\bigl(f_b^m\bigr),
\qquad b\in[B]_0,
\end{align*}
so that the data-driven procedure is exactly the SB test applied to the statistic sequence $(f_b)_{b\in[B]_0}$, i.e.,
\begin{align*}
p_{\mathrm{SB},\mathcal{F}}
=\frac{1}{B+1}\sum_{b=0}^B \mathds{1}(f_b \le f_0).
\end{align*}
Let $\uSB$ denote the corresponding SB threshold. By Proposition~\ref{prop: single-batch equivalence},
\begin{equation}\label{eq:uSB_for_f}
\uSB
=
\sup\Bigl\{u\in\mathbb{R}:\ \frac{1}{B+1}\sum_{b=0}^B \mathds{1}(f_b \le u)\le \alpha\Bigr\}.
\end{equation}
Since the SB test rejects when $f_0 < \uSB$ as in \eqref{eq:SB_statistic_view}, we have
\begin{align*}
\mP\bigl(p_{\mathrm{SB},\mathcal{F}}>\alpha\bigr)
=
\mP\bigl(f_0 \ge \uSB\bigr).
\end{align*}
Let $A\coloneqq\{N_{\mathrm{eff}}(\varepsilon)\le N\}$ so that $\mP(A^c)\le \eta$ by
\eqref{eq: Neff_highprob}. On the event $A$, by the definition of $N_{\mathrm{eff}}(\varepsilon)$ there exist indices $m_1,\ldots,m_N\in[M]$ such that for all $b\in[B]_0$,
\begin{align*}
\min_{j\in[N]} \ple\bigl(f_b^{m_j}\bigr)
\le (1+\varepsilon)\min_{m\in[M]} \ple\bigl(f_b^{m}\bigr)
=(1+\varepsilon) f_b,
\end{align*}
or equivalently,
\begin{equation}\label{eq:fb_lower_by_subset}
f_b \ge \frac{1}{1+\varepsilon}\, g_b,
\qquad\text{where }\ 
g_b \coloneqq \min_{j\in[N]} \ple\bigl(f_b^{m_j}\bigr).
\end{equation}

Define the SB threshold for the sequence $(g_b)_{b\in[B]_0}$ by
\begin{align*}
\hat{u}_\alpha^{\mathrm{SB},(g)} \coloneqq
\sup\Bigl\{u\in\mathbb{R}:\ \frac{1}{B+1}\sum_{b=0}^B \mathds{1}(g_b \le u)\le \alpha\Bigr\}.
\end{align*}
From~\eqref{eq:fb_lower_by_subset}, on the event $A$ each $f_b$ is bounded below by $g_b/(1+\varepsilon)$. Consequently, for all $b\in[B]_0$ and all $u$,
\begin{align*}
\mathds{1}(f_b \le u)
\le
\mathds{1}\bigl(g_b \le (1+\varepsilon)u\bigr).
\end{align*}
Therefore, the empirical distribution function of $(f_b)_{b\in[B]_0}$ is dominated by that of $(g_b)_{b\in[B]_0}$ after rescaling the argument by the factor $(1+\varepsilon)$. Invoking the characterization of the SB threshold in
\eqref{eq:uSB_for_f}, we obtain the lower bound
\begin{align*}
\uSB
\ge
\frac{1}{1+\varepsilon}\,\hat{u}_\alpha^{\mathrm{SB},(g)}
\qquad\text{on } A.
\end{align*}
Next, by \Cref{thm:single-batch power} and the fact that the Bonferroni threshold $c_{\alpha,N}=\alpha/N$ satisfies the condition of that theorem for the minimum merging function with $N$ coordinates, we conclude that the corresponding SB threshold satisfies
\begin{align*}
\hat{u}_\alpha^{\mathrm{SB},(g)} > \frac{\alpha}{N}.
\end{align*}
Finally,
\begin{align*}
\mP\bigl(p_{\mathrm{SB},\mathcal{F}}>\alpha\bigr)
=\mP(f_0\ge \uSB)
&\le \mP(A^c)+\mP\!\left(f_0 > \frac{\alpha}{(1+\varepsilon)N},\,A\right)\\
&\le \eta + \mP\!\left(f_0 > \frac{\alpha}{(1+\varepsilon)N}\right)\\
&= \eta + \mP\!\left(\min_{m\in[M]} \ple\bigl(f_0^m\bigr) > \frac{\alpha}{(1+\varepsilon)N}\right)\\
&\le \eta + \min_{m\in[M]}\mP\!\left(\ple\bigl(f_0^m\bigr) > \frac{\alpha}{(1+\varepsilon)N}\right),
\end{align*}
where the last inequality uses $\{\min_m X_m>t\}=\cap_m\{X_m>t\}$. This proves the claim.

\subsection{Proof of \Cref{prop: SB-TB equivalence rate}}\label{sec: proof of SB-TB equivalence rate}
Fix $K$ and work conditionally on the observed data $\mathbf X$. Throughout, we assume that, conditional on $\mathbf X$, the transformations $g_1,\ldots,g_{2B}$ are i.i.d.\ uniform draws from $\mathcal{G}$. For any function $F$, we denote its left limit at $t$ by $F(t-)\coloneqq\lim_{s\uparrow t}F(s)$. 

\parheading{Step 1 (Holdout vs.\ within-testing-batch p-values).}
For each $k\in[K]$, define the empirical CDFs
\begin{align*}
\widehat F_{\mathrm{ref}}^k(t)
\coloneqq
\frac{1}{B}\sum_{i=1}^B \mathds{1}(\tilde T_i^k\le t),
\qquad
\widehat F_{\mathrm{test}}^k(t)
\coloneqq
\frac{1}{B+1}\sum_{j=0}^B \mathds{1}(T_j^k\le t).
\end{align*}
For each $b\in[B]_0$,
\begin{align*}
p_{\mathrm{HO}}(T_b^k)
=
\frac{1}{B+1}\Bigl\{1+\sum_{i=1}^B \mathds{1}(\tilde T_i^k\ge T_b^k)\Bigr\}
=
1-\frac{B}{B+1}\widehat F_{\mathrm{ref}}^k(T_b^k-),
\end{align*}
while the SB (within-batch) permutation p-values satisfy
\begin{align*}
p(T_b^k)
=
\frac{1}{B+1}\sum_{j=0}^B \mathds{1}(T_j^k\ge T_b^k)
=
1-\widehat F_{\mathrm{test}}^k(T_b^k-).
\end{align*}
Hence,
\begin{equation}
\label{eq:pval_diff_bound_SB_TB_new}
\max_{b\in[B]_0}
\bigl|p_{\mathrm{HO}}(T_b^k)-p(T_b^k)\bigr|
\le
\sup_{t\in\mathbb R}
\bigl|\widehat F_{\mathrm{ref}}^k(t-)-\widehat F_{\mathrm{test}}^k(t-)\bigr|
+\frac{1}{B+1}.
\end{equation}

\parheading{Step 2 (Uniform DKW control).}
Let
\begin{align*}
\widehat F_{\mathrm{test,MC}}^k(t)
\coloneqq
\frac{1}{B}\sum_{j=1}^B \mathds{1}(T_j^k\le t).
\end{align*}
Then
\begin{align*}
\sup_{t \in \mathbb{R}}
\bigl|\widehat F_{\mathrm{test}}^k(t-)-\widehat F_{\mathrm{test,MC}}^k(t-)\bigr|
\le \frac{1}{B+1}.
\end{align*}
Conditional on $\mathbf X$, both $\widehat F_{\mathrm{ref}}^k$ and $\widehat F_{\mathrm{test,MC}}^k$ are empirical CDFs of independent i.i.d.\ samples from the same randomization distribution. By the Dvoretzky--Kiefer--Wolfowitz (DKW)
inequality applied to $\widehat F_{\mathrm{ref}}^k$ and $\widehat F_{\mathrm{test,MC}}^k$, and a union bound, for every $\eta>0$,
\begin{align*}
\mP\!\left(
\sup_{t \in \mathbb{R}}
\bigl|\widehat F_{\mathrm{ref}}^k(t-)-\widehat F_{\mathrm{test}}^k(t-)\bigr|
>2\eta+\frac{1}{B+1}
\ \middle|\ \mathbf X
\right)
\le 4e^{-2B\eta^2}.
\end{align*}
A union bound over $k\in[K]$ (with $K$ fixed) yields, for every $\eta>0$,
\begin{align*}
\mP\!\left(
\varepsilon_B>2\eta+\frac{2}{B+1}
\right)
\le 4K e^{-2B\eta^2},
\qquad
\varepsilon_B
\coloneqq
\max_{k\in[K]}\max_{b\in[B]_0}
\bigl|p_{\mathrm{HO}}(T_b^k)-p(T_b^k)\bigr|.
\end{align*}
In particular, for any $\delta>0$, there exist constants $c>0$ and $B_0<\infty$
such that for all $B\ge B_0$,
\begin{align*}
\mP(\varepsilon_B>\delta) \le 4K\,e^{-cB\delta^2}.
\end{align*}
\parheading{Step 3 (From p-values to merged values via uniform continuity).}
Define the TB merged values
\begin{align*}
\tilde f_b
\coloneqq f\bigl(p_{\mathrm{HO}}(T_b^1),\ldots,p_{\mathrm{HO}}(T_b^K)\bigr),
\qquad b\in[B]_0,
\end{align*}
and the SB merged values
\begin{align*}
 f_b
\coloneqq f\bigl(p(T_b^1),\ldots,p(T_b^K)\bigr),
\qquad b\in[B]_0.
\end{align*}
Since $f$ is continuous on the compact set $[0,1]^K$, it is uniformly continuous. Thus, for every $\varepsilon>0$ there exists $\delta_{f}(\varepsilon)>0$ such that whenever $\|x-y\|_\infty\le \delta_{f}(\varepsilon)$, one has $|f(x)-f(y)|\le \varepsilon$. Recalling the definition of $\varepsilon_B$ in Step~2, we therefore have the implication
\begin{align*}
\varepsilon_B\le \delta_f(\varepsilon)
\quad\Longrightarrow\quad
\max_{b\in[B]_0}|\tilde f_b-f_b|\le \varepsilon.
\end{align*}
Consequently, with
\(\Delta_B\coloneqq \max_{b\in[B]_0}|\tilde f_b-f_b|\),
\begin{equation}
\label{eq:DeltaB_epsB_uc}
\mP(\Delta_B>\varepsilon)
\;\le\;
\mP\bigl(\varepsilon_B>\delta_f(\varepsilon)\bigr).
\end{equation}
Combining \eqref{eq:DeltaB_epsB_uc} with the exponential tail bound for $\varepsilon_B$ from Step~2 shows that $\mP(\Delta_B>\varepsilon)$ decays exponentially fast in $B$ (with constants depending only on $f$ and $\varepsilon$).

\parheading{Step 4 (From merged values to thresholds).}
Let $m\coloneqq\lfloor (B+1)\alpha\rfloor+1$ and write $\mathrm{os}_m(\cdot)$ for the $m$-th order statistic. Since order statistics are $1$-Lipschitz with respect to the sup-norm,
\begin{align*}
|\uTB-\uSB|
=
\bigl|\mathrm{os}_m(\tilde f_0,\ldots,\tilde f_B)-\mathrm{os}_m(f_0,\ldots,f_B)\bigr|
\le
\Delta_B.
\end{align*}
Therefore, for every $\varepsilon>0$,
\begin{align*}
\mP\bigl(|\uTB-\uSB|>\varepsilon\bigr)
\;\le\;
\mP(\Delta_B>\varepsilon),
\end{align*}
and the right-hand side decays exponentially fast in $B$ by Step~3. This completes the proof.

\subsection{Proof of \Cref{prop:SBvsTB:finiteB_gap_main}}
\label{app:SBvsTB_finiteB_gap}

We start with an explicit construction underlying \Cref{prop:SBvsTB:finiteB_gap_main} in \Cref{subsubsec: setup}, followed by the assumptions in \Cref{subsubsec: assumptions}. The proof of the limiting rejection probabilities is given in \Cref{subsubsec: proof}. Throughout we use the minimum p-value merging rule, which makes the finite-$B$ discretization effect most transparent.

For notational convenience in the comparative analysis across SB and TB, we redefine the TB procedure by treating the transformations for $b=1,\dots,B$ as the reference batch and the subsequent $B$ transformations for $b=0,B+1,\dots,2B$ as the testing batch. Since the $2B$ sampled transformations are independent and identically distributed conditional on the data, this relabeling is distributionally equivalent to the convention used in \Cref{alg:TB}.

\subsubsection{Setup} \label{subsubsec: setup}

Fix integers $B \ge 1$ and $K \ge 1$. For each $n$, let $\mathbf X^{(n)}$ denote the observed data and let $g_0 = \mathrm{id}$ and $g_1,\ldots,g_{2B}$ be Monte Carlo transformations sampled from the randomization mechanism. For each coordinate $k\in[K]$, define
\begin{align*}
T_{b,n}^k \coloneqq T^k\!\bigl(g_b(\mathbf X^{(n)})\bigr),
\qquad b\in[2B]_0.
\end{align*}
\parheading{SB p-values (first batch).}
For $b\in[B]_0$ and $k\in[K]$, define the usual permutation p-values
\begin{equation}
\label{eq:app_psb}
p_{b,n}^k
\coloneqq
\frac{1}{B+1}\sum_{i=0}^{B}\mathds{1}\!\bigl(T_{i,n}^k \ge T_{b,n}^k\bigr).
\end{equation}
Let the merged SB values be
\begin{align*}
\tilde p_{b,n}^{\mathrm{SB}} \coloneqq \min_{k\in[K]} p_{b,n}^k,
\qquad b\in[B]_0,
\end{align*}
and define the SB aggregation p-value (under the minimum merger) as
\begin{align*}
p_{\mathrm{SB},n}
\coloneqq
\frac{1}{B+1}\sum_{b=0}^{B}\mathds{1}\!\left(\tilde p_{b,n}^{\mathrm{SB}}
\le \tilde p_{0,n}^{\mathrm{SB}}\right).
\end{align*}
\parheading{TB holdout p-values (second batch, first batch as reference).}
For $i\in[B]$ and $k\in[K]$, define the holdout p-values by ranking the second batch against the reference set $\{T_{1,n}^k,\ldots,T_{B,n}^k\}$,
\begin{equation}
\label{eq:app_pho}
p_{B+i,n}^{\mathrm{HO},k}
\coloneqq
\frac{1}{B+1}\biggl\{
1+\sum_{\ell=1}^{B}\mathds{1}\!\bigl(T_{\ell,n}^k \ge T_{B+i,n}^k\bigr)
\biggr\}.
\end{equation}
We also define $p_{0,n}^{\mathrm{HO},k}$ by the same formula with $B+i$ replaced by $0$. Note that $p_{0,n}^{\mathrm{HO},k}=p_{0,n}^k$ by construction.

Let the merged TB testing-batch values be
\begin{align*}
\tilde p_{0,n}^{\mathrm{TB}} \coloneqq \min_{k\in[K]} p_{0,n}^{\mathrm{HO},k},
\qquad
\tilde p_{B+i,n}^{\mathrm{TB}} \coloneqq \min_{k\in[K]} p_{B+i,n}^{\mathrm{HO},k},
\quad i\in[B],
\end{align*}
and define the TB aggregation p-value as the rank p-value of the observed merged value among the $B+1$ testing-batch merged values:
\begin{align*}
p_{\mathrm{TB},n}
\coloneqq
\frac{1}{B+1}\left\{
1+\sum_{i=1}^{B}\mathds{1}\!\left(\tilde p_{B+i,n}^{\mathrm{TB}}
\le \tilde p_{0,n}^{\mathrm{TB}}\right)
\right\}.
\end{align*}
\subsubsection{Assumptions} \label{subsubsec: assumptions}

We impose three simplifying conditions.

\begin{enumerate}
\item[(A1)] \label{assump:A1} \textbf{Perfect rank alignment across coordinates.}
For all $k,k'\in[K]$, all $b,b'\in[2B]_0$, and all $n\in\mathbb{N}$,
\begin{align*}
T_{b,n}^k \le T_{b',n}^k
\quad\Longleftrightarrow\quad T_{b,n}^{k'} \le T_{b',n}^{k'}
\qquad\text{almost surely}.
\end{align*}
In particular, this implies that all permutation and holdout p-values are invariant in $k$, and hence the minimum merger reduces to a single coordinate:
\begin{align*}
\tilde p_{b,n}^{\mathrm{SB}} = p_{b,n}^1,\quad b\in[B]_0,
\qquad
\tilde p_{0,n}^{\mathrm{TB}} = p_{0,n}^{\mathrm{HO},1}=p_{0,n}^1,
\qquad
\tilde p_{B+i,n}^{\mathrm{TB}} = p_{B+i,n}^{\mathrm{HO},1},\quad i\in[B].
\end{align*}
\item[(A2)] \label{assump:A2}
\textbf{Exchangeability and conditioning on full order statistics.}
For each fixed $n$ and $k$, conditional on the data used to compute the statistic, the Monte Carlo draws $T_{1,n}^k,\ldots,T_{2B,n}^k$ are i.i.d.\ from the randomization distribution and are distinct with probability one (e.g., after jittering). Let
\begin{align*}
\mathcal G_n
\coloneqq
\sigma\bigl(T_{(1),n}^k,\ldots,T_{(2B),n}^k\bigr)
\end{align*}
denote the $\sigma$-field generated by the full collection of order statistics. Conditional on $\mathcal G_n$, the ranks of $T_{1,n}^k,\ldots,T_{2B,n}^k$ form a uniformly random permutation of $\{1,\ldots,2B\}$.

\item[(A3)] \label{assump:A3}\textbf{Strong-signal regime.} 
Assume that
\begin{equation*}
\mP \Bigl(T_{0,n}^1 > \max_{1\le b\le B} T_{b,n}^1\Bigr)\;\rightarrow\;1
\qquad\text{as }n\to\infty.
\end{equation*}
\end{enumerate}

\subsubsection{Proof of the statement}
\label{subsubsec: proof}

We work under the setup and assumptions described in
\Cref{subsubsec: setup,subsubsec: assumptions}. We show that, for every $\alpha\in(0,1)$,
\begin{align*}
\mP(p_{\mathrm{SB},n}\le \alpha)
\;\rightarrow\;
\mathds{1}\!\left\{\alpha \ge \frac{1}{B+1}\right\},
\end{align*}
and
\begin{align*}
\mP(p_{\mathrm{TB},n}\le \alpha)
\;\rightarrow\;
\mP\!\left(
W \le \lfloor (B+1)\alpha\rfloor - 1
\right),
\qquad n\to\infty ,
\end{align*}
where $W$ is supported on $\{0,\ldots,B\}$ with
\begin{align*}
\mP(W=k)=\frac{\binom{2B-k-1}{B-1}}{\binom{2B}{B}},
\qquad k=0,\ldots,B.
\end{align*}
Define the event
\begin{align*}
E_n \coloneqq \Bigl\{T_{0,n}^1 > \max_{1\le b\le B} T_{b,n}^1\Bigr\}.
\end{align*}
By \hyperref[assump:A3]{\textup{(A3)}}, we have $\mP(E_n)\to 1$. Moreover, by \hyperref[assump:A1]{\textup{(A1)}}, it suffices to restrict attention to the single coordinate $k=1$.

\bigskip
\noindent\textbf{Step 1 (SB behavior on $E_n$).}
On the event $E_n$, we have $T_{b,n}^1<T_{0,n}^1$ for all $b\in[B]$. Therefore, by \eqref{eq:app_psb},
\begin{align*}
p_{0,n}^1=\frac{1}{B+1}.
\end{align*}
For each $b\in[B]$, the corresponding SB permutation p-value satisfies
\begin{align*}
p_{b,n}^1
=
\frac{1}{B+1}\left\{
\mathds{1}(T_{0,n}^1\ge T_{b,n}^1)
+\sum_{\ell=1}^{B}\mathds{1}(T_{\ell,n}^1\ge T_{b,n}^1)
\right\}
\;\ge\;\frac{2}{B+1}
\qquad\text{on }E_n,
\end{align*}
since $\mathds{1}(T_{0,n}^1\ge T_{b,n}^1)=1$ and the summation term is at least $1$
(it includes $\ell=b$). Consequently, on $E_n$, the merged SB value $\tilde p_{0,n}^{\mathrm{SB}}$ is the unique minimum among $\{\tilde p_{b,n}^{\mathrm{SB}}\}_{b=0}^{B}$, and hence
\begin{align*}
p_{\mathrm{SB},n}=\frac{1}{B+1}
\qquad\text{on }E_n.
\end{align*}
It follows that
\begin{align*}
\mP(p_{\mathrm{SB},n}\le \alpha)
=
\mP(E_n)\,\mathds{1}\!\left\{\frac{1}{B+1}\le \alpha\right\}
+
\mP(p_{\mathrm{SB},n}\le \alpha,\,E_n^c),
\end{align*}
and therefore
\begin{align*}
\mP(p_{\mathrm{SB},n}\le \alpha)
\;\rightarrow\;
\mathds{1}\!\left\{\alpha \ge \frac{1}{B+1}\right\}.
\end{align*}
\bigskip
\noindent\textbf{Step 2 (TB behavior on $E_n$).}
On the event $E_n$, we have $\tilde p_{0,n}^{\mathrm{TB}}=p_{0,n}^{\mathrm{HO},1}=1/(B+1)$. Define
\begin{align*}
W_n
\coloneqq
\sum_{i=1}^{B}
\mathds{1}\!\left\{
p_{B+i,n}^{\mathrm{HO},1}=\frac{1}{B+1}
\right\}
=
\sum_{i=1}^{B}
\mathds{1}\!\left\{
T_{B+i,n}^1 > \max_{1\le \ell\le B} T_{\ell,n}^1
\right\}.
\end{align*}
Let $\mathcal G_n$ denote the $\sigma$-field generated by the order statistics $(T_{(1),n}^1,\ldots,T_{(2B),n}^1)$ of $\{T_{1,n}^1,\ldots,T_{2B,n}^1\}$. Conditional on $\mathcal G_n$, the ranks of $T_{1,n}^1,\ldots,T_{2B,n}^1$ form a uniformly random permutation of $\{1,\ldots,2B\}$ by \hyperref[assump:A2]{\textup{(A2)}}. Consequently, the set of ranks occupied by the reference batch $\{T_{1,n}^1,\ldots,T_{B,n}^1\}$ is a uniformly random subset of size $B$ from $\{1,\ldots,2B\}$.

Let
\begin{align*}
R_n
\coloneqq
\max\!\left\{
\text{ranks of } T_{1,n}^1,\ldots,T_{B,n}^1
\right\}.
\end{align*}
Then $R_n\in\{B,B+1,\ldots,2B\}$ and, conditional on $\mathcal G_n$,
\begin{align*}
\mP(R_n=r\mid\mathcal G_n)
=
\frac{\binom{r-1}{B-1}}{\binom{2B}{B}},
\qquad r=B,\ldots,2B.
\end{align*}
Moreover, since no reference statistic can exceed its own maximum rank $R_n$, all statistics with rank exceeding $R_n$ must belong to the testing batch. It therefore follows deterministically that
\begin{align*}
W_n = 2B - R_n.
\end{align*}
Hence, conditional on $\mathcal G_n$, the distribution of $W_n$ is given by
\begin{align*}
\mP(W_n = k \mid \mathcal G_n)
=
\frac{\binom{2B-k-1}{B-1}}{\binom{2B}{B}},
\qquad k=0,1,\ldots,B,
\end{align*}
which does not depend on the realized values of the order statistics. On the event $E_n$, the TB aggregation p-value satisfies
\begin{align*}
p_{\mathrm{TB},n}
=
\frac{1+W_n}{B+1}.
\end{align*}
Writing $m\coloneqq \lfloor (B+1)\alpha\rfloor$, we therefore have
\begin{align*}
\{p_{\mathrm{TB},n}\le \alpha\}
=
\{W_n \le m-1\}
\quad\text{on }E_n.
\end{align*}
By the law of total probability,
\begin{align*}
\mP(p_{\mathrm{TB},n}\le \alpha)
=
\mP(p_{\mathrm{TB},n}\le \alpha,\,E_n)
+
\mP(p_{\mathrm{TB},n}\le \alpha,\,E_n^c).
\end{align*}
The second term is bounded by $\mP(E_n^c)=o(1)$. For the first term, note that
\begin{align*}
\mP(W_n\le m-1,\,E_n)
=
\mP(W_n\le m-1)
-
\mP(W_n\le m-1,\,E_n^c),
\end{align*}
and since $\mP(W_n\le m-1,\,E_n^c)\le \mP(E_n^c)=o(1)$, we obtain
\begin{align*}
\mP(W_n\le m-1,\,E_n)
=
\mP(W\le m-1)+o(1),
\end{align*}
where $W$ follows a negative hypergeometric random variable with probability mass function
\begin{align*}
\mP(W=k)=\frac{\binom{2B-k-1}{B-1}}{\binom{2B}{B}},
\qquad k=0,\ldots,B,
\end{align*}
corresponding to a population of size $2B$ with $B$ successes, and a stopping rule of $r=1$ failure.
Combining the above displays yields
\begin{align*}
\mP(p_{\mathrm{TB},n}\le \alpha)
=
\mP\!\left(W\le \lfloor (B+1)\alpha\rfloor-1\right)
+ o(1),
\end{align*}
which establishes the stated limit.

Finally, note that
\begin{align*}
\mP(W=B)=\frac{1}{\binom{2B}{B}}>0.
\end{align*}
Since $\lfloor (B+1)\alpha\rfloor \le B$ for all $\alpha\in(0,1)$, the limiting TB rejection probability is strictly less than one. This completes the proof.

\subsection{Proof of \Cref{cor: SB-consistency}}\label{sec: proof of SB-consistency}

We begin by establishing pointwise consistency. Define
\begin{align*}
\bar p_{0,n} \coloneqq \frac{1}{K_n}\sum_{k=1}^{K_n} p_{0,n}^k
\end{align*}
as the average of the individual permutation p-values. By \Cref{thm:single-batch power}, the SB procedure calibrated with any deterministic merging function dominates the corresponding worst-case calibrated test. In particular, since
\begin{align*}
f(u_1,\ldots,u_{K_n}) \coloneqq \frac{2}{K_n}\sum_{k=1}^{K_n} u_k
\end{align*}
is a valid merging function \citep{rueschendorf1982,meng1994posterior}, implying 
\begin{align*}
\mathbb{P}\big(p_{\mathrm{SB,avg}}>\alpha\big)=\mathbb{P}\big(\bar p_{0,n}\geq \hat u^\alpha_{\mathrm{SB,avg}}\big)\leq\mathbb{P}\big(\bar p_{0,n}>\alpha/2\big),
\end{align*}
it suffices to show that $\bar p_{0,n}\to 0$ in probability under the alternative. By Markov's inequality and the layer-cake representation,
\begin{align*}
\mP\!\left(\bar p_{0,n} \ge \alpha/2\right)
&=
\mP\!\left(\frac{1}{K_n}\sum_{k=1}^{K_n} p_{0,n}^k \ge \alpha/2\right) \\
&\le
\frac{2}{\alpha}\,\mE[p_{0,n}^1]
=
\frac{2}{\alpha}\int_0^1 \mP\!\left(p_{0,n}^1>t\right)\,dt,
\end{align*}
where the last equality follows from the identity
\(
\mE[X]=\int_0^1 \mP(X>t)\,dt
\)
for nonnegative random variables bounded by~$1$. By assumption,
\(
\mP(p_{0,n}^1>t)\to 0
\)
for every \(t\in(0,1)\), and since \(0\le p_{0,n}^1\le 1\), the dominated convergence theorem yields
\begin{align*}
\mP\!\left(\bar p_{0,n} \ge \alpha/2\right)\;\rightarrow\;0.
\end{align*}
This establishes pointwise consistency of the SB average aggregation test.

\medskip

We now turn to uniform consistency. Assume that the individual permutation p-values are uniformly consistent over
\(\mathcal P\), i.e.,
\begin{align*}
\sup_{P\in\mathcal P}\mP_P(p_{0,n}^1>t)\;\rightarrow\;0
\qquad\text{for every }t\in(0,1).
\end{align*}
Fix $\alpha\in(0,1)$. Applying Markov's inequality and the layer-cake representation uniformly over $P\in\mathcal P$ gives
\begin{align*}
\sup_{P\in\mathcal P}\mP_P\!\left(\bar p_{0,n}>\alpha/2\right)
\le
\frac{2}{\alpha}
\int_0^1
\sup_{P\in\mathcal P}\mP_P(p_{0,n}^1>t)\,dt.
\end{align*}
The integrand converges pointwise to zero and is uniformly bounded by~$1$. Therefore, the dominated convergence theorem implies
\begin{align*}
\sup_{P\in\mathcal P}
\mP_P\!\left(\bar p_{0,n}>\alpha/2\right)
\;\rightarrow\;0,
\end{align*}
which establishes uniform consistency over~$\mathcal P$ and completes the proof.

\subsection{Proof of \Cref{Lemma: sup equivalence}} \label{sec: proof of sup equivalence}

For simplicity, write 
\begin{align*}
	A \coloneqq \sup\biggl\{t : \frac{1}{n} \sum_{i=1}^n \mathds{1}(Z_i \leq t) \leq \alpha  \biggr\} \quad \text{and} \quad B \coloneqq \sup\biggl\{t : \frac{1}{n} \sum_{i=1}^n \mathds{1}(Z_i < t) \leq \alpha  \biggr\}.
\end{align*}
Since $F_n^{-}(t) \coloneqq \frac{1}{n} \sum_{i=1}^n \mathds{1}(Z_i < t) \leq F_n(t) \coloneqq \frac{1}{n} \sum_{i=1}^n \mathds{1}(Z_i \leq t)$, it is clear that $\{t: F_n(t) \leq \alpha\} \subseteq \{t: F_n^{-}(t) \leq \alpha\}$ and thereby $A \leq B$. 
	
For the reverse direction, i.e., $A \geq B$, assume that $A < B$ for contradiction. Since $B$ is the supremum, there is a sequence $s_k \nearrow B$ with $F_n^{-}(s_k) \leq \alpha$ for all $k$. Choose $k$ sufficiently large so that $s_k > A$ and $s_k$ is not equal to any of $Z_1,\ldots,Z_n$. Then we have 
\begin{align*}
	F_n(s_k) = F_n^{-}(s_k) \leq \alpha. 
\end{align*}
Therefore $F_n(s_k) \leq \alpha$ contradicts the definition of $A$ since $s_k > A$. This contradiction shows that $A \geq B$ and hence $A = B$.

Finally, let $k = \lfloor n\alpha \rfloor + 1$. For any $t < Z_{(k)}$, at most $k-1$ observations are less than or equal to $t$, so $\frac{1}{n} \sum_{i=1}^n \mathds{1}(Z_i \leq t) \leq \frac{k-1}{n} \leq \alpha$. This implies that $t$ is in the set whose supremum defines $A$. Since this holds for all $t < Z_{(k)}$, we have $A \geq Z_{(k)}$. Conversely, for any $t \geq Z_{(k)}$, at least $k$ observations are less than or equal to $t$, so $\frac{1}{n} \sum_{i=1}^n \mathds{1}(Z_i \leq t) \geq \frac{k}{n} > \alpha$. This implies that $t$ is not in the set whose supremum defines $A$. Therefore, we must have $A \leq Z_{(k)}$. Combining both directions gives $A = Z_{(k)}$.

\subsection{Proof of \Cref{Lemma: generalized superuniform lemma}} \label{sec: proof of generalized superuniform lemma}
Let $\mathcal{B}(\overline{\mathbb R})$ denote the Borel $\sigma$-algebra on $\overline{\mathbb R} =\mathbb R\cup\{-\infty,+\infty\}$. Define a finite measure $\nu$ on $(\overline{\mathbb R},\mathcal{B}(\overline{\mathbb R}))$ as the pushforward of $w\,d\mu$ by $t$:
\begin{align*}
\nu(B)
\coloneqq
\int_\Omega w(\omega)\,\mathds{1}\{t(\omega)\in B\}\,d\mu(\omega),
\qquad B\in\mathcal{B}(\overline{\mathbb R}).
\end{align*}
For any $x\in\overline{\mathbb R}$,
\begin{align*}
\nu([x,\infty])
=
\int_\Omega w(\omega)\,\mathds{1}\{t(\omega)\ge x\}\,d\mu(\omega).
\end{align*}
Hence,
\begin{align*}
\int_\Omega w(\omega)\,
\mathds{1}\!\left\{
\int_\Omega w(\omega')\,\mathds{1}\{t(\omega')\ge t(\omega)\}\,d\mu(\omega')
\le \alpha
\right\}
\,d\mu(\omega)
=
\nu\!\left(\{x\in\overline{\mathbb R}:\nu([x,\infty])\le \alpha\}\right).
\end{align*}
Define $G(x)\coloneqq\nu([x,\infty])$. Then $G$ is non-increasing, since $[y,\infty]\subseteq [x,\infty]$ whenever $x<y$, which implies $G(y)\le G(x)$. Let
\begin{align*}
S\coloneqq\{x\in\overline{\mathbb R}: G(x)\le \alpha\},
\qquad x_\alpha\coloneqq\inf S,
\end{align*}
with the convention $\inf\emptyset=+\infty$. Because $G$ is non-increasing, the set $S$ is upward closed: if $x\in S$ and $y\ge x$, then $G(y)\le G(x)\le\alpha$, hence $y\in S$. Consequently, $S$ is of the form $[x_\alpha,\infty]$ or $(x_\alpha,\infty]$.

\medskip If $S=\emptyset$, then $\nu(S)=0\le\alpha$, and the claim holds. Hence assume $S\neq\emptyset$.

\medskip\noindent
\textbf{Case 1: $x_\alpha\in S$.}
In this case, $S\subseteq [x_\alpha,\infty]$, and therefore
\begin{align*}
\nu(S)
\le
\nu([x_\alpha,\infty])
=
G(x_\alpha)
\le
\alpha.
\end{align*}
\medskip\noindent
\textbf{Case 2: $x_\alpha\notin S$.}
Then $S\subseteq (x_\alpha,\infty]$. By the definition of $x_\alpha=\inf S$ and the assumption $S\neq\emptyset$, there exists a sequence $(x_m)_{m\ge1}\subset S$ such that $x_m\downarrow x_\alpha$. By monotonicity of $G$,
\begin{align*}
\nu([x_m,\infty]) = G(x_m)\le \alpha
\qquad\text{for all } m\ge1.
\end{align*}
Moreover, the sets $[x_m,\infty]$ form an increasing sequence and satisfy
\begin{align*}
\bigcup_{m\ge1}[x_m,\infty] = (x_\alpha,\infty].
\end{align*}
By continuity of measures for increasing sequences,
\begin{align*}
\nu((x_\alpha,\infty])
=
\lim_{m\to\infty}\nu([x_m,\infty])
\le \alpha.
\end{align*}
Since $S\subseteq (x_\alpha,\infty]$, we conclude that $\nu(S)\le\alpha$.

\medskip

In both cases, $\nu(S)\le \alpha$, which proves the claim.

\subsection{Proof of \Cref{prop: u_alpha expression}} \label{sec: proof of u_alpha expression}

Observe that the maximum of events can be reduced to a union of events as 
\begin{align*}
	& \max_{k \in [K]} \Bigl\{ T_{B+b}^k - \mathrm{Quantile}_{1-u}\{T_{0}^k,T_{1}^k,\ldots,T_{B}^k\} \Bigr\} > 0  \\
	\Longleftrightarrow \ & \bigcup_{k \in [K]} \Bigl\{ T_{B+b}^k > \mathrm{Quantile}_{1-u}\{T_{0}^k,T_{1}^k,\ldots,T_{B}^k\} \Bigr\}.
\end{align*}
Let $T_{(j)}^k$ be the $j$-th order statistic of $\{T_0^k,T_1^k,\ldots,T_B^k\}$ and then \Cref{Lemma: quantile 4} yields that
\begin{align*}
	\mathrm{Quantile}_{1-u}\{T_{0}^k,T_{1}^k,\ldots,T_{B}^k\} = T_{(\lceil (1-u)(B+1) \rceil)}^k.
\end{align*}
Therefore, we have
\begin{align*}
	T_{B+b}^k > T_{(\lceil (1-u)(B+1) \rceil)}^k \ \Longleftrightarrow \ \lceil (1-u)(B+1) \rceil \leq \sum_{j=0}^{B} \mathds{1}(T_{B+b}^k > T_{j}^k).
\end{align*}
Since $\lceil x \rceil \leq m$ is equivalent to $x \leq m$ for integers $m$, we have
\begin{align*}
	\lceil (1-u)(B+1) \rceil \leq \sum_{j=0}^{B} \mathds{1}(T_{B+b}^k > T_{j}^k) \ \Longleftrightarrow \ (1-u)(B+1) \leq \sum_{j=0}^{B} \mathds{1}(T_{B+b}^k > T_{j}^k),
\end{align*} 
which can be rearranged to
\begin{align*}
	u_{kb} \coloneqq \frac{B+1 - \sum_{j=0}^{B} \mathds{1}(T_{B+b}^k > T_{j}^k)}{(B+1)} = \frac{1}{(B+1)} \sum_{j=0}^B \mathds{1}(T_{B+b}^k \leq T_j^k) \leq u. 
\end{align*}
Recalling $u_b = \min_{k \in [K]} u_{kb}$, we have 
\begin{align*}
	\max_{k \in [K]} \Bigl\{ T_{B+b}^k - \mathrm{Quantile}_{1-u}\{T_{0}^k,T_{1}^k,\ldots,T_{B}^k\} \Bigr\} > 0  \ \Longleftrightarrow \ u_b \leq u.
\end{align*} 
Therefore, the definition of $\tilde{u}_\alpha$ in \eqref{eq: estimated u alpha} can be rewritten as
\begin{align*}
	\tilde{u}_\alpha & = \sup\biggl\{u > 0 : \frac{1}{B} \sum_{b=1}^B \mathds{1}\bigl(u_b \leq u\bigr) \leq \alpha  \biggr\} = u_{(\lfloor{B\alpha\rfloor} + 1)} 
\end{align*}
where the last equality follows from \Cref{Lemma: sup equivalence}. If $u_b=0$ for all $b\in[B]$, then the supremum is taken over an empty set and $\tilde{u}_\alpha$ is set to zero by convention. It then indeed holds that $\tilde{u}_\alpha=0=u_{(\lfloor{B\alpha\rfloor} + 1)}$. Hence the claim follows.

\subsection{Proof of \Cref{prop: general sequential aggregation}}\label{sec: proof of general sequential aggregation}

We prove the type I error bound $\mP(0\in\bigcup_{j=1}^J A_j)\le\alpha$ by establishing exchangeability of the index $0$ with each $b\in[B]$ via a row-permutation equivariance argument, then derive tightness under almost surely distinct scores.

\medskip
\noindent \textbf{Step 1: Reduction.} The sets $A_1,\ldots,A_J$ are pairwise disjoint, since $A_j\subseteq S_{j-1}=[B]_0\setminus\bigcup_{\ell=1}^{j-1}A_\ell$ for each $j\in[J]$, so
\begin{align*}
\mP\biggl(0\in\bigcup_{j=1}^J A_j\biggr)
=
\sum_{j=1}^J \mP(0\in A_j).
\end{align*}
It therefore suffices to show $\mP(0\in A_j)\le\lfloor(B+1)\alpha_j\rfloor/(B+1)$ for each $j\in[J]$ separately.

\medskip
\noindent \textbf{Step 2: Setup.} For $b\in[B]_0$ and $k\in[K]$, define
\begin{align*}
R_b^k \coloneqq p(T_b^k),
\qquad
R_b \coloneqq (R_b^1,\ldots,R_b^K),
\qquad
R \coloneqq (R_b)_{b\in[B]_0}.
\end{align*}
For any permutation $\pi$ of $[B]_0$, let $\pi R$ denote the matrix obtained by permuting the rows of $R$, i.e., $(\pi R)_b\coloneqq R_{\pi(b)}$. Under the group-invariance hypothesis, the rows $R_0,\ldots,R_B$ are exchangeable, so $\pi R\stackrel{d}{=}R$ for every permutation $\pi$ of $[B]_0$.

\medskip
\noindent \textbf{Step 3: Equivariance of survival sets.} The key property of the stage scores is their row-permutation equivariance: for every permutation $\pi$ of $[B]_0$, every $b\in[B]_0$, and every $\ell\in[J]$,
\begin{align}
\label{eq:gen-score-equivariance}
z_{b,\ell}(\pi R)
=
h_\ell\bigl(((\pi R)_b^k)_{k\in I_\ell}\bigr)
=
h_\ell\bigl((R_{\pi(b)}^k)_{k\in I_\ell}\bigr)
=
z_{\pi(b),\ell}(R),
\end{align}
which follows immediately from the definition \eqref{eq: general seq score} and $(\pi R)_b = R_{\pi(b)}$. We use \eqref{eq:gen-score-equivariance} to establish equivariance of the survival sets $S_\ell$ by induction on $\ell$.

Fix $j\in[J]$. Set $S_0(R)\coloneqq [B]_0$, and define recursively, for $\ell\in[j]$,
\begin{align*}
c_\ell(R)
\coloneqq
\sup\Bigl\{u\in\mathbb{R}:
\frac{1}{B+1}\sum_{b\in S_{\ell-1}(R)}
\mathds{1}\bigl(z_{b,\ell}(R)\le u\bigr)\le \alpha_\ell
\Bigr\},
\end{align*}
and
\begin{align*}
A_\ell(R)
\coloneqq
\bigl\{b\in S_{\ell-1}(R): z_{b,\ell}(R)<c_\ell(R)\bigr\},
\qquad
S_\ell(R)\coloneqq S_{\ell-1}(R)\setminus A_\ell(R).
\end{align*}

We establish by induction on $\ell$ that
\begin{equation}\label{eq:seq-gen-equivariance}
S_\ell(\pi R)=\pi^{-1}\!\bigl(S_\ell(R)\bigr)
\end{equation}
for every permutation $\pi$ of $[B]_0$ and every $\ell\in\{0\}\cup[j]$.
The base case $\ell=0$ holds trivially since $S_0(\pi R)=S_0(R)=[B]_0$.

For the inductive step, assume \eqref{eq:seq-gen-equivariance} holds for $\ell-1$, where $\ell\in[j]$. By \eqref{eq:gen-score-equivariance}, $z_{b,\ell}(\pi R)=z_{\pi(b),\ell}(R)$ for all $b\in[B]_0$. Using the inductive hypothesis $S_{\ell-1}(\pi R)=\pi^{-1}(S_{\ell-1}(R))$, we obtain, for every $u\in\mathbb{R}$,
\begin{align*}
\sum_{b\in S_{\ell-1}(\pi R)}
\mathds{1}\bigl(z_{b,\ell}(\pi R)\le u\bigr)
&=
\sum_{b\in \pi^{-1}(S_{\ell-1}(R))}
\mathds{1}\bigl(z_{\pi(b),\ell}(R)\le u\bigr) \\
&=
\sum_{b\in S_{\ell-1}(R)}
\mathds{1}\bigl(z_{b,\ell}(R)\le u\bigr),
\end{align*}
where the last equality uses that $b\mapsto\pi(b)$ is a bijection from $\pi^{-1}(S_{\ell-1}(R))$ onto $S_{\ell-1}(R)$. Hence $c_\ell(\pi R)=c_\ell(R)$. Consequently,
\begin{align*}
b\in A_\ell(\pi R)
&\iff
b\in S_{\ell-1}(\pi R)
\ \text{and}\
z_{b,\ell}(\pi R)<c_\ell(\pi R) \\
&\iff
\pi(b)\in S_{\ell-1}(R)
\ \text{and}\
z_{\pi(b),\ell}(R)<c_\ell(R) \\
&\iff
\pi(b)\in A_\ell(R) \\
&\iff
b\in \pi^{-1}\bigl(A_\ell(R)\bigr),
\end{align*}
so $A_\ell(\pi R)=\pi^{-1}(A_\ell(R))$, and therefore
\begin{align*}
S_\ell(\pi R)
=
S_{\ell-1}(\pi R)\setminus A_\ell(\pi R)
=
\pi^{-1}\bigl(S_{\ell-1}(R)\bigr)\setminus\pi^{-1}\bigl(A_\ell(R)\bigr)
=
\pi^{-1}\bigl(S_\ell(R)\bigr),
\end{align*}
which completes the induction.

\medskip
\noindent \textbf{Step 4: Type I error bound.} Fix $b\in[B]_0$, and let $\pi$ be any permutation of $[B]_0$ with $\pi(0)=b$. Since $\pi(0)=b$ and the induction gives $A_j(\pi R)=\pi^{-1}(A_j(R))$,
\begin{align*}
b\in A_j(R)
\iff
0\in \pi^{-1}\bigl(A_j(R)\bigr)
\iff
0\in A_j(\pi R).
\end{align*}
Since $\pi R\stackrel{d}{=}R$, we have $\mP(b\in A_j)=\mP(0\in A_j)$ for all $b\in[B]_0$, and therefore
\begin{align*}
\mP(0\in A_j)
=
\frac{1}{B+1}\sum_{b=0}^B \mP(b\in A_j)
=
\mE\left[\frac{|A_j|}{B+1}\right].
\end{align*}
Since $z_{b,j}$ takes finitely many values as $b$ ranges over $S_{j-1}$, there exists $\epsilon>0$ such that $(c_j-\epsilon,\,c_j)$ contains no score $z_{b,j}$ with $b\in S_{j-1}$. Hence
\begin{align*}
\frac{|A_j|}{B+1}
=
\frac{1}{B+1}\sum_{b\in S_{j-1}}
\mathds{1}(z_{b,j}<c_j)
=
\frac{1}{B+1}\sum_{b\in S_{j-1}}
\mathds{1}(z_{b,j}\le c_j-\epsilon)
\le \alpha_j
\qquad\text{almost surely,}
\end{align*}
Since $|A_j|$ is integer-valued, the bound $|A_j|/(B+1)\le\alpha_j$ a.s.\ implies $\mP(0\in A_j)\le\lfloor(B+1)\alpha_j\rfloor/(B+1)$. Summing over $j=1,\ldots,J$ and applying subadditivity of the floor function gives
\begin{align}
\label{eq:seq-floor-bound}
\mP\biggl(0\in \bigcup_{j=1}^J A_j\biggr)
&=
\sum_{j=1}^J \mP(0\in A_j) \notag\\
&\le
\frac{1}{B+1}\sum_{j=1}^J \bigl\lfloor(B+1)\alpha_j\bigr\rfloor \notag\\
&\le
\frac{\bigl\lfloor(B+1)\sum_{j=1}^J\alpha_j\bigr\rfloor}{B+1}
\le
\frac{\lfloor(B+1)\alpha\rfloor}{B+1}
\le \alpha.
\end{align}

\noindent \textbf{Step 5: Tightness under almost surely distinct scores.} Suppose now that for every $j\in[J]$, the stage-$j$ scores
$\{z_{b,j}: b\in S_{j-1}\}$ are almost surely distinct. We claim that
\begin{align*}
\mP\biggl(0\in \bigcup_{j=1}^J A_j\biggr)
=
\frac{1}{B+1}\sum_{j=1}^J \lfloor (B+1)\alpha_j\rfloor.
\end{align*}
Set $q_j \coloneqq \lfloor (B+1)\alpha_j\rfloor$ for $j\in[J]$. It suffices to show
$|A_j|=q_j$ almost surely for each $j\in[J]$. Indeed, if this holds, then by the
exchangeability argument above,
\begin{align*}
\mP(0\in A_j)=\frac{\mE|A_j|}{B+1}=\frac{q_j}{B+1},
\qquad j\in[J],
\end{align*}
and since $A_1,\ldots,A_J$ are disjoint by construction,
\begin{align*}
\mP\biggl(0\in \bigcup_{j=1}^J A_j\biggr)
=
\sum_{j=1}^J \mP(0\in A_j)
=
\frac{1}{B+1}\sum_{j=1}^J q_j.
\end{align*}
We record for later use that
\begin{align}\label{eq:q-sum-bound}
\sum_{\ell=1}^J q_\ell
\le
\Bigl\lfloor (B+1)\sum_{\ell=1}^J \alpha_\ell \Bigr\rfloor
\le
\lfloor (B+1)\alpha\rfloor
\le B.
\end{align}

For $j=1$, we have $S_0=[B]_0$, so $|S_0|=B+1\ge q_1+1$. Let
$z_{(1),1}<\cdots<z_{(B+1),1}$ denote the ordered stage-$1$ scores in $S_0$.
Since the scores are almost surely distinct and $q_1=\lfloor (B+1)\alpha_1\rfloor$,
for every $u\in\mathbb R$,
\begin{align*}
\frac{1}{B+1}\#\{b\in S_0:z_{b,1}\le u\}\le \alpha_1
&\quad\Longleftrightarrow\quad
\#\{b\in S_0:z_{b,1}\le u\}\le q_1\\
&\quad\Longleftrightarrow\quad
u<z_{(q_1+1),1}.
\end{align*}
Hence the feasible set in the definition of $c_1$ is exactly
$(-\infty,z_{(q_1+1),1})$, so $c_1=z_{(q_1+1),1}$ and
$A_1=\{b\in S_0:z_{b,1}<z_{(q_1+1),1}\}$. Therefore $|A_1|=q_1$ almost surely.

For the inductive step, let $j\in\{2,\ldots,J\}$ and assume $|A_\ell|=q_\ell$ almost surely for all $\ell<j$. Then
\begin{align*}
|S_{j-1}|
=
B+1-\sum_{\ell=1}^{j-1}|A_\ell|
=
B+1-\sum_{\ell=1}^{j-1}q_\ell
\ge q_j+1,
\end{align*}
by \eqref{eq:q-sum-bound}. Let
$z_{(1),j}<\cdots<z_{(|S_{j-1}|),j}$ denote the ordered stage-$j$ scores in $S_{j-1}$.
Repeating the argument for $j=1$ with $S_{j-1}$ in place of $S_0$,
\begin{align*}
\frac{1}{B+1}\#\{b\in S_{j-1}:z_{b,j}\le u\}\le \alpha_j
\quad\Longleftrightarrow\quad
u<z_{(q_j+1),j}.
\end{align*}
Hence $c_j=z_{(q_j+1),j}$ and
$A_j=\{b\in S_{j-1}:z_{b,j}<z_{(q_j+1),j}\}$, so $|A_j|=q_j$ almost surely.
This completes the induction and therefore the proof of tightness under almost surely distinct scores.

\subsection{Proof of \Cref{lem:rank1_hit_implies_reject}} \label{sec: proof of rank1_hit_implies_reject}
Recall that
\begin{align*}
\ple(f_b^m)
&=\frac{1}{B+1}\sum_{i=0}^B \mathds{1}\{f_i^m \le f_b^m\} \\
&\in\{1/(B+1),2/(B+1),\ldots,1\},
\qquad b\in[B]_0,\ m\in[M].
\end{align*}
On the event \eqref{eq:rank1_hit_event}, we have $p_0^{\min}=1/(B+1)$. Since $p_b^{\min}$ also takes values in the same grid, it follows that for all $b\in[B]_0$,
\begin{align*}
\{p_b^{\min}\le p_0^{\min}\}
=
\{p_b^{\min}=1/(B+1)\}.
\end{align*}
Therefore,
\begin{align*}
p_{\mathrm{SB},\mathcal{F}}
=
\frac{1}{B+1}\sum_{b=0}^B
\mathds{1}\Bigl(p_b^{\min}=\frac{1}{B+1}\Bigr)
=
\frac{W}{B+1},
\end{align*}
where $W = \bigl|\{b\in[B]_0:\ p_b^{\min}=1/(B+1)\}\bigr|$.

It remains to show that $W \le M$. Fix $m\in[M]$ and define the (possibly empty) set of \emph{strict minimizers}
\begin{align*}
S_m
\coloneqq
\Bigl\{b\in[B]_0:\ f_b^m < f_i^m\ \text{for all } i\neq b\Bigr\}.
\end{align*}
By definition, $S_m$ has cardinality at most one.
Moreover, for the usual permutation p-value based on the ordering of $\{f_i^m\}_{i\in[B]_0}$,
\begin{align*}
\ple(f_b^m)=\frac{1}{B+1}
\quad\Longleftrightarrow\quad
b\in S_m,
\end{align*}
because $\ple(f_b^m)=1/(B+1)$ holds if and only if exactly one value in the multiset $\{f_i^m\}_{i\in[B]_0}$ is $\le f_b^m$, i.e., $f_b^m$ is a strict minimum.
Consequently,
\begin{align*}
\Bigl\{p_b^{\min}=\frac{1}{B+1}\Bigr\}
\ \Longrightarrow\ 
\exists\, m\in[M]\ \text{such that }\ b\in S_m.
\end{align*}
Hence
\begin{align*}
W
\le
\Bigl|\bigcup_{m=1}^M S_m\Bigr|
\le
\sum_{m=1}^M |S_m|
\le M.
\end{align*}
If $M\le \lfloor (B+1)\alpha\rfloor$, then
\begin{align*}
p_{\mathrm{SB},\mathcal{F}}
=
\frac{W}{B+1}
\le
\frac{\lfloor (B+1)\alpha\rfloor}{B+1}
\le
\alpha,
\end{align*}
which concludes the proof.

\subsection{Proof of \Cref{prop:aggregation_strictly_better_than_each_rule}} \label{sec: proof of aggregation strictly better than each rule}
We prove the two claims separately.

\smallskip
\noindent\textbf{Claim 1.} Fix $j\in[J]$. By assumption \eqref{eq:component_rank1_hit}, $\ple\bigl(f_0^{m(j)}\bigr)=\tfrac{1}{B+1}$ almost surely under $P_j$. Since $p_0^{\min}=\min_{m\in[M]}\ple(f_0^m)\le \ple(f_0^{m(j)})$, it follows that $p_0^{\min}=\tfrac{1}{B+1}$ almost surely under $P_j$ as well. Applying \Cref{lem:rank1_hit_implies_reject} with $M\le\lfloor(B+1)\alpha\rfloor$, we conclude that $p_{\mathrm{SB},\mathcal{F}}\le\alpha$ almost surely under $P_j$. Since $j\in[J]$ was arbitrary, the same holds under the mixture $P=\sum_{j}\pi_j P_j$, giving $\mP_{P}(p_{\mathrm{SB},\mathcal{F}}\le\alpha)=1$.

\smallskip
\noindent\textbf{Claim 2.} Fix $m\in[M]$. By assumption, there exist $j\in[J]$ with $\pi_j>0$ and $\varepsilon>0$ such that
\begin{align*}
\mP_{P_j}\!\left(\ple\bigl(f_0^{m}\bigr)\le \alpha\right)\le 1-\varepsilon.
\end{align*}
Therefore,
\begin{align*}
\mP_{P}\!\left(\ple\bigl(f_0^{m}\bigr)\le \alpha\right)
&=
\sum_{\ell=1}^J \pi_\ell\, \mP_{P_\ell}\!\left(\ple\bigl(f_0^{m}\bigr)\le \alpha\right) \\
&\le
\pi_j(1-\varepsilon)+\sum_{\ell\neq j}\pi_\ell
\;=\;1-\pi_j\varepsilon
\;<\;1. \qedhere
\end{align*}

\subsection{Proof of \Cref{lem: SB quasi arithmetic finite population}} \label{sec: proof of SB quasi arithmetic finite population}

Fix $t\in(0,1]$. Let $S\coloneqq\{b\in[B]_0:\bar p_{\phi,b}\le t\}$ and $\ell\coloneqq |S|$. The case $\ell=0$ is trivial, so assume $\ell\ge 1$.

\bigskip 
\noindent \textbf{Upper bound.}
By definition of $\bar p_{\phi,b}$ and monotonicity of $\phi$,
$\bar p_{\phi,b}\le t$ implies $\frac1K\sum_{k=1}^K \phi(p_b^{(k)})\le \phi(t)$. Summing over $b\in S$ yields
\begin{equation}\label{eq:SBphi:upper}
\sum_{b\in S}\sum_{k=1}^K \phi\bigl(p_b^{(k)}\bigr)
\le
\ell K \phi(t).
\end{equation}

\bigskip 
\noindent \textbf{Lower bound.}
Fix $k\in[K]$ and let $p_{(1)}^{(k)}\le\cdots\le p_{(B+1)}^{(k)}$ denote the order statistics of $(p_0^{(k)},\ldots,p_B^{(k)})$. By Lemma~\ref{lemma:super-uniform-property}, for $\alpha=(j-1)/(B+1)$ we have
\begin{align*}
\frac{1}{B+1}\sum_{b=0}^B \mathds{1}\!\left(p_b^{(k)}\le \frac{j-1}{B+1}\right)
\le \frac{j-1}{B+1},
\end{align*}
so at most $j-1$ values are $\le (j-1)/(B+1)$. Hence $p_{(j)}^{(k)}>(j-1)/(B+1)$, and since each $p_b^{(k)}$ lies on the grid $\{1/(B+1),\ldots,1\}$, it follows that
\begin{equation}\label{eq:SBphi:orderstat_lb}
p_{(j)}^{(k)}\ge \frac{j}{B+1},
\qquad j\in[B+1].
\end{equation}
Since $\phi$ is increasing and $|S|=\ell$,
\begin{align*}
\sum_{b\in S}\phi\bigl(p_b^{(k)}\bigr)
\ge
\sum_{j=1}^{\ell}\phi\bigl(p_{(j)}^{(k)}\bigr)
\ge
\sum_{j=1}^{\ell}\phi\left(\frac{j}{B+1}\right),
\end{align*}
where the last inequality uses \eqref{eq:SBphi:orderstat_lb}. Summing over $k=1,\ldots,K$ yields
\begin{equation}\label{eq:SBphi:lower}
\sum_{b\in S}\sum_{k=1}^K \phi\bigl(p_b^{(k)}\bigr)
\ge K\sum_{j=1}^{\ell}\phi\left(\frac{j}{B+1}\right).
\end{equation}

\bigskip 
\noindent \textbf{Combine.}
Combining \eqref{eq:SBphi:upper} and \eqref{eq:SBphi:lower} gives
\begin{align*}
\phi(t)\ge \frac{1}{\ell}\sum_{j=1}^{\ell}\phi\left(\frac{j}{B+1}\right),
\end{align*}
hence $t\ge t_{\phi,B}(\ell)$, by monotonicity of $\phi$ and the definition of $t_{\phi,B}(\ell)$. Therefore $\ell\le \max\{\ell':t_{\phi,B}(\ell')\le t\}$, which proves the deterministic inequality
\begin{align*}
\frac{1}{B+1}\sum_{b=0}^{B}\mathds{1}\!\left(\bar p_{\phi,b}\le t\right)
\le G_{\phi,B}(t).
\end{align*}
For the probabilistic statement, under the group-invariance hypothesis the rows $(p_b^{(1)},\ldots,p_b^{(K)})$, $b\in[B]_0$, are exchangeable. Since $\bar p_{\phi,b}$ is a measurable function of the $b$th row, the vector $(\bar p_{\phi,0},\ldots,\bar p_{\phi,B})$ is exchangeable as well. Thus
\begin{align*}
\mP(\bar p_{\phi,0}\le t)
=
\mE\!\left[\frac{1}{B+1}\sum_{b=0}^{B}\mathds{1}(\bar p_{\phi,b}\le t)\right],
\end{align*}
and taking expectations of the deterministic bound yields $\mP(\bar p_{\phi,0}\le t)\le G_{\phi,B}(t)$.

Finally, we prove the validity. By the deterministic inequality above,
\begin{align*}
\frac{1}{B+1}\sum_{b=0}^B \mathds{1}\!\left(\bar p_{\phi,b}\le \bar p_{\phi,0}\right)
\;\le\; G_{\phi,B}(\bar p_{\phi,0})
\end{align*}
holds almost surely. Under the group-invariance hypothesis, the row-wise statistics $(\bar p_{\phi,0},\ldots,\bar p_{\phi,B})$ are exchangeable, and therefore
\begin{align*}
\mP\!\left(G_{\phi,B}(\bar p_{\phi,0})\le \alpha\right)
\le
\mP\!\left(
\frac{1}{B+1}\sum_{b=0}^B \mathds{1}\!\left(\bar p_{\phi,b}\le \bar p_{\phi,0}\right)
\le \alpha
\right)
\le \alpha,
\end{align*}
where the last inequality follows from \Cref{Lemma: permutation p-value}. This shows that $G_{\phi,B}(\bar p_{\phi,0})$ is super-uniform and hence a valid p-value.

\subsection{Proof of \Cref{lemma: compare cr and crB}} \label{sec: proof of compare cr and crB}

Define
\begin{align*}
g_r(\ell)
\coloneqq
\frac{\ell}{\left(\frac{1}{\ell}\sum_{j=1}^{\ell} j^r\right)^{1/r}},
\qquad \ell\in\mathbb{N},
\end{align*}
so that $c_{r,B}=\max_{1\le \ell\le B+1} g_r(\ell)$.

\medskip 
\noindent \textbf{Step 1: upper bound by $c_r$.}
Since $x\mapsto x^r$ is increasing on $[0,\infty)$,
\begin{align*}
\sum_{j=1}^{\ell} j^r
\;>\;
\int_{0}^{\ell} x^r\,dx
=
\frac{\ell^{r+1}}{r+1},
\qquad \ell\ge 1.
\end{align*}
Therefore
\begin{align*}
\frac{1}{\ell}\sum_{j=1}^{\ell} j^r
>
\frac{\ell^{r}}{r+1}
\quad\Longrightarrow\quad g_r(\ell)
<
(r+1)^{1/r}
=c_r.
\end{align*}
Taking the maximum over $\ell\le B+1$ gives $c_{r,B}<c_r$, proving (i).

\medskip 
\noindent \textbf{Step 2: monotonicity and convergence.}
Monotonicity in $B$ is immediate because the maximization set $\{1,\ldots,B+1\}$ increases with $B$.

To identify the limit, note that
\begin{align*}
\frac{1}{\ell}\sum_{j=1}^{\ell}\left(\frac{j}{\ell}\right)^r
\;\rightarrow\;
\int_0^1 x^r\,dx
=
\frac{1}{r+1}
\qquad (\ell\to\infty),
\end{align*}
i.e., a standard Riemann-sum limit. Hence
\begin{align*}
g_r(\ell)
=
\left(\frac{1}{\ell}\sum_{j=1}^{\ell}\left(\frac{j}{\ell}\right)^r\right)^{-1/r}
\;\rightarrow\;
(r+1)^{1/r}
=c_r.
\end{align*}
Given $\varepsilon>0$, choose $\ell$ large with $g_r(\ell)>c_r-\varepsilon$. For all $B$ such that $B+1\ge \ell$,
$c_{r,B}\ge g_r(\ell)>c_r-\varepsilon$. Together with $c_{r,B}<c_r$, this yields $c_{r,B}\uparrow c_r$.
\end{appendix}
\end{document}